\documentclass[ecta,nameyear,draft]{econsocart}
\RequirePackage[colorlinks,citecolor=blue,linkcolor=blue,urlcolor=blue,pagebackref]{hyperref}
\usepackage{makecell} 
\RequirePackage{amsthm,amsmath,amsfonts,amssymb}
\usepackage{dsfont}
\usepackage{multirow}
\usepackage{nicefrac}
\usepackage{relsize}
\usepackage{comment}
\RequirePackage{graphicx}
\usepackage{pdflscape}
\usepackage{enumitem} % for (roman)
\usepackage{footmisc}
\usepackage[section]{algorithm}
\usepackage{algpseudocode}
\usepackage{array}
\usepackage{colortbl} % for cell shading if needed
\usepackage{booktabs} % for better lines
\usepackage{diagbox}
\startlocaldefs

\theoremstyle{plain}

\newtheorem{theorem}{Theorem}

\newtheorem{lemma}{Lemma}
\newtheorem{corollary}{Corollary}
\newtheorem{assumption}{Assumption}

\theoremstyle{definition}
\newtheorem{definition}{Definition}
\newtheorem{example}{Example}
\newtheorem{remark}{Remark}

\newcommand{\onesmat}{\mathbf{1}}
\newcommand{\R}{{D}}
\newcommand{\bpi}{\boldsymbol{\pi}}
\newcommand{\E}{\text{\textnormal{E}}}
\newcommand{\X}{X}
\newcommand{\W}{\mathds{W}}
\newcommand{\tqp}{{\scriptscriptstyle{\textnormal{QP}}}}
\newcommand{\Q}{\mathds{Q}}
\newcommand{\Pmat}{\mathds{P}}
\newcommand{\tht}{{\scriptscriptstyle{\textrm{IPW}}}}
\newcommand{\twls}{{\scriptscriptstyle{\textnormal{WLS}}}}
\newcommand{\tols}{{\scriptscriptstyle{\textnormal{OLS}}}}
\newcommand{\tobs}{{\scriptscriptstyle{\textnormal{obs}}}}

\newcommand{\0}{\mathbf{0}}
\newcommand{\bpiInv}{\bpi^{-1}}
\newcommand{\thj}{{\scriptscriptstyle{\textnormal{HJ}}}}
\newcommand{\m}{\mathbf{m}}
\newcommand{\I}{I}
\newcommand{\tci}{\scriptscriptstyle{\textnormal{CI}}}
\newcommand{\tmi}{\scriptscriptstyle{\textnormal{MI}}}
\newcommand{\tgr}{{\scriptscriptstyle{\textnormal{GR}}}}
\newcommand{\tqmle}{{\scriptscriptstyle{\textnormal{QMLE}}}}

\newcommand{\tas}{{\scriptscriptstyle{\textnormal{AS}}}}
\newcommand{\tL}{{\scriptscriptstyle{\textnormal{L}}}}
\newcommand{\tNO}{{\scriptscriptstyle{\textnormal{NOH}}}}
\newcommand{\toc}{{\scriptscriptstyle{\textnormal{Opt}}}}
\newcommand{\toci}{{\scriptscriptstyle{\textnormal{Opt-I}}}}

\newcommand{\Om}{\mathbf{\omega}}
\newcommand{\ones}[1]{ 1_{\scriptscriptstyle {#1}}}
\newcommand{\dmat}{\mathbf{\Omega}}
\newcommand{\V}{\text{\textnormal{Var}}}
\newcommand{\cV}{\text{\textnormal{Cov}}}

\newcommand{\z}{\mathbf{z}}
\newcommand{\diag}{\operatorname{diag}}
\newcommand{\bfI}{\text{\textnormal{I}}}
\newcommand{\tnew}{\scriptscriptstyle{\textnormal{NEW}}}
\newcommand{\matp}{\mathbf{p}}
\newcommand{\dtildep}{\widetilde{ \dmat }_{\hspace{-.6mm}{}{/}}{}_{\scriptscriptstyle \hspace{-.6mm}\matp}}
\newcommand{\tgn}{{\scriptscriptstyle{\textnormal{N}}}}
\newcommand{\s}{\mathbf{S}}
\newcommand{\A}{\mathbf{A}}

\newcommand{\Tr}{\mathbf{Tr}}

\newcommand{\Ome}{\mathbf{\Omega}}
\newcommand{\Otildep}{\tilde{ \Ome }_{\hspace{-.6mm}{}{/}}{}_{\scriptscriptstyle \hspace{-.6mm}\matp}}
\usepackage{chngcntr}
\counterwithin{table}{section}
\newcommand{\viii}[1]{\left\lvert\kern-0.25ex\left\lvert\kern-0.25ex\left\lvert #1 
    \right\rvert\kern-0.25ex\right\rvert\kern-0.25ex\right\rvert_2}
\newcommand{\viiii}[1]{\left\lvert\kern-0.25ex\left\lvert\kern-0.25ex\left\lvert #1 
    \right\rvert\kern-0.25ex\right\rvert\kern-0.25ex\right\rvert}

\endlocaldefs

\begin{document}

\begin{frontmatter}

\title{Design-based Estimation Theory for Complex Experiments}
\begin{aug}
% use \particle for den|der|de|van|von (only lc!)
% [add1]{\fnms{}~\snm{}\ead[label=e?]{}}
%
%% e-mail is mandatory for each author
%
%%% initials in fnms (if any) with spaces
%
\author[add1]{\fnms{Haoge}~\snm{Chang}\ead[label=e1]{hc3516@columbia.edu}}
%%%%%%%%%%%%%%%%%%%%%%%%%%%%%%%%%%%%%%%%%%%%%%
%% Addresses                                %%
%%%%%%%%%%%%%%%%%%%%%%%%%%%%%%%%%%%%%%%%%%%%%%
\address[add1]{%
\orgdiv{Department of Economics},
\orgname{Columbia University}}
\end{aug}

%% Put support info here. Reminder: do not thank the handling coeditor anonymously or by name
\begin{funding}
The author expresses special thanks to Joel Middleton for extensive guidance in the development of this research. I thank my advisors, Don Andrews, Xiaohong Chen, and P. M. Aronow, for their guidance and support. I thank Patrick Lopatto and Anna Wilke for reading the paper carefully and providing valuable and extensive feedback. I thank Jason Abaluck, Max Cytrynbaum, Lucas Finamor, Paul Goldsmith-Pinkham, Philip Haile, Zijian He, John Eric Humphries, Bjoern Hoeppner, Yuichi Kitamura, Cyrus Samii, Pedro Sant'anna, Fredrik S\"{a}vje, Michael Sullivan, Ye Wang, Ed Vytlacil and Longqi Yang for helpful advice and discussions. 
\end{funding}
\coeditor{\fnm{[Name} \snm{Surname}; will be inserted later]}

\begin{abstract}
This paper considers the estimation of treatment effects in randomized experiments with complex experimental designs, including cases with interference between units. We develop a design-based estimation theory for general experimental designs. Our theory facilitates the analysis of many design-estimator pairs that researchers commonly employ in practice and provides procedures to consistently estimate asymptotic variance bounds. We propose new classes of estimators with favorable asymptotic properties from a design-based point of view. In addition, we propose a scalar measure of experimental complexity which can be linked to the design-based variance of the estimators. We demonstrate the performance of our estimators using simulated datasets based on an actual network experiment studying the effect of social networks on insurance adoptions.
\end{abstract}

\begin{keyword}
\kwd{Analysis of Randomized Experiments}
\kwd{Design-based Inference}
\kwd{Regression Adjustment}
\end{keyword}

\end{frontmatter}
%%%%%%%%%%%%%%%%%%%%%%%%%%%%%%%%%%%%%%%%%%%%%%%%%%%%%%%%%%%%%%%%%%%%%%%%%
%%%% Main text entry area:
%%%%%%%%%%%%%%%%%%%%%%%%%%%%%%%%%%%%%%%%%%%%%%%%%%%%%%%%%%%%%%%%%%%%%%%%%

\section{Introduction}
Randomized experiments have become a standard tool in  economic research. Traditionally presented as estimating the average effect of a binary treatment, modern experimental designs have been greatly enriched to capture a variety of economically relevant effects, such as time effects (e.g., \cite{athey2022design} and \cite{roth2021efficient}), peer effects (e.g., \cite{sacerdote2014experimental}), social incentives (e.g., \cite{ashraf2018social}), and spillover effects (e.g., \cite{hudgens2008toward}, \cite{miguel2004worms} and \cite{cai2015social}). Many such experimental designs involve nonstandard treatment assignment mechanisms and/or interference of treatment status among experimental units according to spatial/network/time proximity.\footnote{By interference, we mean the exposure of one unit to treatment may include other units' assignments. This typically arises when researchers are interested in some spillover effects, e.g. \cite{hudgens2008toward} and \cite{aronow2017estimating}. } We refer to these experimental designs are as \textit{complex experiments}.\footnote{The term \textit{complex} is borrowed from the survey sampling literature \citep{chaudhuri2005survey}, where it refers to survey designs that depart from common types of random sampling.}  Many researchers analyze experimental data using a regression model with (possibly clustered) robust standard errors. Although such procedures are justifiable for simple experimental designs,\footnote{For example, linear regression models are justifiable in two-arm completely randomized designs (\cite{freedman2008regression}, \cite{lin2013agnostic}).} they can be ad hoc when applied to complex experimental designs. It is not clear to what extent the results rely on the modeling assumptions and how to interpret the results when the regression models are thought to be misspecified.

Design-based statistical theory provides a powerful framework for analyzing complex experiments. In the design-based framework, the randomization of treatment assignment is the sole source of statistical randomness. Estimation and inferential theory are formulated on this randomness alone, without reference to any other stochastic model (e.g., sampling from a superpopulation  and/or random disturbance terms). This framework has important implications for weighting in the estimation of average treatment effects and for the estimation of standard errors. In simple experiments, the design-based framework provides procedures compatible with current empirical practices for analyzing experimental data with regression models. But in more complex settings, the design-based framework can nevertheless be adapted to provide general-purpose estimation strategies that do not rely on regression models for validity.
%Additionally, the design-based framework can be easily adapted to complex experiments in which researchers may find it challenging to analyze the data at hand.

Estimation theory in the design-based setting has been investigated for many designs on a case-by-case basis. Many important insights have been derived from studying particular experimental designs, but a design-based estimation theory that can be applied to general experimental designs has not hitherto been developed. A design-based estimation theory with broad applicability is important for practice, as it provides guidance to empirical researchers using novel experimental designs that deviate from well-analyzed cases or simple experimental designs that deviate from the standard ones due to practical limitations and implementation reasons. Such designs appear frequently in economic research. 

This paper studies design-based estimation theory for general experimental designs. Our results can be applied to standard designs (e.g., completely randomized designs, clustered randomized designs, and pairwise randomized designs) as well as complex designs where analytical results were not previously available. Under mild regularity assumptions, we provide procedures to consistently and efficiently estimate the average effects of interest and procedures to consistently estimate asymptotic variance bounds.\footnote{In the design-based framework, the asymptotic variance is not generally identified.  Starting with \cite{splawa1990application}, the common solution to the issue of unidentified variances has been to estimate a {\it variance bound}, an identified quantity that is provably greater than the variance. The variance bound formula reduces to the standard (cluster) robust standard errors in simple designs. For example, see \cite{lin2013agnostic} and \cite{schochet2021design}.} We also provide a novel scalar measure of experimental complexity which can be linked to the design-based variance of the estimators, enabling researchers to understand the strengths and weaknesses of particular experimental designs. This measure can be used in the designing-stage of the experiment before collecting any outcome data.

\sloppy Building off of recent advances in design-based estimation theory \citep{middleton2018unified,middleton2021unifying}, the paper makes three main contributions. As the first contribution, we extend the theoretical analysis of many standard estimators to a broader class of experimental designs. Specifically, we analyze a family of design-estimator pairs commonly employed by researchers in practice. We define the class of moment estimators and study their properties with general experimental designs. Special cases of these estimators include the inverse-probability weighted (IPW), Hajek, weighted least squares (WLS), and generalized regression estimators.\footnote{Generalized regression estimators have the same form as doubly-robust estimators in the observational setting, as noted by \cite{kang2007demystifying}. } We provide conditions for convergence to probability limits and characterize the asymptotic variances for these estimators.\footnote{Refer to Section \ref{Setup} for the definition of asymptotics in this setting. } We provide procedures for consistent plug-in variance-bound estimation for general designs under a weak moment assumption.

As a second contribution, we offer new estimators that have desirable asymptotic properties and are applicable with general experimental designs. The new estimators increase estimation precision by having smaller design-based asymptotic variances. The new classes of estimators are based on the class of generalized regression estimators. The first class we consider is the class of standard \textit{Quasi-Maximum Likelihood GR estimators} (QMLE-GR). This class follows the classical model-assisted estimation strategy in the survey analysis literature  \citep{sarndal2003model} and it is useful when the researcher has a good approximating model for potential outcomes and covariates. However, in terms of asymptotic variances, this strategy is not guaranteed to be superior to the baseline IPW estimator when the model is misspecified. This problem motivates the second class of estimators, the \textit{no-harm GR estimators} (No-harm-GR). This class of estimators is based on the QMLE estimates but estimates a multiplicative constant in addition. Estimators of this class have an asymptotic variance no worse than that of the baseline IPW estimator. This class of estimators is inspired by the \cite{cohen2020no}'s estimators in a two-arm completely randomized design. The final class is the \textit{optimal GR estimators} (Opt-GR). This class of estimators leads to the greatest reduction of asymptotic variances when compared with estimators using the same class of parametric models for adjustments. This class of estimators can be traced back to \cite{lin2013agnostic}, and \cite{middleton2018unified} studies such estimators for linear models in two-arm experiments. We further consider refinements that combine some of the above approaches, leading to a class of Optimal-Imputed GR estimators (Opt-I GR). We demonstrate the finite sample performances of the proposed estimators using simulated datasets based on an actual network experiment (\cite{cai2015social}). 

As a third contribution, we propose measures of experimental complexity, as a result of the characterization of asymptotic variances. These measures are the largest eigenvalues of the variance-covariance matrices of the inverse probability-weighted treatment assignment indicators. Theoretically, these quantities govern the rate of convergence of moment-type estimators from a design-based point of view. We shall also give a minimax interpretation for such measures: they are the worst-case variance of the IPW estimators when the outcomes are restricted to a unit ball. A collection of such measures provides useful scalar summaries of the relative strengths and weaknesses of an experimental design for measuring different effects of interest. We believe that these measures are useful for researchers to better understand their experimental designs in complex settings and we demonstrate their uses in the simulations.
\subsection{Literature Review}
This paper builds on the insights in \cite{middleton2018unified,middleton2021unifying}, which proposed the use of matrix spectral theory in the design-based framework. This paper inherits and generalizes the insight. Compared with the previous works, this paper 1) provides a rigorous asymptotic analysis for a large class of estimators (moment-type estimators), 2) considers general asymptotic variance bound estimation under weaker conditions, 3) proposes and analyzes new classes of estimators (QMLE-GR, No-harm-GR and Opt-GR), 4) proposes the measures of complexity and 5) specializes the results to network experiments. 

This paper adds to the literature on design-based estimation theory. The survey sampling literature includes a large body of literature on design-based estimation theory (for example, see \cite{sarndal2003model} and \cite{chaudhuri2005survey}). Many results in the literature focus on estimating average/total quantities in complex (but not fully general) survey designs and do not consider interference. We consider the case of estimating the contrast of multiple average quantities under general experimental designs and our setup accommodates interference.

We contribute to the literature on estimation theory for the design-based analysis of experiments \citep{imbens2015causal}. \cite{freedman2008b,freedman2008randomization, freedman2008regression}, \cite{lin2013agnostic}, \cite{bloniarz2016lasso}, \cite{wu2018loop}, \cite{guo2021generalized}, \cite{cohen2020no} and \cite{lei2021regression} study estimation problems in two-arm completely randomized designs. A collection of papers studies estimation and inferential theory with various experimental designs, for example, \cite{middleton2015unbiased, lu2016randomization,li2019rerandomization, roth2021efficient, schochet2021design, negi2021revisiting,athey2022design} and \cite{gao2023causal}. \cite{hudgens2008toward}, \cite{aronow2017estimating}, \cite{hu2022average} and \cite{gao2023causal} study estimation theory in experiments with interference, and \cite{pollmann2020causal} studies spatial experiments. \cite{aronow2013class} considers unbiased difference-type estimation for complex experiments but the paper does not provide any guarantees of variance reduction. Our paper builds on the previous insights and considers the case of general experimental designs, parametric linear and nonlinear models for adjustments, and various strategies for estimating the adjusting models (QMLE, No-harm, Optimal and Optimal-I). 

This paper is also related to the literature on variance characterization and variance bound estimation in design-based settings, for example, \cite{robins1988confidence,mukerjee2018using, aronow2014sharp,pashley2021insights,de2020level,xu2022}. \cite{harshaw2021optimized} studies the problem of optimizing variance bounds. 

Our paper is also related to the literature that analyzes experiment data accounting for variation from both model-based and design-based uncertainties, for example, \cite{bugni2018inference}, \cite{bugni2019inference}, \cite{bai2021inference}, \cite{bai2022inference}, \cite{cytrynbaum2021designing} and \cite{bugni2022inference}. \cite{abadie2020sampling} provides inferential results for the linear regression model that allows for both design-based and sampling-based uncertainty.

The organization of the paper is as follows. Section \ref{Setup} includes model setup, notations and basic assumptions. Section \ref{Section:NetworkExperiments} gives a network experiment example. 
Section \ref{Section:Estimation} defines moment-type estimators and provides estimation and variance bound estimation results. Section \ref{Section:Nonlinear} includes results on various model-assisted estimation strategies (QMLE, No-harm, Optimal and Optimal-I). Section \ref{Section:Simulation} provides simulation results based an actual network experiment (\cite{cai2015social}).

\section{Setup and Notations}\label{Setup}
We consider a Neyman causal model \citep{splawa1990application,imbens2015causal}, where one conducts a randomized experiment with $k$ treatment arms on $n$ experimental units. Each unit $i\in\{1,...,n\}$ is associated with a $k$-vector of nonrandom \textit{potential outcomes}:
\begin{align}
\left(y_{i}(1), y_{i}(2), ..., y_{i}(k)\right)\in\mathbb{R}^k.
\end{align}
 Each unit $i$ is randomly assigned to one of the $k$ treatment arms. We denote the random vector of assignment indicators by 
 \begin{align}
     (\R_{1i}, \R_{2i},...,\R_{ki})\in\{0,1\}^k,
 \end{align}
 where $\R_{ai}=1$ means that the unit $i$ is assigned to the treatment arm $a$ and $\R_{ai}=0$ otherwise. For unit $i$, the observed outcome is generated according to 
 \begin{align}
     Y_i^{\tobs}=\sum_{a=1}^k \R_{ai}y_{i}(a).
 \end{align} One may also observe for each unit $i$ an additional $p$-dimensional row vector of pretreatment covariates $x_i=(x_{1i},x_{2i},...,x_{pi})\in\mathbb{R}^{p}$. In this paper, we assume that the dimension of the covariates does not change with the sample size. We stack the covariate vectors vertically to create a matrix $\mathbf{x}\in\mathbb{R}^{n\times p}$. The observed data for unit $i$ can then be represented as $(Y_i^{\tobs},\R_{1i}, \R_{2i},...,\R_{ki},x_i)\in\mathbb{R}^{1+k+p}$. 
Let $\pi_{ai}=\E[D_{ai}]$ denote the probability of assignment of unit $i$ to treatment arm $a$. We shall hereafter assume that $\pi_{ai}$ is positive for all treatment arm $a$ and unit $i$, unless stated otherwise.

The parameters of interest are contrasts (linear combination) between the group-specific means of the potential outcomes. For example, in a two-arm experiment, the parameter of interest could be the average treatment effect (ATE) between the treated group and control group and it is defined as $\frac{1}{n} \sum_i  \left(y_{i}(2)-y_{i}(1)\right)$. 
 
\subsection{Notation}\label{Section:Notation}
Let $y^1$, $y^2$, ..., $y^k$ represent column $n$-vectors of potential outcomes associated with each of the arms, with the $i$th element of each vector corresponding to the $i$th unit. Thus, $y^a=[y_{i}(a)]_{i=1}^n=\left(y_{1}(a),y_{2}(a),...,y_{n}(a)\right)'\in\mathbb{R}^n$.  We stack these vectors vertically to create a column vector $y$ with length $kn$: \begin{align}
y = \left(
     y^{1'},
     \hdots,
     y^{k'}
\right)'\in\mathbb{R}^{kn},
\end{align}
\indent Let $1_{\scriptscriptstyle n}$ be a column $n$-vector of ones. A $kn \times k$ \textit{intercept matrix} is defined as
{\small
\begin{align}\label{intercept}
\onesmat = & \left[
\begin{matrix}
1_{\scriptscriptstyle n} &  & & 
\\  & 1_{\scriptscriptstyle n} &   & 
\\  &  & \ddots &
\\  &  & & 1_{\scriptscriptstyle n} 
\end{matrix} \right]\in\mathbb{R}^{kn\times k}. 
\end{align}}
Entries left blank are equal to 0.\footnote{Formally, this matrix is defined as $\onesmat=[a_{st}]_{s=1,..,kn}^{t=1,...,k}$, where $a_{st}=1$ if  $(t-1)\leq \frac{s}{n}\leq t$ and 0 otherwise.  } A \textit{k}-vector of the average potential outcomes of the arms can then be written as $\mu_n=\frac{1}{n} \onesmat' y$. From here on, we denote the $k$-vector average potential outcomes as $\mu_n$ and estimators as $\hat{\nu}_n^{\textnormal{(type)}}\in\mathbb{R}^k$. The superscript denotes the type of estimator. For example, an inverse-probability weighted estimator for the average potential outcomes will be denoted as $\hat{\nu}_n^{\tht}$. \\
\indent Next, define an $n \times n$ diagonal matrix with $n$ assignment indicators for treatment arm 1 on the diagonal and 0 otherwise, 
\begin{align}
\R^1 =&
%=\text{diag}\left[ \begin{matrix} R_0 \\ R_1 \end{matrix}\right]=
\diag(\{\R_{1i}\}_{i=1}^n)\hspace{2mm}\in\mathbb{R}^{n\times n},
\end{align}
and define $\R^2$, $\R^3$, $\hdots$, $\R^k$ analogously. Arrange these matrices to create a diagonal $kn \times kn$ matrix
{\small
\begin{align}
\R =&
\left[ \begin{matrix}
\R^{1} \\ & \R^{2} \\ & & \ddots \\& & &  \R^{k}
\end{matrix}\right] \hspace{2mm}\in\mathbb{R}^{kn\times kn}.
\end{align}}
The $kn\times kn$ diagonal matrix of assignment probabilities is written as $\bpi=\E[\R]$. \\
\indent For the purpose of covariate adjustments, we also define the $kn \times (k+p) $ matrix,
{\small
\begin{align}
\X = & \left[
\begin{matrix}
1_{\scriptscriptstyle n} & & & & \mathbf{x}
\\  & 1_{\scriptscriptstyle n} & & & \mathbf{x}
\\  &  & \ddots & & \vdots
\\  &  & & 1_{\scriptscriptstyle n} & \mathbf{x}
\end{matrix} \right]\in\mathbb{R}^{kn\times (k+p) },
\end{align}}
which augments the intercept matrix $\onesmat$ with covariates.

\indent Let $c\in\mathbb{R}^k$ denote an arbitrary column contrast vector such that the parameter of interest can be written as $\frac{1}{n}c'\onesmat' y$. For example, in a two-arm experiment, the ATE is defined by choosing $c=\left(-1 ,1\right)'$ and $\frac{1}{n}c'\onesmat' y=\frac{1}{n} \sum_i  \left(y_{i}(2)-y_{i}(1)\right)$. We write the parameter of interest associated with a contrast vector $c$ as $\mu_n^c= \frac{1}{n}c'\onesmat' y$.
%In a $2^2$ factorial experiment, the average difference between the first two arms is defined by choosing $c=(-1, 1, 0 ,0)'$ and $ \frac{1}{n}c'\onesmat' y= \frac{1}{n} \sum_i  \left(y_{2i}-y_{1i}\right)$. %[[PA]]The average marginal causal effect is defined by choosing $c=(-\frac{1}{2}, -\frac{1}{2}, \frac{1}{2}, \frac{1}{2} )'$ and $\frac{1}{n}c' \onesmat' y = \frac{1}{n} \sum_i \frac{1}{2} \left(y_{3i}+y_{4i} - y_{1i}-y_{2i}\right)$.\\ 

\indent To conclude, in this notation, we say researchers observe the assignment $\R$, outcomes $\R y$,  and a matrix of $p$ pretreatment covariates $\mathbf{x}\in\mathbb{R}^{n\times p}$. In a randomized experiment, $\bpi$ is known,  or can be approximated to arbitrary precision by repeating the randomization procedure and collecting draws \citep{fattorini}. 
%When $\bpi$ does not have a closed-form expression, it can be estimated to arbitrary precision by repeating the original randomization until a target level of precision is achieved\footnote{\cite{fattorini} gives practical guidance.}.\\

\indent We let $\I_{k}$ denote the identity matrix of dimension $k\times k$ and $\0_{k}$ a zero matrix of dimension $k\times k$. $\ones{k}$ denotes a column $k$-vector of 1s. We define $[k]=\{1,...,k\}$ and $[k_1,k_2]=\{k_1,...,k_2\}$, for arbitrary positive integers $k$, $k_1$ and $k_2$. A list of mathematical objects, operators, and quantities used in this paper is included in Appendix \ref{Mathematical Objects, Operators and Quantites}.
\subsection{Asymptotic Schemes}
All results in this paper are asymptotic. We consider a nested sequence of increasing populations, $\{U_n\}_{n}$, where the index $n$ indicates the size of the population under study. Each unit in the population is characterized by its fixed potential outcomes and pretreatment covariates. The potential outcomes and the pretreatment covariates are fixed and the population grows deterministically. $\{U_n\}_{n\geq 1}$ are nested: $U_1\subset U_2\subset ...\subset  U_n...$. Each finite population $U_n$ has an associated experimental design and a realized randomization. Although the populations form a nested sequence, the sequences of realized assignments do not. This asymptotic scheme is widely used in the literature \citep{isaki1982survey,aronow2014sharp,li2017general}. 

We work in the finite-population (design-based) framework \citep{imbens2015causal}, where the potential outcomes $y$ and pretreatment covariates $\mathbf{x}$ are considered fixed parameters, and the only source of randomness in our model is from the random treatment assignments $\{D_{ai}\}_{a\in[k],i\in[n]}$.
\\
\indent In general, the true parameter values are quantities that change with the sample size $n$. We will write (finite) population quantities with a subindex $n$. For example, the average potential outcomes will be denoted as $\mu_n=\frac{1}{n}\onesmat'y$. We call $\hat{\nu}_n$ a consistent estimator for $\nu_n$ if $\hat{\nu}_n-\nu_n=o_p(1)$, where the stochasticity is generated by the experimental design. With an abuse of language, we refer to $\nu_n$ as the probability limit of $\hat{\nu}_n$.

We state two assumptions for data moments, which are needed for the convergence of estimators. Recall that $k$ denotes the number of treatment arms and $p$ denotes  the number of pretreatment covariates.
\begin{assumption}[Bounded fourth moments]\label{A:BoundedFourthMoments}
For all $n$,
\begin{equation}
\frac{1}{n}\sum_{a=1}^k\sum_{i=1}^n y_{ai}^4<C_1, \hspace{5pt} \frac{1}{n}\sum_{s=1}^p\sum_{i=1}^n x_{si}^4<C_1, 
\end{equation}
where $C_1$ is a finite constant.
\end{assumption}
\indent For our analysis of weighted least squares (WLS) estimators below, we require the design matrix to be invertible for large $n$. 
\begin{assumption}[Invertibility]\label{A:Invertibility} there exists an integer $n_0>0$ such that for all $n>n_0$, $\lambda_{\min}\left(\frac{1}{n}\mathbf{x}'\mathbf{x} \right)\geq c_{\ref{A:Invertibility}}$,
where $\lambda_{\min}\left(\frac{1}{n}\mathbf{x}'\mathbf{x} \right)$ is the smallest eigenvalue of the matrix $\frac{1}{n}\mathbf{x}'\mathbf{x}$ and $c_{\ref{A:Invertibility}}$ is a positive constant.
\end{assumption}
%Finally, we assume that the weighted least squares coefficients can be correctly interpreted. That is, they can be viewed as estimators for the average potential outcomes. This assumption is easily achieved by demeaning the columns of the pretreatment covariates $X$.
%\begin{assumption}[Centered columns]\label{A:CenteredColumn}
%For all $n$ and $s\in [p]$,
%$\frac{1}{n}\sum_{i=1}^n x_{si}=0$.
%\end{assumption}

\section{An example: network experiments}\label{Section:NetworkExperiments}
This section provides a concrete example to illustrate the setup described above. We consider the network experiments proposed in \cite{aronow2017estimating}. Components of the experimental design include:
\begin{itemize}
    \item \textbf{A finite population} $U_n$ with units indexed by $i\in [n]$. Each unit has a trait vector $\xi_i\in\Xi_n$ (i.e., network connections) and a pretreatment covariate vector $x_i\in\mathbb{R}^p$. Let $\Xi_n$ denote the set of traits.
    \item \textbf{An experimental design} that randomly selects units into $M$ treatment values. One realization of the assignment vector has the form $Z=(Z_1,...,Z_n)\in\{0,...,M-1\}^n$. The distribution of the random assignment vector $Z$, denoted as $P(Z)$, is known. Let $\Omega_n\subset\{0,...,M-1\}^n $ denote the set of possible random assignment vectors.
    \item \textbf{An exposure mapping} that maps the assignment treatment vectors and a unit-specific trait to an exposure value, $F_n:\Omega_n\times\Xi_n\to\Delta_n$, where $\Delta_n$ denotes the set of possible exposure values. This map is specified by the researcher depending on the research questions at hand. $\Delta_n$ is usually specified to be a finite set.\footnote{In this setup, it is possible that the exposure mappings are misspecified. See \cite{aronow2017estimating}, \cite{savje2021causal}, \cite{savje2021average}, and \cite{leung2022causal} for a discussion of estimation and inferential theories in this context. We will proceed as if the exposure mapping is correctly specified. Some theories on estimation and inference can also be found in  \cite{wang2020design} and \cite{gao2023causal}.
 } 
\end{itemize}
For one experiment, researchers randomly draw an assignment vector $Z$ and observe the scalar outcomes $\{Y_i(Z)\}_{i=1}^n$. Notice that up to this stage the outcome for unit $i$ has depended on the entire assignment vector. Consistent estimation under unrestricted interference is deemed virtually impossible \citep{savje2021average}. One strategy to alleviate the problem is to restrict the interference patterns using exposure mappings, under the assumption that the potential outcomes are correctly specified according to the exposure value:
\begin{assumption}\label{A:EffectiveExposure}
For $i=1,...,n$ and $Z,\tilde{Z}'\in\Omega_n$, $Y_i(Z)=Y_i(Z')$ if  $F_n(Z,\xi_i)=F_n(\tilde{Z},\xi_i)$. $\Delta_n$ is a finite set and does not change with $n$.
\end{assumption}
Assumption \ref{A:EffectiveExposure} implies that the exposures are "effective treatments" as defined in \cite{manski2013identification}. The assumption that $\Delta_n$ equals a finite set $\Delta$ is a typical assumption made in the literature. We note that this experimental setup is very general, and it can be generalized to other settings in which the exposure mappings are not necessarily mediated by a network.\\
\indent Enumerate the element in $\Delta$ as $\{1,...,k\}$. With Assumption \ref{A:EffectiveExposure}, one can write the potential outcomes associated with unit $i$ as $\left(y_{i}(1),...,y_{i}(k)\right)\in\mathbb{R}^k$. The assignment vector associated with unit $i$ can be written as $\left(\R_{1i},...,\R_{ki}\right)$. This maps the problem back to our framework. 
To make the assumptions concrete, we provide an example from Section 9 of \cite{aronow2017estimating}.
\begin{example}\label{Example:Network}
Consider a situation where we observe $n$ units connected in undirected networks. Each unit $i$ is associated with the trait $\theta_i$, which is the $i$th row vector of the unnormalized adjacency matrix. The treatment values are $\{0,1\}$, and the treatment assignment vector is denoted as $Z\in\{0,1\}^n$ with $Z_i$ denoting the treatment assignment of unit $i$. The exposure mapping is assumed to be:
\begin{equation}
F_n(Z,\theta_i)=
\begin{cases}
d_{11} (\text{Direct+Indirect Exposure}): \hspace{4mm} Z_iI(Z'\theta_i>0)=1\\
d_{10} (\text{Isolated Direct Exposure}): \hspace{4mm}Z_iI(Z'\theta_i=0)=1\\
d_{01} (\text{Indirect Exposure}): \hspace{4mm}(1-Z_i)I(Z'\theta_i>0)=1\\
d_{00} (\text{No Exposure}): \hspace{4mm}(1-Z_i)I(Z'\theta_i=0)=1\\
\end{cases}.    
\end{equation}
Units are assigned to treatment and control groups independently with probability $p$.
\end{example}
The theories examined in this paper can be viewed as offering strategies for analyzing and planning such (though not limited to) experiments.

\section{Estimation and Variance Estimation}\label{Section:Estimation}  

 In this section, we study the problem of estimating average potential outcomes, $\mu_n=\frac{1}{n}\onesmat'y\in\mathbb{R}^k$.
 
\indent Section \ref{Section:Consistency} introduces the class of moment-type estimators that allows the simultaneous analysis for many commonly-used estimators. It nests the class of linear estimators introduced in \cite{mukerjee2018using} and \cite{middleton2021unifying}, which includes the IPW, Hajek (HJ), Weighted Least Square (WLS), Completely Imputed (CI), Missing Imputed (MI) and Generalized Regression (GR) estimators. We establish the convergence rate of the moment-type estimators, and give conditions under which the WLS is a consistent estimator of the average potential outcomes.

In Section \ref{Section:Variances}, we study asymptotic variance characterization, bounding, and variance bound estimation for the moment-type estimators. We provide a simple and general asymptotic variance formula. We highlight that the asymptotic variances can be written in the bilinear form $\frac{1}{n}z'\dmat z$, where $z$ reflects the choice of estimators and the matrix $\dmat$, introduced in Definition \ref{FO}, reflects the experimental design. This forms the basis for variance bounding and variance bound estimation. We discuss variance bounding and provide a plug-in variance bound estimator. We briefly discuss inference in Section \ref{section:inference}.

Appendix \ref{network_experiments} formulates the lower-level conditions for the network experiments presented in Section \ref{Section:NetworkExperiments}.

In Section \ref{section:design}, we propose using the largest singular value $\sigma_{\max}\left(\Ome\right)$of the matrix $\Ome$ as an input for experimental designs. The value can be interpreted as the worst-case variance of the IPW estimator when the outcomes are subject to a moment condition.

We remind readers that $a$ is an index for treatment arms, $i$ is an index for experiment units, $k$ is the number of treatment arms, and $n$ is the number of experiment units.

\subsection{Moment-type estimator}\label{Section:Consistency}
We first define moment-type estimators. We will give several examples of moment-type estimators after the definition.
\begin{definition}\label{def:MomentEstimators} A moment-type estimator $\widehat{\nu}_n\in\mathbb{R}^k$ has the form
\begin{equation}
\widehat{\nu}_n=F(\widehat{m}_n^1,\widehat{m}_n^2,...,\widehat{m}_n^{S_1},m_n^{S_1+1},...,m_n^{S_1+S_2}),
\end{equation}
where
\begin{enumerate}[label=(\roman*)]
    \item $F: \mathbb{R}^{S_1+S_2}\to \mathbb{R}^k$ is a known mapping,
    \item $\widehat{m}_n^s=\frac{1}{n}1_{\scriptscriptstyle kn}'\bpiInv\R \phi^s$, $s\in[S_1]$ are scalar estimates of finite population moments,
    \item $m_n^s=\frac{1}{n}1_{\scriptscriptstyle kn}' \phi^s$, $s\in[S_1+1,S_1+S_2]$ are known finite population moments,\footnote{We allow some moments to be nonrandom to accommodate generalized regression estimators. }
    \item $\{\phi^s\}_{s\in [S_1+S_2]}\subset\mathbb{R}^{kn}$ are vectors of nonrandom variable, which may depend on, but are not limited to, potential outcomes, covariates, or treatment assignment probabilities. 
\end{enumerate}
 \end{definition}
 We associate a probability target for each moment-type estimator. 
\begin{definition}
The probability target $\nu_n$ of a moment-type estimator $\widehat{\nu}_n\in \mathbb{R}^k$ is defined as 
\begin{equation}
    \nu_n = F(m_n^1,m_n^2,...,m_n^{S_1},...,m_n^{S_1+1},...,m_n^{S_1+S_2} )\in \mathbb{R}^k,\footnote{We implicitly assume that the probability target is well-defined.}
\end{equation}
where $m_n^s=\E[\widehat{m}_n^s]$ with $s\in[S_1]$.
\end{definition}
 
In short, the probability limit of a moment-type estimator is obtained by replacing estimated moments with their finite population counterparts. For notational simplicity, we hereafter write $\widehat{m}_n=\allowdisplaybreaks(\widehat{m}_n^1,\widehat{m}_n^2,...,m_n^{S_1+S_2})$, $m_n=\allowdisplaybreaks(m_n^1,m_n^2,...,m_n^{S_1+S_2})$. We shall write $\widetilde{m}_n=\allowdisplaybreaks (\widetilde{m}_n^1,\widetilde{m}_n^2,...,\widetilde{m}_n^{S_1},m_n^{S_1+1},...,m_n^{S_1+S_2})$ for some arbitrary $\widetilde{m}_n^1, ...,\widetilde{m}_n^{S_1}$. 

The moment-type estimators nest many commonly-used estimators as special cases.
\begin{example}[Inverse-probability weighted (IPW) estimator]
\begin{equation}
     \widehat{\nu}^{\tht}_n =\frac{1}{n} \onesmat' \bpiInv \R y \in \mathbb{R}^k, 
\end{equation}
and its probability target is $\frac{1}{n}\onesmat' y $. 
\end{example}
\begin{example}[Weighted least square  (WLS) estimators]

	\begin{align}
	 \widehat{\nu}^{\twls}_n  = \begin{bmatrix}
 \I_k \vert \0_{k\times p}
    \end{bmatrix}\widehat{b}_n^{\twls}=\begin{bmatrix}
 \I_k \vert \0_{k\times p}
    \end{bmatrix}\left(\X' \Om \R \X \right)^{+}\X' \Om \R y \in \mathbb{R}^k,
	\end{align}
where $\Om\in\mathbb{R}^{kn\times kn}$ is a diagonal matrix with strictly positive entries and $+$ indicates the Moore-Penrose inverse. Its probability target is $\begin{bmatrix}
 \I_k \vert \0_{k\times p}
    \end{bmatrix}b^{\twls}_n$, where $b^{\twls}_n=(\X'\Om\bpi \X)^{+} \X' \Om\bpi y$ and $b^{\twls}_n$ is a minimizer of the criterion $\left(y-\X\beta\right)'\Om\pi\left(y-\X\beta\right)$.
\end{example}

\begin{example}[Generalized Regression (GR) estimators]\label{GenReg}
\begin{equation}
\widehat{\nu}^{\tgr}_n  = \frac{1}{n}\onesmat'\X\widehat{b}^{\twls}_n + \frac{1}{n}  \onesmat' \bpiInv \R\left(y-X\widehat{b}^{\twls}_n \right) \in \mathbb{R}^k,
\end{equation}
and its probability target is $\frac{1}{n}\onesmat' y $.
\end{example}
All estimators listed above are moment-type estimators. For example, the IPW estimator can be written as:
\begin{equation}
      \widehat{\nu}^{\tht}_n = \left(\frac{1}{n}\sum_{i=1}^n \frac{D_{1i}y_i(1)}{\pi_{1i}}, \frac{1}{n}\sum_{i=1}^n \frac{D_{2i}y_i(2)}{\pi_{2i}},...,\frac{1}{n}\sum_{i=1}^n \frac{D_{ki}y_i(k)}{\pi_{ki}}\right)',
\end{equation}
which are a mapping of unknown moments of potential outcomes.\\
The WLS estimator can be represented as
\begin{equation}
\widehat{\nu}^{\twls}_n=\left(\widehat{\beta}_1,...,\widehat{\beta}_k\right),
\end{equation}
where $\{\widehat{\beta}_{a}\}_{a\in[k]}$ are the realized estimators of the arm-specific coefficients from the weighted regression:
\begin{equation}
    Y_i\sim \sum_{i=1}^n \beta_{a}\R_{ai}+ X_i'\beta, \text{ weights=}\sum_{a=1}^k \R_{ai}\Om_{ai}.
    \end{equation}
The WLS estimator is a mapping of terms such as $\frac{1}{n}\sum_{i=1}^n \R_{ai}x_{ki}y_{ai}\Om_{ai}$ or $\frac{1}{n}\sum_{i=1}^n \R_{ai}x_{ki}x_{si}\Om_{ai}$ for some $k,s\in[p]$.

The GR estimator takes the doubly-robust influence function form and is commonly considered in the literature. Besides the estimated moments, the GR estimator also involves known moments of the form $\frac{1}{n}\sum_{i=1}^n x_{ki}$, $k\in[p]$.

Besides the given examples, the moment-type estimator also includes Hajek, completely-imputed (CI), and missing-imputed (MI) estimators. These estimators are routinely considered in the survey sampling literature, and we include their definitions in Appendix \ref{other_estimators}.\footnote{Missing imputed estimators and completely-imputed estimators are considered in \cite{isaki1982survey}. CI, MI, and GR estimators have a long history in the survey sampling literature (see, for example, \cite{brewer1979class}, \cite{wright1983finite},  \cite{sarndal1984cosmetic}, and \cite{chaudhuri2005survey}). Missing imputed estimators have been recently studied by \cite{guo2021generalized} in nonlinear adjustment problems.  }

\indent We now show convergence of the estimator $\widehat{\nu}_n$ to its probability target $\nu_n$. We note that $\nu_n$ is not necessarily the average potential outcomes, and we will specify conditions under which the WLS is a consistent estimator for the average potential outcomes.\footnote{IPW, GR and Hajek estimators have average potential outcomes as their probability target and hence are consistent estimators. Conditions under which the CI and MI estimators are consistent estimators of the average potential outcomes are given in Section \ref{other_estimators}. } To characterize the rate of convergence, we introduce the first-order design matrix. This matrix is first introduced in \cite{middleton2021unifying} and it is an important conceptual object encoding information about the experimental design. 
\begin{definition}\label{FO}
The first-order design matrix is the variance-covariance matrix of inverse probability-weighted treatment assignments, written as 
\begin{align}\label{dmat}
	\dmat=&\V \left(\ones{kn}'\bpiInv \R \right)\in \mathbb{R}^{kn\times kn }.
\end{align}
\end{definition}

We give an example of the first-order design matrix with two treatment arms $\left(a=1,2\right)$ and two units $\left(i=1,2\right)$:
{\footnotesize
\begin{equation}
    \Ome^{\textrm{example}}=\begin{bmatrix}
        \V\left(\frac{\R_{11}}{\pi_{11}}\right) &\cV\left(\frac{\R_{11}}{\pi_{11}},\frac{\R_{12}}{\pi_{12}}\right)  &\cV\left(\frac{\R_{11}}{\pi_{11}},\frac{\R_{21}}{\pi_{21}}\right) &    \cV\left(\frac{\R_{11}}{\pi_{11}},\frac{\R_{22}}{\pi_{22}}\right) \\
     \cV\left(\frac{\R_{12}}{\pi_{12}},\frac{\R_{11}}{\pi_{11}}\right) &    \V\left(\frac{\R_{12}}{\pi_{12}}\right)  &  \cV\left(\frac{\R_{12}}{\pi_{12}},\frac{\R_{21}}{\pi_{21}}\right)& \cV\left(\frac{\R_{12}}{\pi_{12}},\frac{\R_{22}}{\pi_{22}}\right)  \\
     \cV\left(\frac{\R_{21}}{\pi_{21}},\frac{\R_{11}}{\pi_{11}}\right) & \cV\left(\frac{\R_{21}}{\pi_{21}},\frac{\R_{12}}{\pi_{12}}\right) &   \V\left(\frac{\R_{21}}{\pi_{21}}\right) & \cV\left(\frac{\R_{21}}{\pi_{21}},\frac{\R_{22}}{\pi_{22}}\right)   \\
    \cV\left(\frac{\R_{22}}{\pi_{22}},\frac{\R_{11}}{\pi_{11}}\right) & \cV\left(\frac{\R_{22}}{\pi_{22}},\frac{\R_{12}}{\pi_{12}}\right) &   \cV\left(\frac{\R_{22}}{\pi_{22}},\frac{\R_{21}}{\pi_{21}}\right) &   \V\left(\frac{\R_{22}}{\pi_{22}}\right)
    \end{bmatrix},
\end{equation}}
where the subscripts follow the convention $\R_{ai}$, where $a$ indexes treatment arms and $i$ indexes units, and $\pi_{ai} = \E[\R_{ai}]$. It should be clear that the first-order design matrix depends on both first-order and second-order assignment probabilities.

We use $\sigma_{\max}(A)$ to denote the largest singular value of a matrix $A$. We use $\|v\|_2$ to denote the Euclidean ($l_2$) norm when $v$ is a vector and the Frobenius norm when $v$ is a matrix. We show that $\sigma_{\max}(A)$ upper-bounds the statistical convergence rate for the moment-type estimators under weak regularity conditions.

\begin{assumption}\label{A:MomentEstimator2}
Let $\widehat{\nu}_n$ be a moment-type estimator and $F$ be its associated mapping. The following conditions hold:    
\begin{enumerate}[label=(\roman*)] 
    \item $F$ is uniformly locally Lipschitz: there exist positive scalars $n_0$, $C_{\ref{A:MomentEstimator2},1}$ and $\epsilon$ such that $$\|F(\widetilde{m}_n)-F(m_n)\|_2\leq C_{\ref{A:MomentEstimator2},1} \|\widetilde{m}_n-m_n\|_2$$ for all  $\widetilde{m}_n$ such that $\|\widetilde{m}_n-m_n\|_2<\epsilon$ and $n\geq n_0$.
    \item  There exists a positive constant $C_{\ref{A:MomentEstimator2},2}$ such that  $\frac{1}{n}\|\phi^s\|_2^2 \leq C_{\ref{A:MomentEstimator2},2}$, $s\in [S_1+S_2]$ uniformly for all $n$.
\end{enumerate} 
\end{assumption}
\begin{remark}
   Assumption \ref{A:MomentEstimator2}-(i) imposes a weak condition on the local continuity of the mapping $F$. It rules out the case where the mapping $F$ becomes more singular at $m_n$ as $n$ increases. This may happen, for example, if the design matrix for a WLS estimator $\frac{1}{n}\X'\Om\bpi\X$, has an eigenvalue that approaches 0 as $n$ increases. Assumption \ref{A:MomentEstimator2}-(ii) is a typical data-moment condition.
\end{remark}

\begin{theorem}\label{Thm:Consistency}
Let $\widehat{\nu}_n$ be an estimator that satisfies Assumption \ref{A:MomentEstimator2}. If $\sigma_{\max}\left(\dmat\right)/n=o(1)$, then,
\begin{equation}
    \widehat{\nu}_n-\nu_n = O_p\left(\sqrt{\frac{\sigma_{\max}\left(\dmat\right)}{n}}\right).
\end{equation}
\end{theorem}
Theorem \ref{Thm:Consistency} implies the following corollary for the IPW, WLS, and GR estimators.
\begin{corollary}\label{C:Consistency}
Under Assumptions \ref{A:BoundedFourthMoments} and \ref{A:Invertibility}, and if $\sigma_{\max}\left(\dmat\right)=O(1)$, the IPW estimator converges to its probability limits at a $\sqrt{n}$-rate: $\widehat{\nu}^{\tht}_n-\nu^{\tht}_n=O_p\left(n^{-\frac{1}{2}}\right)$.

If, in addition, there exist positive $c$ and $C$ such that $0<c<\lambda_{\min}(\Om\bpi)<\lambda_{\max}(\Om\bpi)<C$ for all $n$, then the WLS and GR estimators converge to their probability limits at a $\sqrt{n}$-rate: $\widehat{\nu}^{\twls}_n-\nu^{\twls}_n=O_p\left(n^{-\frac{1}{2}}\right)$ and $\widehat{\nu}^{\tgr}_n-\nu^{\tgr}_n=O_p\left(n^{-\frac{1}{2}}\right)$.
\end{corollary}
\begin{remark}
Corollary \ref{C:Consistency} requires the largest singular value of the matrix $\dmat$ to be uniformly bounded. This condition can be shown to be satisfied for complete randomizations with treatment probability strictly bounded between 0 and 1. We check this condition in Section \ref{CR}.
It is also satisfied, for example, for stratified randomization, and cluster randomizations with bounded cluster sizes.\footnote{To be precise, we refer to stratified randomization as the randomization scheme with a fixed number of strata, diverging numbers of units in each stratum and non-vanshing treatment probabilities. We refer to the cluster randomization as the randomization scheme with increasing numbers of clusters, bounded maximum numbers of units in the clusters, and nonvanshing treatment probabilities. } There are settings where $\sigma_{\max}\left(\dmat\right)$ increases with sample size $n$. Examples of such cases are 1) cluster randomizations with increasing cluster sizes; 2) network experiments when the maximum degree of the network grows with the sample size. In these cases, if $\sigma_{\max}\left(\dmat\right)/n=o(1)$, the estimators still converge to its probability limit, although at a slower rate.   Corollary \ref{C:Consistency} also requires that the diagonal entries of $\Om\bpi$ are uniformly bounded above and below in $n$ for the WLS and GR estimators. This is satisfied, for example, if $\Om=\bpiInv$. 
\end{remark}
\begin{remark}
The condition on $\sigma_{\max}\left(\dmat\right)$ may be difficult to verify directly. Given that $\dmat$ is symmetric, one can provide an upper bound for this quantity using other matrix norms, such as the matrix column ($l_1$-induced) norm, Theorem 5.6.9 in \cite{horn2012matrix}).\footnote{Because $\dmat$ is symmetric, the matrix row ($l_\infty$-induced) norm and the maximum matrix column norm agree.} We use this property to check the condition for a two-arm completely randomized experiment in Section \ref{CR}.    
\end{remark}
We conclude this section with a lemma establishing the consistency of the WLS estimators for the average potential outcomes. This result implies that the inverse-probability weighted WLS estimator with centered covariates is consistent for estimating the average potential outcomes in general. Results for CI and MI estimators are included in Section \ref{other_estimators}.

\begin{lemma}\label{lemma:WLS}
If $\frac{1}{n}\sum_{i=1}^n x_{si}=0$ for all $s\in[p]$ and columns of the matrix $\onesmat \in \mathbb{R}^{kn\times k}$ are in the column space of the matrix $\bpi\Om\X$, then $\nu_n^{\twls}=\frac{1}{n}\onesmat'y$.
In particular, $\nu_n^{\twls}=\frac{1}{n}\onesmat'y$ if $\frac{1}{n}\sum_{i=1}^n x_{si}=0$ for all $s\in[p]$ and $\Om=\bpiInv$.\footnote{The current result is a special case of results in an unpublished work \cite{middleton2021unifying3}.}
\end{lemma}

\iffalse
\textcolor{red}{fixed effects?}
\fi
\subsection{Asymptotic Variances: Characterization, bounding and estimation}\label{Section:Variances}

To characterize the asymptotic variance, we strengthen Assumption \ref{A:MomentEstimator2} to use a linearization argument. In addition to using the notation $m_n$ and $\widehat{m}_n$ from Definition \ref{def:MomentEstimators}, we further define the column vectors $\widetilde{m}_{n,r}=(\widetilde{m}_n^1,...,\widetilde{m}_n^{l_1})'\in\mathbb{R}^{S_1}$ and $m_{n,r}=(m_n^1,...,m_n^{l_1})'\in\mathbb{R}^{S_1}$, as well as $\widetilde{m}_{n}=(\widetilde{m}_n^1,...,\widetilde{m}_n^{l_1},m_n^{S_1+1},...,m_n^{S_1+s_2})'\in\mathbb{R}^{S_1+S_2}$\footnote{The subscript r stands for "random moments".}
\begin{assumption}\label{A:MomentEstimators3} In addition to Assumption \ref{A:MomentEstimator2}, $F$ is uniformly locally linearly approximable: there exist positive scalars $n_0$, $C_{\ref{A:MomentEstimators3}}$, and $\epsilon$, and a linear map $dF_{m_n}\in\mathbb{R}^{k\times S_1}$ such that
\begin{equation}
 \|F(\widetilde{m}_n)-F(m_n)-dF_{m_n}\left(\widetilde{m}_{n,r}-m_{n,r}\right)\|_2\leq C_{\ref{A:MomentEstimators3}} \sum_{i=1}^{S_1}(\widetilde{m}_n^i-m_n^i)^2  
\end{equation}
for all $\widetilde{m}_n$ such that $\sum_{i=1}^{S_1}(\widetilde{m}_n^i-m_n^i)^2<\epsilon$ and for all $n\geq n_0$.
\end{assumption}
For a moment-type estimator $\widehat{\nu}_n=F(\widehat{m}_n)$, we denote the linearized version of $\widehat{\nu}_n$ as $\widehat{\nu}_n^\tL=F(m_n)+dF_{m_n}(\widehat{m}_{n,r}-m_{n,r})$, where $\widehat{m}_{n,r}=(\widehat{m}_n^1,...,\widehat{m}_n^{S_1})'\in\mathbb{R}^{S_1}$. Let $\diag()$ be the operator that maps a length-n vector to an n-by-n diagonal matrix.
\begin{theorem}\label{Thm:FirstOrderExpansion}
Define $\dmat$ as in (\ref{dmat}). Let $\widehat{\nu}_n$ be an estimator that satisfies Assumptions \ref{A:MomentEstimator2} and \ref{A:MomentEstimators3}, and $\sigma_{\max}\left(\dmat\right)/n=o(1)$. Then,
\begin{equation}\label{convergence_rate}
    \widehat{\nu}_n-\widehat{\nu}_n^\tL = o_p\left(\sqrt{\frac{\sigma_{\max}\left(\dmat\right)}{n}}\right)=o_p(1).
\end{equation}
The variance of $\widehat{\nu}_n^L$ can be written as
		\begin{equation}\label{var.general}
		\V(\widehat{\nu}_n^\tL)
		 = \frac{1}{n^2} z'\dmat z \in \mathbb{R}^{k\times k},
		\end{equation}
where,
\begin{equation}
z'=\sum_{s=1}^{S_1}dF^s_{m_n}1_{\scriptscriptstyle kn}'\diag(\phi^s)\in \mathbb{R}^{k\times kn} ,  
\end{equation}
and $dF^s_{m_n}$ is the $s$th column of the linear map $dF_{m_n}$ defined in Assumption \ref{A:MomentEstimators3} and $\phi^s$, $s\in[S_1]$, are the vectors of constants defined in Assumption \ref{def:MomentEstimators}.\footnote{Note that $dF_{m_n}^s\in\mathbb{R}^{k\times 1}$ and $1_{\scriptscriptstyle kn}'\diag(\phi^s)\in\mathbb{R}^{1\times kn}$.} Moreover, if $n^{-1}\|z\|^2_2=O(1)$ and $\sigma_{\max}\left(\dmat\right)=O(1)$, $\V(\widehat{\mu}_n^\tL)=O(\sigma_{\max}\left(\dmat\right)/n)=O(n^{-1})$.

\end{theorem}

This theorem is the key result of the section. It is shown that many linearized estimators have a variance that can be written as a "quadratic" form.\footnote{A typical definition of quadratic forms is the expression $v'\mathbf{A}v$, where $v$ is a column vector and $\mathbf{A}$ a square matrix. In the theorem below, our $v$ is a matrix instead of a column vector, so we are in a sense mis-using the name. If we are interested in a scalar parameter specified by a contrast vector $c$, then the variance $\V(\sum_{a=1}^k c_a\widehat{\mu}_a^{\tL})$ is a quadratic form in the standard sense. } Notice $\dmat$ depends only on the information of experimental designs, which is available to researchers. The matrix $z$ can depend on potential outcomes, the experimental design, and/or covariates, and it needs to be recovered from the data. 

\cite{middleton2021unifying} formally shows that the IPW estimator has an asymptotic variance that can be written as a bilinear form. Our results generalize \cite{middleton2021unifying}'s intuition regarding asymptotic variances with a rigorous proof. With an abuse of language, we refer to $n\V(\widehat{\nu}_n^{\tL})=\frac{1}{n}z'\dmat z$ as the \textit{asymptotic variance} of the estimator $\widehat{\nu}_n$. We specialize Theorem \ref{Thm:FirstOrderExpansion} for estimators introduced in Section \ref{Section:Estimation}.
\begin{corollary}\label{C:FirstOrderExpansion}
Under Assumptions \ref{A:BoundedFourthMoments} and \ref{A:Invertibility}, and if $\sigma_{\max}\left(\Omega\right)=O(1)$, the IPW estimator is $\sqrt{n}$-equivalent to its linearizations:
\begin{equation}
       \sqrt{n}(\widehat{\nu}^{\tht}_n-\widehat{\nu}_n^{\tht,\tL})=o_p(1), 
\end{equation}
with
\begin{align}
     & \widehat{\nu}_n^{\tht,\tL}= \frac{1}{n}\onesmat' y +  \frac{1}{n}\onesmat'\bpiInv(\R-\bpi) y, z^{\tht}=\diag(y)\onesmat.
\end{align}
If, in addition, there exist positive $c$ and $C$ such that $0<c<\lambda_{\min}(\Om\bpi)<\lambda_{\max}(\Om\bpi)<C$ for all $n$,  the WLS and GR estimators are $\sqrt{n}$-equivalent to their linearizations:
\begin{equation}
    \sqrt{n}(\widehat{\nu}^{\twls}_n-\widehat{\nu}_n^{\twls,\tL})=o_p(1),     \sqrt{n}(\widehat{\nu}^{\tgr}_n-\widehat{\nu}_n^{\tgr,\tL})=o_p(1),
\end{equation}
with
\begin{align}
        & \widehat{\nu}_n^{\twls,\tL}= \begin{bmatrix}
 \I_k \vert \0_{k\times p}
    \end{bmatrix}b^{\twls}_n +  \begin{bmatrix}
 \I_k \vert \0_{k\times p}
    \end{bmatrix}(\X'\Om\bpi \X)^{-1} \X'\Om (\R-\bpi) (y-\X b^{\twls}_n),\\
    & z^{\twls}=\diag(y-\X b^{\twls}_n)\bpi\Om\X(\frac{1}{n}\X'\Om\bpi \X)^{-1},\\
    & \widehat{\nu}_n^{\tgr,\tL}=\frac{1}{n}\onesmat' y + \frac{1}{n}\onesmat' \bpiInv(\R-\bpi)(y-\X' b^{\twls}_n ),\\
    & z^{\tgr}=\diag(y-Xb^{\twls}_n)\onesmat.
\end{align}
\end{corollary}
Now we turn to the subject of variance (bound) estimation. In the design-based framework, the true asymptotic variance is not identified nor consistently estimable. One can read off the lack-of-identification problem from entries in the first-order design matrix $\dmat$: some entries have the value $-1$, and this happens if $y_{ai}$ and $y_{bj}$ can never be simultaneously observed across all realized assignments and we have $\E[\left(\R_{ai}\R_{bj}\right)/\left(\pi_{ai}\pi_{bj}\right)]-\E[\R_{ai}/\pi_{ai}]\E[\R_{bj}\pi_{bj}]=0-1=-1$. To construct a variance bound, one needs to find a variance bound matrix $\widetilde{\dmat}$ that dominates $\dmat$ in the positive semidefinite sense. In order for $\widetilde{\dmat}$ to be identified, $\widetilde{\dmat}$ must take the value $0$ at entries that are $-1$ in $\dmat$. The concepts are formalized in \cite{mukerjee2018using} and \cite{middleton2021unifying}.

\begin{definition}[Identified variance bound matrix]\label{def.identified.bound}
	$\widetilde{\dmat}$ be is a identified variance bound matrix for $\dmat$, if 
	\begin{align}
	\bfI (\dmat=-1) \leq  \bfI(\widetilde{\dmat}= 0).
	\end{align}
	where $\leq$ denotes the pointwise inequality,
    and $\widetilde{ \dmat }-\dmat$ is positive semidefinite. $\bfI\left(\dmat=-1\right)$ is a $kn \times kn$ matrix of ones and zeros indicating the location of $-1$'s in $\dmat$, and $\bfI (\widetilde{ \dmat } =0 )$ is a $kn \times kn$ matrix of ones and zeros indicating the location of 0's in $\widetilde{ \dmat }$.
\end{definition}
Entries in $\bfI\left(\dmat=-1\right)$ are indications that the associated terms in the variance quadratics are impossible to observe. $\bfI (\dmat=-1) \leq  \bfI(\widetilde{\dmat}= 0)$ are indications that the associated terms are not used in the variance bound estimation and thus the variance bound (of an IPW estimator) is identified and can be estimated. We note that it is not always necessary to bound the entire matrix $\dmat$. Depending on the parameter of interest, one may only need to bound a principal submatrix of $\dmat$. This happens, for example, if the researcher is only interested in comparing two arms of a multi-arm experiment. From now on, for an estimator $\widehat{\nu}_n$ with an asymptotic variance $\frac{1}{n}z'\dmat z$, we refer to the estimator's asymptotic variance bound as $\frac{1}{n}z'\widetilde{\dmat} z$.\footnote{
To demonstrate how this definition maps to a typical case,  consider a two-arm completely randomized experiment with $n_t$ units in the treatment group and $n_c$ units in the control group, with $n=n_t+n_c$. Define the rescaled demeaning matrix $\A_n=\frac{n}{n-1}\left(\I_{n}-\frac{1}{n}\ones{n}\ones{n}'\right)\in\mathbb{R}^{n\times n}$. The first-order design matrix for the design is 
\begin{equation}
    \dmat=\begin{bmatrix}
    \frac{n_c}{n_t} \A_n & -\A_n\\
  -\A_n &  \frac{n_t}{n_c} \A_n'
\end{bmatrix}\in\mathbb{R}^{2n\times 2n}.
\end{equation}
The standard Neyman bound (e.g. see \cite{imbens2015causal}) matrix can be written as 
\begin{equation}\label{NeymanBound}
\widetilde{\dmat}^{\tgn}=\dmat+\begin{bmatrix}
\A_n&  \A_n\\
\A_n & \A_n
\end{bmatrix}=\begin{bmatrix}
  \frac{n}{n_t} \A_n&  0\\
0&   \frac{n}{n_c} \A_n
\end{bmatrix}\in\mathbb{R}^{2n\times 2n}.
\end{equation}
 Note that no entries in $\widetilde{\dmat}^{\tgn}$ take on the value $-1$ and the added matrix is a positive semidefinite matrix. Thus, the Neyman bound matrix is an identified variance-bound matrix by the definition.} 

For our purpose, we need a variance-bound matrix that can suit general experimental designs. The Aronow-Samii variance bound (\cite{aronow2017estimating}) is a general variance bound and applicable to arbitrary designs. We refer readers to \cite{harshaw2021optimized} for a general variance-bounding technique.
\begin{definition}[Aronow-Samii variance bound]
	%Let $\cmat = \diagSmallBracket{\onesmat c}$ be the $kn \times kn$ contrast matrix with the first $n$ diagonal elements equal to $c_1$, the next $n$ equal to $c_2$ and so on. Then 
	The Aronow-Samii variance bound uses the variance bound matrix
	\begin{align*}
	\widetilde{\dmat}^\tas=\dmat +\bfI\left(\dmat=-1\right)+ \diag (
		\bfI\left(\dmat=-1\right) \ones{kn}).
	\end{align*}
\end{definition}
With an identified variance-bound matrix $\widetilde{\dmat}$ (not necessarily the AS bound), we turn to variance bound estimation. We provide regularity conditions for consistent plug-in variance bound estimation. We first define the second-order design tensor.\footnote{The tensor is first introduced in \cite{middleton2021unifying}.}
 The operator norm of this object, which we define below, determines the rate of convergence for the variance-bound estimator. We use $\otimes$ to denote the \textit{tensor product} of two matrices, which results in an fourth-order tensor. For any two matrices $A=[a_{ij}]\in\mathbb{R}^{n_1\times n_2}$ and $B=[b_{ij}]\in\mathbb{R}^{n_3\times n_4}$, we denote
\begin{equation}
    A\otimes B = [c_{ijkl}]=[a_{ij}b_{kl}]\in \mathbb{R}^{n_1\times n_2\times n_3\times n_4}.
\end{equation}
\begin{definition}\label{SO}
  The \textit{second-order design tensor} is a fourth-order tensor $\s\in\mathbb{R}^{kn \times kn \times kn \times kn}$ of variances and covariances of inverse probability weighted pairwise joint inclusion indicators, written as
  \begin{align}\label{Smat}
      \s =\Big(  \E\left[\left(\R 1_{\scriptscriptstyle kn}1_{\scriptscriptstyle kn}'\R\right) \otimes \left( \R 1_{\scriptscriptstyle kn}1_{\scriptscriptstyle kn}'\R\right) \right]  -\matp \otimes \matp \Big) / \left( \matp \otimes \matp \right),
  \end{align}
  where $\matp=\E\left[\R 1_{\scriptscriptstyle kn}1_{\scriptscriptstyle kn}'\R\right]$ is a matrix with inclusion probabilities on the diagonal and pairwise joint inclusion probabilities off the diagonal, $\otimes$ is the tensor product, and $/$ is elementwise division with division by zero defined to be zero.
\end{definition}
Next, define an inverse probability weighted version of the variance bound matrix, $\widetilde{ \dmat }$, as
\begin{align} %\label{dmat.div.pmat}
\dtildep  = \widetilde{ \dmat } / \matp,
\end{align}
where $\matp$ is defined in Definition \ref{SO}  and $/$ denotes elementwise division with division by zero defined to be zero. With the matrix, an unbiased estimator of the variance bound $\widetilde{\V}\left(\widehat{\nu}_n^{\tL}\right) = \frac{1}{n^2}z'\widetilde{\dmat}z$ can be written as
\begin{align}%\label{var_ht_est}
\widehat{\widetilde{\V}}\left(\widehat{\mu}^\tL\right) =  \frac{1}{n^2}{z}' \R \dtildep \R z,
\end{align}
had $z$ is known. Since $z$ generally involves unknown quantities, an appeal to the plug-in principle suggests the use of
\begin{comment}
With the matrix $\dtildep$, an unbiased estimator of a variance bound for the IPW estimator can be written as
\begin{align}%\label{var_ht_est}
\widehat{\widetilde{\V}}\left(\widehat{\mu}^{\tht} \right) =  \frac{1}{n^2}{z^\tht}' \R \dtildep \R z^\tht,
\end{align}
with $z^\tht = \diag(y)\onesmat$.\footnote{ \cite{middleton2021unifying} proposed the procedure and proof for the unbiased estimation of the variance bound for the IPW estimator. We prove the validity of the plug-in procedure in full generality here. } 
This estimator is unbiased for the variance bound $\frac{1}{n^2}{z^\tht}' \widetilde{\dmat}z^\tht$ because $\E \left[\R \dtildep \R \right]=\widetilde{ \dmat }$ by construction.
For estimators in Table \ref{TableZ} other than the IPW estimator, the $z$'s in $z' \widetilde{ \dmat }z$ include quantities that must be estimated. 
\end{comment}
\begin{align}\label{var_lin_est}
\widehat{\widetilde{\V}}\left(\widehat{\nu}_n^{\tL}\right) = \frac{1}{n^2}\widehat{z}'  \dtildep \widehat{z}
\end{align} 
where $\widehat{z}$ has the same form as $z$ but with unknown quantities replaced by their estimators.\footnote{ \cite{middleton2021unifying} suggests the use of plug-in estimators but does not offer a formal justification.}\\
\indent Specializing to the IPW, WLS and GR estimators, the variance bound estimators have the form:
\begin{align}
&\widehat{\widetilde{\V}}\left(\widehat{\nu}_n^{\tht,\tL}\right) =  \frac{1}{n^2}\widehat{z}^{\tht'}\widetilde{\dmat}\widehat{z}^{\tht}, \text{with  }\widehat{z}^{\tht}= \R\diag(y)\onesmat\\
&\widehat{\widetilde{\V}}\left(\widehat{\nu}_n^{\twls,\tL}\right) =  \frac{1}{n^2}\widehat{z}^{\twls'}\widetilde{\dmat}\widehat{z}^{\twls}, \text{with  }\widehat{z}^{\twls}=\R\diag\left(\widehat{\epsilon}\right)\bpi\Om\X(\frac{1}{n}\X'\Om\bpi \X)^{-1},\\
& \widehat{\widetilde{\V}}\left(\widehat{\nu}_n^{\tgr,\tL}\right) =  \frac{1}{n^2}\widehat{z}^{\tgr'}\widetilde{\dmat}\widehat{z}^{\tgr}, \text{with  }\widehat{z}^{\tgr}=\R\diag\left(\widehat{\epsilon}\right)\onesmat,
\end{align}
where $\R\widehat{\epsilon}=\R\left(y-X\widehat{b}^{\twls}\right)$ and we note that $\R\diag\left(\widehat{\epsilon}\right)$ and $\R\diag\left(y\right)$ do not involve unobserved quantities.

\indent The following theorem considers the problem of consistent plug-in variance bound estimation. We introduce a few tensor notations used below.\footnote{We only introduce the minimally necessary notation here. For complete notation, refer to Appendix \ref{Mathematical Objects, Operators and Quantites}.} A real fourth order tensor $\mathbf{A}= \left ({a_{ \scriptscriptstyle i_{\scriptscriptstyle 1}...i_{\scriptscriptstyle 4}}}\right)\in\mathbb{R}^{n_1\times ...\times n_4} $ is a multi-array of entries, where $i_j=1,...,n_j$ for $j=1,...,4$. When $n=n_1=...=n_4$, $\mathbf{A}$ is called a fourth-order $n$-dimensional tensor. For a fourth-order n-dimensional tensor $\mathbf{A}$, we use the symbol $\sigma_{\max}(\mathbf{A})$ to denote the optimal value of the following optimization problem:
\begin{align}
          &\max_{w,x,y,z\in\mathbb{R}^n}\mathbf{A}(w,x,y,z)=\sum_{i=1}^n \sum_{j=1}^n\sum_{k=1}^n\sum_{l=1}^na_{ijkl}w_{i}x_{j}y_{k}z_{l}\\ &\text{   subject to  }    \sum_{i=1}^n w_i^4= 1, \sum_{i=1}^n x_i^4= 1, \sum_{i=1}^n y_i^4= 1 \text{ and }\sum_{i=1}^n z_i^4= 1.
\end{align}

This quantity is defined in \cite{lim2005singular} as a generalization of matrix singular values to tensors. Note that the arguments are constrained to be in the $l_4$ ball instead of the $l_2$ ball in $\mathbb{R}^n$. 

Recall that $\s$ is defined in Definition \ref{SO}. Let $\|\cdot\|_2$ denote the Frobenius norm if applied to a matrix and the $l_2$ vector norm if applied to a vector. Define $\circ$ to be the entrywise multiplication of two tensors. The following theorem provides the convergence rate for the plug-in variance bound estimator.
\begin{theorem}\label{Thm:ConsistentVariance}
Consider the estimator $\widehat{\widetilde{\V}}\left(\widehat{\mu}_n^\tL\right) = n^{-2}\widehat{z}'  \dtildep  \widehat{z} $ for the quantity $\widetilde{\V}\left(\widehat{\mu}_n^\tL\right) = n^{-2} z'\widetilde{\dmat}z$. If $n^{-1}\|\widehat{z}-\R z\|^2_2=O_p\left(\delta_n^2\right)$ and $n^{-1}\|z\|^2_2=O\left(1\right)$, then

\begin{equation}
    \widehat{\widetilde{\V}}(\widehat{\mu}_n^\tL) - \widetilde{\V}(\widehat{\mu}_n^\tL)=O_p\left(\max\left\{\sqrt{\frac{1}{n^3}\sigma_{\max}\left(\left(\widetilde{\dmat}\otimes \widetilde{\dmat}\right)\circ \s\right)}, \frac{\delta_n}{n}\sigma_{\max}\left(\dtildep\right)\right\}\right).
\end{equation}

\end{theorem}
We specialize the theorem for the IPW, WLS and GR estimators. 
\begin{corollary}\label{C:ConsistentVariance}
Under Assumptions \ref{A:BoundedFourthMoments} and \ref{A:Invertibility}, if $\sigma_{\max}\left(\dmat\right)=O(1)$, $\sigma_{\max}((\widetilde{\dmat}\otimes \widetilde{\dmat})\circ \s)=o(n)$,
$\sigma_{\max}\left(\dtildep\right)=O(1)$, the plug-in variance bound estimator is consistent:
\begin{equation*}    n(\widehat{\widetilde{\V}}(\widehat{\mu}^\tL_n) - \widetilde{\V}(\widehat{\mu}^\tL_n))=o_p(1)
\end{equation*}
for IPW, WLS and GR estimators. If there exist a positive integer $n_0$ and a positive constant $c$ such that $n\V(\mu^{\tL}_n)\succeq c\I_k$ for all $n\geq n_0$, then for all $t\in\mathbb{R}^k\backslash\{0\}$, $\frac{t'\widehat{\widetilde{\V}}(\widehat{\mu^{\tL}_n})t}{t'\widetilde{\V}(\widehat{\mu^{\tL}_n})t}\overset{p}{\to} 1$.\footnote{\cite{wu2021randomization} and \cite{lei2021regression} establish consistent variance bound estimation in two-arm completely randomized experiments with a weaker moment condition. It is possible to adapt their proof strategies to our setting by additionally assuming that $\|(\widetilde{\dmat}\otimes \widetilde{\dmat})\circ \s \|_{\infty}$  (as defined in Lemma \ref{LemmA:VanillaRootn}) is bounded uniformly for large $n$. We omit the proof here for simplicity.} 
\end{corollary}

\begin{remark}
   The conditions $\sigma_{\max}((\widetilde{\dmat}\otimes \widetilde{\dmat})\circ \s)=o(n)$ and
$\sigma_{\max}\left(\dtildep\right)=O(1)$ may be difficult to verify directly. Because $\dtildep$ is a symmetric matrix, we can bound $\sigma_{\max}\left(\dtildep\right)$ using the maximum row norm. On the other hand, $(\widetilde{\dmat}\otimes \widetilde{\dmat})\circ \s$ is a fourth-order tensor for which inequalities involving the stated quantity $\sigma_{\max}((\widetilde{\dmat}\otimes \widetilde{\dmat})\circ \s)$  are, to our knowledge, less established. We provide such an inequality. This inequality is shown to be useful for checking the conditions for completely randomized designs in Section \ref{CR}.
\end{remark}
\begin{lemma}\footnote{This lemma is Lemma \ref{LemmaTensor} in the appendix.}\label{LemmA:VanillaRootn}
Consider a fourth-order n-dimensional tensor $\mathbf{A}=[a_{ijkl}]\in\mathbb{R}^{n\times n\times n \times n}$. Define the quantities (absolute slice sums)
\begin{equation}
\|\A\|_{-i}=\max_i\sum_{j=1}^n\sum_{k=1}^n\sum_{l=1}^n |a_{ijkl}|, \hspace{2mm} \|\A\|_{-j}=\max_j\sum_{i=1}^n\sum_{k=1}^n\sum_{l=1}^n |a_{ijkl}|, 
\end{equation}
and similarly for $\|\A\|_{-k}$ and $\|\A\|_{-l}$. Define
\begin{equation}
    \|\A\|_{\infty} =\max\{\|\A\|_{-i},\|\A\|_{-j},\|\A\|_{-k},\|\A\|_{-l}\}.
\end{equation}
We have $  \sigma_{\max}(\mathbf{A})\leq \|\A\|_{\infty}$.
\end{lemma}
\begin{remark}
Note that $\dtildep$ may not be a positive-semidefinite matrix. This is a cause for concern because in some cases the estimated variance bound can be negative. Such problems do not arise in our simulation, though they are possible when the true variance bound is close to zero. One method to solve this problem is to set the negative spectra of the symmetric matrix $\dtildep$ to zero, resulting in an upwardly biased variance bound estimator.
\end{remark}

\subsection{Inference}\label{section:inference}
For hypothesis testing and confidence interval construction, standard theoretical and empirical practice is to establish a central limit theorem for the linearized estimators. For example, let $c\in \mathbb{R}^k$ be a contrast vector,  one usually wants to prove $\left(c'\widehat{\nu}_n^{\tL}-c'\nu_n\right)/\sqrt{c'\V\left(\widehat{\nu}_n^{\tL}\right)c}\overset{d}{\to} N(0,1)$, where $N(0,1)$ is a standard random variable with mean $0$ and variance 1. This justifies the usual inferential procedure for using a normal quantile for testing and confidence interval construction. 

Unfortunately, we are not aware of a CLT general enough to be applied with the assumptions introduced so far. It is difficult to prove a CLT in our setting because the assignment random variables are allowed to be almost arbitrarily correlated. In fact, Proposition 11  in \cite{savje2021average} has shown that the Chevbyshev's inequality is asymptotically sharp when there is strong interference. 

For the network experiments discussed in Section \ref{Section:NetworkExperiments}, the CLT can be justified when i) there are many disjoint components (e.g. villages and classrooms) in a large network and assignments are independent across different components but can be arbitrarily correlated within; ii) the network can be fully connected but the network has a bounded degree and assignments are independent among units.  

When a CLT is not applicable or cannot be rigorously justified, one can still use other tail bounds, for example, the Chebyshev inequality, for inference, although this may entail a loss of efficiency when a CLT indeed approximately holds. 
The following theorem provides a justification for the plug-in inference. It says that if a critical value can be used for infeasible inference with the linearized estimator, the plug-in inference with the moment-type estimators and variance-bound estimator will remain asymptotically valid under regularity conditions.
\begin{theorem}\label{plug-in-inference}
Let $\widehat{\nu}_n$ be a moment-type estimator and $\widehat{\nu}_n^L$ be the linearized estimator and, $c\in\mathbb{R}^k\backslash\{0\}$ a column contrast vector.  Suppose  $\sigma_{\max}\left(\Ome\right)/n=o(1)$, and there exists a positive constant $c_{\ref{plug-in-inference}}$ such that  $\V\left(c'\widehat{\nu}_n^{\tL}\right)\geq c_{\ref{plug-in-inference}} n^{-1}\sigma_{\max}\left(\dmat\right)$ uniformly for all $n$, and $\widetilde{\V}\left(c'\widehat{\nu}_n^{\tL}\right)/\widehat{\widetilde{\V}}\left(c'\widehat{\nu}_n^{\tL}\right)\overset{p}{\to}1$.  If there exists a continuous function $q:(0,1)\to\mathbb{R}$ such that $\lim_{n\to\infty}\textrm{P}_n\left(\left|c'\widehat{\nu}_{n}^{\tL}-c'\nu_{n}\right|\geq q(\alpha)\sqrt{c'\V(\widehat{\nu}_{n}^{\tL})c}\right)\leq \alpha$ for all $\alpha\in (0,1)$, then, under Assumptions \ref{A:MomentEstimator2} and \ref{A:MomentEstimators3},
$\lim_{n\to\infty}\textrm{P}_n\left(\left|c'\widehat{\nu}_{n}-c'\nu_{n}\right|\geq q(\alpha)\sqrt{  c'\widehat{\widetilde{\V}}(\widehat{\nu}_{n}^{\tL})}c\right)\leq \alpha$.
\end{theorem}
In particular, the theorem holds for $\widehat{\nu}_n^{\tht}$, $\widehat{\nu}_n^{\twls}$ and $\widehat{\nu}_n^{\tgr}$ under the premises of Corollary \ref{C:ConsistentVariance} and if there exist a positive integer $n_0$ and a positive constant $\epsilon$ such that $n\V(\nu^{\tL}_n)\succeq \epsilon\I_k$ for all $n\geq n_0$.

\subsection{$\sigma_{\max}\left(\Ome\right)$ as an input for designing experiments}\label{section:design}
The scalar value $\sigma_{\max}\left(\Ome\right)$ is the largest singular value of the first-order design matrix $\Ome$ and it appears in the convergence rate calculation in (\ref{Thm:Consistency}). A smaller $\sigma_{\max}\left(\Ome\right)$ generically implies a better estimation quality for moment-type estimators. We now show that $\sigma_{\max}\left(\Ome\right)$ can be interpreted as the worst-case variance for the IPW estimator. Recall the variance representation of the IPW estimator:
\begin{equation}
\V\left(\widehat{\nu}_n^{\tht}\right) = \frac{1}{n}\onesmat'\diag(y)\Ome\diag(y)\onesmat\in\mathbb{R}^k.
\end{equation}
With a (column) contract vector $c\in\mathbb{R}^k$, the variance of the IPW estimator $c'\widehat{\mu}_n^{\tht}$ is:
\begin{equation}
\V\left(c'\widehat{\nu}_n^{\tht}\right) = \frac{1}{n}c'\onesmat'\diag(y)\Ome\diag(y)\onesmat c\in \mathbb{R}.
\end{equation}
Since $\Ome$ is a semidefinite variance-covariance matrix, the largest singular value is the largest eigenvalue. The variational definition of the eigenvalues implies that:
\begin{equation}\label{variatonal}
\sigma_{\max}\left(\Ome\right) =\lambda_{\max}\left(\Ome\right)= \max_{\frac{1}{n}\sum_{i=1}^n\sum_{a=1}^k\left(c_ay_{ai}\right)^2 \leq 1} \V\left(c'\widehat{\mu}_n^{\tht}\right).
\end{equation}
Hence $\sigma_{\max}\left(\Ome\right)$ can be interpreted as the worst-case variance of the IPW estimator when the outcome is restricted such that $\frac{1}{n}\sum_{i=1}^n\sum_{a=1}^k \left(c_{a}y_{ai}\right)^2\leq 1$. 

The quantity introduced above may be too pessimistic because it measures the worst-case scenario with respect to all possible varying contrast vectors $c$ and the potential outcome vector $y$. For a fixed contrast vector $c$, we can define a contrast-specific first-order design matrix:
\begin{align}
     & \Ome^c = \diag(c'\onesmat')\Ome\diag(\onesmat c)=[\Ome_{aa'ii'}c_ac_{a'}]\in \mathbb{R}^{kn\times kn}, 
\end{align}
where $\Ome_{aa'ii'}=\cV\left(\R_{ai}/\pi_{ai},\R_{a'i'}/\pi_{a'i'}\right)$ is the entry in $\Ome$ associated with unit $i$ and treatment arm $a$, and unit $i'$ and arm $a'$. It can be shown that:
\begin{equation}
    \sigma_{\max}\left( \Ome^c\right) = \max_{\frac{1}{n}\sum_{a,i}y_{ai}^2\leq 1} \V\left(c'\widehat{\nu}_n^{\tht}\right),
\end{equation}
as a result of the algebraic identity:
\begin{align}
     \frac{1}{n}c'\onesmat'\diag(y)\Ome\diag(y)\onesmat c  = \frac{1}{n} y'\diag(c'\onesmat')\Omega \diag(\onesmat c) y
      = \frac{1}{n} y'\Ome^c y,
\end{align}
and again apply the variational definition of eigenvalues as in (\ref{variatonal}).

Hence, when we make no additional assumptions other than a scale restriction on the potential outcome vector $y$, the worst-case variance of the IPW estimator with a specific contrast vector $c$ is proportional to $\sigma_{\max}\left(\Ome^c\right)$. In this sense, if researchers perceive themselves as lacking a reliable prior on the potential outcomes, an experimental design with a smaller $\sigma_{\max}\left(\Ome^c\right)$ is preferred to one with a larger $\sigma_{\max}\left(\Ome^c\right)$, when the goal is to measure the average effects defined by the the contrast vector $c$.\footnote{If the researchers have a prior on the potential outcomes, they can instead consider the \textit{anticipated asymptotic variance} (\cite{isaki1982survey}), that is, $\frac{1}{n}\mathcal{E}[ z'\dmat z]$, where $\mathcal{E}$ is the expectation over the randomness with respect to the prior.} %For example, if the potential outcome vector is assumed to follow a multivariate normal distribution $y\sim N(0_{kn},\Sigma)$, where $\Sigma \in \mathbb{R}^{kn\times kn}$ is a variance-covariance matrix, the anticipated asymptotic variance of the IPW estimator is $\mathcal{E}[\frac{1}{n}c'\onesmat'\diag(y)\dmat\allowbreak \diag(y)\onesmat c]= \mathcal{E}[\frac{1}{n}y'\diag(c'\onesmat')\dmat\allowbreak \diag(c'\onesmat')y]=\frac{1}{n}\text{tr}\left(\Sigma \circ \left(\diag(c'\onesmat')\dmat\allowbreak \diag(c'\onesmat')\right) \right)$.}

The following numerical example from \cite{middleton2021unifying} is instructive to conceptualize the measure $\sigma_{\max}\left(\Ome^c\right)$.  We compare $\sigma_{\max}\left(\Ome^c\right)$ for two treatment-control experiment designs with four units. We compare a completely randomized design where two units are assigned to the treatment group and two to the control group, with a pairwise randomized design where the first two units form a pair and the last two units form a pair. The goal is to measure the average treatment effect defined by the contrast vector $c=(-1,1)$.

The $\Ome^c$'s for the two designs are:
{\footnotesize
\begin{equation}
    \Ome^{c,\textrm{cr}}=\begin{bmatrix}
        1  & -\frac{1}{3} & -\frac{1}{3} & -\frac{1}{3} & 1 &  -\frac{1}{3} & -\frac{1}{3} & -\frac{1}{3} \\
        -\frac{1}{3}  & 1 & -\frac{1}{3} & -\frac{1}{3} & -\frac{1}{3} &  1 & -\frac{1}{3} & -\frac{1}{3} \\    
        -\frac{1}{3}  & -\frac{1}{3} & 1 & -\frac{1}{3} & -\frac{1}{3} &   -\frac{1}{3} & 1 & -\frac{1}{3} \\   
        -\frac{1}{3}  & -\frac{1}{3} &  -\frac{1}{3}  & 1 & -\frac{1}{3} &   -\frac{1}{3} &  -\frac{1}{3} & -1 \\      
        1  & -\frac{1}{3} & -\frac{1}{3} & -\frac{1}{3} & 1 &  -\frac{1}{3} & -\frac{1}{3} & -\frac{1}{3} \\
        -\frac{1}{3}  & 1 & -\frac{1}{3} & -\frac{1}{3} & -\frac{1}{3} &  1 & -\frac{1}{3} & -\frac{1}{3} \\    
        -\frac{1}{3}  & -\frac{1}{3} & 1 & -\frac{1}{3} & -\frac{1}{3} &   -\frac{1}{3} & 1 & -\frac{1}{3} \\   
        -\frac{1}{3}  & -\frac{1}{3} &  -\frac{1}{3}  & 1 & -\frac{1}{3} &   -\frac{1}{3} &  -\frac{1}{3} & 1 \\      
    \end{bmatrix},    \Ome^{c,\textrm{pair}}=\begin{bmatrix}
        1  & -1 & 0 & 0 & 1 &  -1 & 0 & 0 \\
        -1  & 1 & 0 & 0 & -1 &  1 & 0 & 0 \\
                0  & 0 & 1 & -1 & 0 &  0 & 1 & -1 \\
        0  & 0 & -1 & 1 & 0 &  0 & -1 & 1 \\  
                1  & -1 & 0 & 0 & 1 &  -1 & 0 & 0 \\
        -1  & 1 & 0 & 0 & -1 &  1 & 0 & 0 \\
                0  & 0 & 1 & -1 & 0 &  0 & 1 & -1 \\
        0  & 0 & -1 & 1 & 0 &  0 & -1 & 1 \\ 
    \end{bmatrix}.
\end{equation}
}

We have the value 2.667 for $\sigma_{\max}\left(\Ome^{c,\textrm{cr}}\right)$ and 4 for $\sigma_{\max}\left(\Ome^{c,\textrm{pair}}\right)$. Hence in terms of robustness (worst-case variance), the completely randomized design is preferable to the pairwise randomized design. This comparison reflects the well-known fact that pairwise randomization is less efficient when potential outcomes are negatively correlated within pairs, as it depends on the researcher's prior belief that paired units have similar outcomes.

We note that this comparison should not be interpreted as a negative result for the matched-pair design. Instead, our primary goal is to underscore the trade-off between robustness and efficiency in experimental designs. In complex experimental settings where prior information on potential outcomes is limited or unavailable, we hope that the proposed metrics $\sigma_{\max}\left(\Ome\right)$ and $\sigma_{\max}\left(\Ome^c\right)$ are informative for guiding design choices. We demonstrate the use of these measures in our simulation section in Section \ref{Section:Simulation}.\footnote{There is a long literature on the minimax design of experiments \citep{wu1981robustness,li1983minimaxity,kallus2018optimal, harshaw2019balancing, bai2023randomize,basse2023minimax,ni2023design}. Our main contribution is to introduce a worst-case variance measure applicable to general experimental designs. Developing a general minimax experimental design is beyond the scope of the present paper.}

\section{Model-Assisted Estimators and Optimality}\label{Section:Nonlinear}
It can be shown that many commonly used estimators, including properly specified WLS estimators, are GR estimators.\footnote{We include a proof in Appendix \ref{other_estimators}.} GR estimators have the doubly-robust form
\begin{equation}
  \underbrace{\onesmat'  f(X,\widehat{\theta}_n)}_{\text{imputation}} +  \underbrace{\frac{1}{n} \onesmat' \bpiInv\R (y-  f(X,\widehat{\theta}_n)}_{\text{correction}}\in\mathbb{R}^k,
\end{equation}
where $f(X,\widehat{\theta}_n)=X'\widehat{b}^{\twls}_n$. 
One intuition reflected in this estimator is that if $X'\widehat{b}^{\twls}_n$ "predicts" the potential outcomes well, GR estimators may be expected to be more precise relative to the baseline IPW estimator in terms of (asymptotic) variances. This viewpoint motivates the model-assisted estimation strategy in the survey sampling literature (\cite{sarndal2003model}). However, if a model-assisted estimator is not constructed carefully, it may reduce asymptotic precision when compared with the IPW estimator \citep{freedman2008regression,lin2013agnostic}. 

In this section, we study the problem of model-assisted estimation strategies in general experimental designs using GR estimators. We focus on parametric models in this paper and leave the study of nonparametric or high-dimensional models for future work. We examine three classes of estimators and an additional hybrid class of estimators.\footnote{Each class has some precedence in the literature for some particular experimental designs, and we extend them to general experimental designs and consider a larger class of models.}

The first class we consider is the standard \textit{Quasi-Maximum Likelihood GR estimators} (QMLE-GR). This class includes the most commonly used models in the literature, such as linear, probit, and logit models. However, in terms of asymptotic variances, this estimation strategy is not guaranteed to be superior to a vanilla IPW estimator \citep{freedman2008b,lin2013agnostic}.  This problem motivates the second class of estimators, the \textit{no-harm GR estimators} (No-harm-GR). This class is based on the QMLE estimates but estimates a multiplicative constant in addition. This class of estimators provably yields an asymptotic variance no worse than that of the baseline IPW estimator. The final class is the \textit{optimal GR estimators} (Opt-GR). This class of estimators is optimal in that it achieves the greatest reduction in asymptotic variance among all GR estimators that use the same class of imputation functions (but differ in their parameter values).

Although Opt-GR estimators have strong asymptotic theoretical guarantees, we identify two limitations. One is that they tend to be unstable when $\sigma_{\max}\left(\dmat\right)$ is large, resulting in poor finite-sample performances. Another drawback is that they require solving nonconvex optimization problems, which may limit their compatibility with, for example, modern machine learning toolboxes. As a result, we propose a hybrid class of estimators called \textit{optimal-imputed} GR estimators (Opt-I GR). This estimator combines the insights from the No-harm GR and Opt-GR estimators: they are Opt-GR estimators with linear imputation functoins, but instead of using the full set of covariates, the linear models use a single covariate that is imputed by a QMLE model. We find that this class of estimators has good finite-sample performance in our simulations.

We remind readers that $k$ denotes the number of treatment arms and $n$ denotes the number of experiment units. Treatment arms are generically indexed by $a\in[k]$, and units are generically indexed by $i\in[n]$. In many cases, we write $\sum_{a=1}^k\sum_{i=1}^n$ as $\sum_{a,i}$ for simplicity. 

Our GR estimators are constructed using the imputation functions $f^a(\cdot,\theta): \mathcal{X}\to \mathbb{R},a\in[k]$ indexed by a finite-dimensional coefficient $\theta$, and $\mathcal{X}$ is the domain of the covariates. Different treatment arms can have their own imputation functions. All our GR estimators are constructed as follows:
\begin{enumerate}
    \item Given a parameter space $\Theta$, we define a population criterion $\mathcal{L}_n :\Theta \to \mathbb{R}$ and a sample criterion $\widehat{\mathcal{L}}_n: \Theta \to \mathbb{R}$.  The target coefficient $\theta_n$ and its estimator $\widehat{\theta}_n$ are defined as:\footnote{We will assume unique identification.}  
    \begin{equation}
        {\theta}_{n} =\arg\min_{\theta\in\Theta} \mathcal{L}_n(\theta),\text{ } \widehat{\theta}_{n} =\arg\min_{\theta\in\Theta} \widehat{\mathcal{L}}_n(\theta).
    \end{equation}
    \item For unit $i$ in the $a$th arm, we shall use the imputation functions to impute  $\widehat{y}_{ai}=f^a(x_i,\widehat{\theta}_{n})$. The estimator for the $a$th arm's average effect is constructed as:
    \begin{equation}\label{DR}
        \widehat{\mu}_{n,a} = \frac{1}{n}\sum_{i=1}^n f^a(x_i,\widehat{\theta}_{n}) + \frac{1}{n}\sum_{i=1}^n\frac{\R_{ai}}{\bpi_{ai}}  (y_{ai}-f^a(x_i,\widehat{\theta}_{n})).
    \end{equation}
    \item With a contrast vector $c\in\mathbb{R}^k$ defining the parameter of interest ($\mu_{n,c}=\frac{1}{n}c'\onesmat'y$), a GR estimator is constructed as:
    \begin{equation}
        \widehat{\mu}_{n,c}=\sum_{a=1}^k c_a \widehat{\mu}_{n,a}
    \end{equation}
\end{enumerate}
Different choices of the criterion $\mathcal{L}(\theta)$ give rise to different estimators. This section examines these choices and their implications.

We define additional notation for later use. Define the column vector $f^a(\theta)=\left(f^a(x_1,\theta),...,f^a(x_n,\theta)\right)\allowbreak\in\mathbb{R}^n$ and column vector $f(\theta)=\left(f^1(\theta)',f^2(\theta)',...,f^k(\theta)'\right)'\in\mathbb{R}^{kn}$.  In other words, $f^a(\theta)$ includes the imputed potential outcomes for all units in arm $a$ with parameter $\theta$, and $f(\theta)$ include the imputed outcomes of all arms. We first examine the QMLE-GR estimators and then use the results to motivate the no-harm GR and optimal-coefficient GR estimators.
\subsection{QMLE-GR estimators}\label{subsection: vanilla}
For the QMLE-GR estimators, the population and sample criteria are defined as follows:
\begin{equation}\label{M-pop}
     \mathcal{L}_n(\theta)=\frac{1}{n}\sum_{a=1}^k\sum_{i=1}^n \omega_{ai} g^a(y_{ai},x_i,\theta), \text{ }\widehat{\mathcal{L}}_n(\theta)=\frac{1}{n}\sum_{a=1}^k\sum_{i=1}^n\frac{\R_{ai}}{\bpi_{ai}} \omega_{ai}g^a(y_{ai},x_i,\theta). 
\end{equation}
The functions $g^a(\cdot,\theta)$ are measures of loss (e.g., squared loss or log likelihood), $\omega_{ai}'s$ are nonnegative weights.
For a parameter of interest $\mu_{n,c}=c'\frac{1}{n}\onesmat'y$, we construct the QMLE-GR estimator $\widehat{\mu}_{n,c}^{\tqmle}$ as in Algorithm \ref{QMLE-GR}. 
\begin{algorithm}\label{QMLE-GR}
\caption{QMLE-GR estimator for $\mu_{n,c}=c'\frac{1}{n}\onesmat'y$ }
\begin{enumerate}
\item Define $\widehat{\mathcal{L}}_n(\theta)$ as in (\ref{M-pop}) and obtain $\widehat{\theta}_n\equiv\arg\min_{\theta\in\Theta}\widehat{\mathcal{L}}_n(\theta)$.
    \item  Compute $\widehat{\mu}_{n,a}^{\tqmle}$, for $a\in [k]$, where
    \begin{equation}\label{DR}
        \widehat{\mu}_{n,a}^{\tqmle} = \frac{1}{n}\sum_{i=1}^n f^a(x_i,\widehat{\theta}_{n}) + \frac{1}{n}\sum_{i=1}^n\frac{\R_{ai}}{\bpi_{ai}}  (y_{ai}-f^a(x_i,\widehat{\theta}_{n})).
    \end{equation}
    \item Output $\widehat{\mu}_{n,c}^{\tqmle}=\sum_{a=1}^k c_a\widehat{\mu}_{n,a}^{\tqmle}$. 
\end{enumerate}
\end{algorithm}
\newline 
The following theorem collects the standard asymptotic  results for the QMLE-GR estimators. It shows that, under regularity assumptions, i) the estimator $\widehat{\mu}_{n,c}^{\tqmle}$ is equivalent to its linearized counterpart, ii) the plug-in asymptotic variance bound estimator is consistent, iii) the estimator and variance bound estimator can be used for inference. The regularity assumptions (Assumptions \ref{A:VanillaConsistency} and \ref{A:Imputation}) are standard \citep{newey1994large}, and we include them in Appendix \ref{regularity_assumptions}. For simplicity of notations, we shall focus on the $\sqrt{n}$-case for all our estimators below. 

Define $\theta_n\equiv\arg\min_{\theta\in\Theta} \mathcal{L}_n(\theta)$, where $\mathcal{L}_n(\theta)$ is defined in (\ref{M-pop}). Further, define the infeasible estimator:
\begin{equation}\label{qmle_linearized}
    \widehat{\mu}_{n,a}^{\tqmle,\tL}=\frac{1}{n}\sum_{i=1}^n f^a(x_i,\theta_n) + \frac{1}{n}\sum_{i=1}^n\frac{\R_{ai}}{\bpi_{ai}}  (y_{ai}-f^a(x_i,\theta_n)),
\end{equation} 
for $a\in[ k]$ and $\widehat{\mu}_{n,c}^{\tqmle,\tL}=\sum_{a=1}^k c_a \widehat{\mu}_{n,a}^{\tqmle,\tL}$. The variance of $\widehat{\mu}_{n,c}^{\tqmle,\tL}$ (asymptotic variance of $\widehat{\mu}^{\tqmle}_{n,c}$) can be expressed as
\begin{equation}\label{qmle_variance}
    \V(\widehat{\mu}_{n,c}^{\tqmle,\tL}) =\frac{1}{n^2} c'\onesmat'\diag(y-f(\theta_n))\hspace{1pt}\dmat\hspace{1pt} \diag(y-f(\theta_n))\onesmat c.
\end{equation}
\begin{theorem}\label{Thm:QMLE}
Define $\widehat{\theta}_n$,  $\widehat{\mu}_{n,a}^{\tqmle}$, and $\widehat{\mu}_{n,c}^{\tqmle}$ as in Algorithm \ref{QMLE-GR}. Under Assumptions \ref{A:BoundedFourthMoments},  \ref{A:VanillaConsistency}, and \ref{A:Imputation}, and if $\sigma_{\max}\left(\dmat\right)=O(1)$, $\sigma_{\max}((\widetilde{\dmat}\otimes \widetilde{\dmat})\circ \s)=o(n)$,
$\sigma_{\max}\left(\dtildep\right)=O(1)$, and there exists a positive constant $c_{\ref{Thm:QMLE}}$ such that $n\V(\widehat{\mu}_{n,c}^{\tqmle,\tL})\geq c_{\ref{Thm:QMLE}}$ uniformly for all large $n$, then following results hold:
\begin{enumerate}[label=(\roman*)]
   \item We have $\left(\widehat{\mu}^{\tqmle}_{n,c}-\widehat{\mu}_{n,c}^{\tqmle,\tL}\right)/\sqrt{\V(\widehat{\mu}_{n,c}^{\tqmle,\tL})}=o_p\left(1\right)$. 
\item Define the variance bound:
\begin{equation}\label{qmle_vbound}
    \widetilde{\V}(\widehat{\mu}_{n,c}^{\tqmle,\tL})= \frac{1}{n^2}c'\onesmat'\diag(y-f(\theta_n))\widetilde{\dmat}\diag(y-f(\theta_n))\onesmat c\in \mathbb{R},
\end{equation}
with an identified variance bound matrix $\widetilde{\dmat}$. The plug-in variance-bound estimator
\begin{equation}\label{qmle_vbound_est}
     \widehat{\widetilde{\V}}(\widehat{\mu}_{n,c}^{\tqmle,\tL})= \frac{1}{n^2}\onesmat'\diag(y-f(\widehat{\theta}_n))\R\dtildep\R\diag(y-f(\widehat{\theta}_n))\onesmat\in \mathbb{R} 
\end{equation}
is consistent: $\widehat{\widetilde{\V}}(\widehat{\mu}_{n,c}^{\tqmle,\tL})/\widetilde{\V}(\widehat{\mu}_{n,c}^{\tqmle,\tL})\overset{p}{\to} 1$.
\item If there exists a continuous function $q:(0,1)\to\mathbb{R}$ such that $\lim\sup_{n\to\infty}\textrm{P}_n\left(\left|\widehat{\mu}_{n,c}^{\tqmle}-\mu_{n,c}\right|\geq q(\alpha)\sqrt{\V(\widehat{\mu}_{n,c}^{\tqmle,\tL})}\right)\leq \alpha$ for all $\alpha\in (0,1)$, then, $\lim\sup_{n\to\infty}\textrm{P}_n\left(\left|\widehat{\mu}_{n,c}^{\tqmle}-\mu_{n,c}\right|\geq q\left(\alpha\right)\sqrt{  \widehat{\widetilde{\V}}(\hat{\mu}_{n,c}^{\tqmle,\tL})}\right)\leq \alpha $. 
\end{enumerate}
\end{theorem}
\subsection{ No-harm-GR estimators}\label{GMM}
As discussed in \cite{freedman2008regression} and \cite{cohen2020no}, QMLE-GR estimators may perform worse than the baseline IPW estimator in terms of asymptotic variances under "misspecification" of the adjusting model. In this section, we consider GR estimators that do not increase asymptotic variances regardless of the configurations of potential outcomes and covariates. Recall that for a parameter of interest $\mu_{n,c}=c'\frac{1}{n}\onesmat'y$, the asymptotic variance of a QMLE-GR estimator is 
\begin{equation}
     \frac{1}{n}c' \onesmat'\diag(y-f(\theta_n))\hspace{1pt}\dmat\hspace{1pt} \diag(y-f(\theta_n))\onesmat c.    
\end{equation}

There is no guarantee that this asymptotic variance is smaller than the variance of the IPW estimator $ \frac{1}{n}c' \onesmat'\diag(y)\hspace{1pt}\dmat\hspace{1pt} \diag(y)\onesmat c$, which is the asymptotic variance of the IPW estimator. To guarantee that our estimation strategy does not result in any harm, we can further define a multiplicative constant $\alpha_n$ that solves the problem
\begin{equation}\label{alphanc} \alpha_n^c\equiv\arg\min_{\alpha\in\mathbb{R}} \frac{1}{n}c' \onesmat'\diag(y-\alpha f(\theta_n))\hspace{1pt}\dmat\hspace{1pt} \diag(y-\alpha  f(\theta_n))\onesmat c.
\end{equation}
Note that the minimum is guaranteed to be no larger than $\frac{1}{n}c' \onesmat'\diag(y)\hspace{1pt}\dmat\hspace{1pt} \diag(y)\onesmat c$ since $\alpha=0$ is in the choice set.
The analytical expression for $\alpha_n^c$ is 
\begin{equation}\label{alphanc}
    \alpha_n^c= \frac{c'\onesmat'\diag\left(y\right)\dmat\diag\left(f\left({\theta}_n\right)\right)\onesmat c}{c'\onesmat'\diag\left(f\left(\theta_n\right)\right)\dmat\diag\left(f\left(\theta_n\right)\right)\onesmat c},
\end{equation}
for which we can construct a feasible consistent estimator:
\begin{equation}\label{alphanchat}    \widehat{\alpha}_n^c=\frac{c'\onesmat'\diag(y)\R\bpiInv\dmat\diag(f(\widehat{\theta}_n))\onesmat c}{c'\onesmat'\diag(f(\widehat{\theta}_n))\dmat\diag(f(\widehat{\theta}_n))\onesmat c}.
\end{equation}
Algorithm \ref{NOHARM-GR} gives a formal template for constructing the No-harm GR estimators.\footnote{    This class of estimators is motivated by \cite{cohen2020no}'s no-harm estimator for a two-arm completely randomized design. Their estimator is different from ours because they use the outputs of QMLE imputations (i.e., $f^a(x_i,\widehat{\theta}_n)$) as regressors for an interacted linear regression model and apply \cite{lin2013agnostic}'s results. In general, one can design no-harm estimators that are more flexible than the No-harm-GR estimators considered here. For example, one can define the multiplicative constants separately for each arm or use the QMLE imputations of arms as regressors, as in \cite{cohen2020no}. We will discuss one class of such estimators in the section below.} 
\begin{algorithm}[h]\label{NOHARM-GR}
\caption{No-harm-GR estimator for $\mu_{n,c}=c'\frac{1}{n}\onesmat'y$ }
\begin{algorithmic}
\Require QMLE coefficient estimates $\widehat{\theta}_n$ from Algorithm \ref{QMLE-GR}
\begin{enumerate}
    \item Compute $\widehat{\alpha}^c_{n}$ as in (\ref{alphanchat}).
    \item  Compute $\widehat{\mu}_{n,a}^{\tNO}$, for $a\in [k]$, where
    \begin{equation}\label{DR}
        \widehat{\mu}_{n,a}^{\tNO} = \frac{1}{n}\sum_{i=1}^n \widehat{\alpha}_n^cf^a(x_i,\widehat{\theta}_{n}) + \frac{1}{n}\sum_{i=1}^n\frac{\R_{ai}}{\bpi_{ai}}  \left(y_{ai}-\widehat{\alpha}_n^c f^a(x_i,\widehat{\theta}_{n})\right).
    \end{equation}
    \item Output $\widehat{\mu}_{n,c}^{\tNO}=\sum_{a=1}^k c_a\widehat{\mu}_{n,a}^{\tNO}$. 
\end{enumerate}
\end{algorithmic}
\end{algorithm}
Statistical guarantees are derived under the following additional assumptions.  
\begin{assumption}\label{A:NOHARM}
 There exists a positive integer $n_0$ and a positive constant $c_{\ref{A:NOHARM}}$ such that
    \begin{equation}
        \frac{1}{n}c'\onesmat'\diag\left(f(\theta_n)\right)\dmat\diag\left(f(\theta_n)\right)c\geq c_{\ref{A:NOHARM},1},
    \end{equation}
    uniformly for all $n\geq n_0$.
\end{assumption}
\begin{remark}
Note that Assumption \ref{A:NOHARM} may be violated in some cases. With linear models and a two-arm completely randomized design, the condition may be violated if all the QMLE coefficients of the covariates are zeros, reducing the imputation functions to arm-specific intercepts. Because the first-order design matrix $\dmat$ for the complete randomization nullifies the intercept matrix (i.e., $\dmat\onesmat=\0_{kn\times k}$), the quantity $\alpha_n^c$ in (\ref{alphanc}) is either weakly identified (when the coefficients are near-zero) or not identified. We shall discuss recommended practices later, after introducing Optimal-GR estimators.
\end{remark}
Define $\theta_n$ as the target QMLE coefficient as before and $\alpha_n^c$ as in (\ref{alphanc}). Define the infeasible estimator:
\begin{equation}
\widehat{\mu}_{n,a}^{\tNO,\tL}=\frac{1}{n}\sum_{i=1}^n \alpha_n^cf^a(x_i,\theta_n) + \frac{1}{n}\sum_{i=1}^n\frac{\R_{ai}}{\bpi_{ai}}  \left(y_{ai}-\alpha_n^cf^a(x_i,\theta_n)\right)
\end{equation} 
for $a\in[ k]$ and $\widehat{\mu}_{n,c}^{\tNO,\tL}=\sum_{a=1}^k c_a \widehat{\mu}_{n,a}^{\tNO,\tL}$. The variance of $\widehat{\mu}_{n,c}^{\tNO,\tL}$ (asymptotic variance of $\widehat{\mu}^{\tNO}_{n,c}$) can be expressed as
\begin{equation}
    \V(\widehat{\mu}_{n,c}^{\tNO,\tL}) =\frac{1}{n^2} c'\onesmat'\diag\left(y-\alpha_{n}^cf\left(\theta_n\right)\right)\hspace{1pt}\dmat\hspace{1pt} \diag\left(y-\alpha_{n}^c f\left(\theta_n\right)\right)\onesmat c.
\end{equation}
Denote the  $l_1$-induced matrix norm by $\viiii{A}_1=\max_{j\in [n]}\{\sum_{i=1}^n |a_{ij}|\}$.
\begin{theorem}\label{Thm:NOHARM}
Define $\widehat{\mu}_{n,c}^{\tNO}$ as in Algorithm \ref{NOHARM-GR}. Under Assumptions \ref{A:BoundedFourthMoments}, \ref{A:NOHARM}, \ref{A:VanillaConsistency} and \ref{A:Imputation}, and if $\viiii{\dmat}_1=O(1)$, $\sigma_{\max}((\widetilde{\dmat}\otimes \widetilde{\dmat})\circ \s)=o(n)$, $\sigma_{\max}\left(\dtildep\right)=O(1)$ and there exists a positive constant $c_{\ref{Thm:NOHARM},2}$ such that $n\V(\widehat{\mu}_{n,c}^{\tNO,\tL})\geq c_{\ref{Thm:NOHARM},2}$ uniformly for all large $n$, then following results hold:
\begin{enumerate}[label=(\roman*)]
   \item We have $\left(\widehat{\mu}^{\tqmle}_{n,c}-\widehat{\mu}_{n,c}^{\tqmle,\tL}\right)/\sqrt{\V(\widehat{\mu}_{n,c}^{\tqmle,\tL})}=o_p\left(1\right)$. 
   \item (No Harm) Define $\widehat{\mu}_{n,c}^{\tht}=c'\frac{1}{n}\bpi^{-1}\R y$. We have:
   \begin{equation}
    \V(\widehat{\mu}_{n,c}^{\tNO,\tL}) \leq \V(\widehat{\mu}_{n,c}^{\tht}) =\frac{1}{n^2} c'\onesmat'\diag\left(y\right)\hspace{1pt}\dmat\hspace{1pt} \diag\left(y\right)\onesmat c.   
   \end{equation}
\item Define the variance bound:
\begin{equation}
    \widetilde{\V}(\widehat{\mu}_{n,c}^{\tNO,\tL})= \frac{1}{n^2}c'\onesmat'\diag\left(y-\alpha_n^cf(\theta_n)\right)\widetilde{\dmat}\diag\left(y-\alpha_n^cf(\theta_n)\right)\onesmat c\in \mathbb{R},
\end{equation}
with an identified variance bound matrix $\widetilde{\dmat}$. The plug-in variance-bound estimator
\begin{equation}
     \widehat{\widetilde{\V}}(\widehat{\mu}_{n,c}^{\tNO,\tL})= \frac{1}{n^2}\onesmat'\diag\left(y-f(\widehat{\theta}_n)\right)\R\dtildep\R\diag\left(y-f(\widehat{\theta}_n)\right)\onesmat\in \mathbb{R} 
\end{equation}
is consistent: $\widehat{\widetilde{\V}}(\widehat{\mu}_{n,c}^{\tNO,\tL})/\widetilde{\V}(\widehat{\mu}_{n,c}^{\tNO,\tL})\overset{p}{\to} 1$.
\item If there exists a continuous function $q:(0,1)\to\mathbb{R}$ such that $\lim\sup_{n\to\infty}\textrm{P}_n\left(\left|\widehat{\mu}_{n,c}^{\tNO}-\mu_{n,c}\right|\geq q(\alpha)\sqrt{\V(\widehat{\mu}_{n,c}^{\tNO,\tL})}\right)\leq \alpha$ for all $\alpha\in (0,1)$, then, $\lim\sup_{n\to\infty}\textrm{P}_n\left(\left|\widehat{\mu}_{n,c}^{\tNO}-\mu_{n,c}\right|\geq q\left(\alpha\right)\sqrt{  \widehat{\widetilde{\V}}(\hat{\mu}_{n,c}^{\tNO,\tL})}\right)\leq \alpha $. 
\end{enumerate}
\end{theorem} 
\begin{remark}
         This assumption $\viiii{\dmat}_1=O(1)$ is stronger than the condition $\sigma_{\max}\left(\dmat\right)=O(1)$, as $\sigma_{\max}\left(\dmat\right)\leq \viiii{\dmat}_1$, but it is satisfied by many standard designs (e.g., completely randomized designs, pairwise randomized designs, Bernoulli designs on a network with bounded degrees). 
\end{remark}

\subsection{ Optimal GR estimators (Opt-GR)}\label{Opt-GMM}
QMLE-GR estimators and No-harm-GR estimators both elicit $\theta_n$ from other criteria. Instead of requiring $\theta_n$ to be the minimizer of a pseudo-criterion, we can choose $\theta_n$ to directly minimize the implied asymptotic variances. That is, we define 
\begin{equation}\label{optimal-gr}
    \theta_n\equiv\arg\min_{\theta\in \Theta}  \frac{1}{n^2}c' \onesmat'\diag(y-f(\theta))\hspace{1pt}\dmat\hspace{1pt} \diag(y-f(\theta))\onesmat c.\footnote{In theory, there is a possibility that even the minimized asymptotic variance may be larger than that of the IPW estimator $\frac{1}{n^2} c' \onesmat'\diag(y)\hspace{1pt}\dmat\hspace{1pt} \diag(y)\onesmat c$. This will not happen if the imputation functions can impute the $\0_{kn}$ vector (i.e., $\0_{kn}\in \{f(\theta)\}_{\theta\in\Theta}$) or if the $\dmat$ matrix nullifies the intercept matrix $\dmat\onesmat=\0_{kn\times k}$ and, for each arm, the imputation functions can impute the same potential outcomes for each unit. If neither condition is satisfied and 0 is in the parameter set $\Theta$, one can guarantee no harm by shifting the parameter set by $f(0)$ and define $\theta_n$ such that $\theta_n=\arg\min_{\theta\in\Theta}\frac{1}{n}c' \onesmat'\diag(y+f(0)-f(\theta_n))\hspace{1pt}\dmat\hspace{1pt} \diag(y+f(0)-f(\theta_n))\onesmat c$. }
\end{equation}
The coefficient $\theta_n$ is optimal in the sense that the implied asymptotic variance $\frac{1}{n}c' \onesmat'\diag\left(y-f(\theta)\right)\hspace{1pt}\dmat\hspace{1pt} \diag\left(y-f(\theta)\right)\onesmat c$ is the smallest when compared with estimators using the same class of parametric models for adjustments.\footnote{If we replace the first-order design matrix $\dmat$ with an identified variance bound $\widetilde{\dmat}$, the implied $\theta_n$ will instead minimize the asymptotic variance bound.} However, it is impossible to construct a sample analog to the criterion above because some entries in $y-f(\theta_n)$ are never observed simultaneously. This is the same problem we encountered in variance estimation, where some pairs of potential outcomes can never be observed simultaneously. 

To identify and estimate $\theta_n$, we modify the criterion to allow for a feasible sample analog. Specifically, we define:
\begin{equation}\label{criterion_81}
    \mathcal{L}_n^{\toc}\left(\theta\right) =  \frac{1}{n^2}c' \onesmat'\diag(f(\theta))\hspace{1pt}\dmat\hspace{1pt} \diag(f(\theta))\onesmat c -  \frac{2}{n^2}c' \onesmat'\diag(y)\hspace{1pt}\dmat\hspace{1pt} \diag(f(\theta))\onesmat c,\footnote{An alternative approach based on the first order conditions can be found in an earlier version of the paper \cite{chang2023design}. We chose the current formulation because it is numerically more stable in our simulation.}
\end{equation}
and its sample analog:
\begin{equation}\label{criterion_82}
    \widehat{\mathcal{L}}^{\toc}_n\left(\theta\right) =  \frac{1}{n^2}c' \onesmat'\diag(f(\theta))\hspace{1pt}\dmat\hspace{1pt} \diag(f(\theta))\onesmat c -  \frac{2}{n^2}c' \onesmat'\diag(y)\hspace{1pt}\R\bpiInv\dmat\hspace{1pt} \diag(f(\theta))\onesmat c.
\end{equation}
The criterion $ \mathcal{L}_n^{\toc}(\theta)$ differs from the original criterion (\ref{optimal-gr}) by the term $\frac{1}{n^2}c' \onesmat'\diag(y)\hspace{1pt}\dmat\hspace{1pt} \diag(y)\onesmat c$. This term does not depend on the parameter value $\theta$ and hence minimizers for (\ref{optimal-gr}) and  (\ref{criterion_81}) are identical. Furthermore, unlike (\ref{optimal-gr}), $\mathcal{L}_n(\theta)$ avoids the joint non-observability issue caused by the term $\frac{1}{n^2}c' \onesmat'\diag(y)\hspace{1pt}\dmat\hspace{1pt} \diag(y)\onesmat c$ and hence admits a valid sample analog $\widehat{\mathcal{L}}^{\toc}(\theta)$.

The Opt-GR estimators are constructed according to Algorithm \ref{Opt-GR}.
\begin{algorithm}[h]\label{Opt-GR}
\caption{Opt-GR estimator for $\mu_{n,c}=c'\frac{1}{n}\onesmat'y$ }
\begin{enumerate}
\item Define
$\widehat{\mathcal{L}}^{\toc}_n(\theta)$ as in (\ref{criterion_82}) and obtain $\widehat{\theta}_n\equiv\arg\min_{\theta\in\Theta}\widehat{\mathcal{L}}^{\toc}_n(\theta)$.
    \item  Compute $\widehat{\mu}^{\toc}_{n,a}$, for $a\in [k]$, where
    \begin{equation*}\
        \widehat{\mu}_{n,a}^{\toc} = \frac{1}{n}\sum_{i=1}^n f^a(x_i,\widehat{\theta}_{n}) + \frac{1}{n}\sum_{i=1}^n\frac{\R_{ai}}{\bpi_{ai}}  (y_{ai}-f^a(x_i,\widehat{\theta}_{n})).
    \end{equation*}
    \item Output $\widehat{\mu}_{n,c}^{\toc}=\sum_{a=1}^k c_a\widehat{\mu}^{\toc}_{n,a}$. 
\end{enumerate}
\end{algorithm}
Theoretical guarantees can be derived under the standard identification assumptions, with additional standard regularity assumptions (Assumption \ref{A:GMM2new}) which we include in Appendix \ref{regularity_assumptions}. Though the identification assumption is standard,  it requires careful consideration in our setting, so we highlight it here. For a set $C \subset \mathbb{R}^s$, we denote its boundary under the standard Euclidean topology by $\textrm{Bd}\left(C\right)$. 
\begin{assumption}\label{A:GMM1} There exists a positive integer $n_0$ such that for all $n\geq n_0$ the following conditions hold:
\begin{enumerate}[label=(\roman*)] 
    \item The parameter space $\Theta$ is a compact set in $\mathbb{R}^{s}$ with an nonempty interior. 
    \item There exists a $\theta_n$ and a positive $c_{\ref{A:GMM1},1}$ such that, for any $\epsilon>0$, $\inf_{\theta\in \Theta\backslash B(\theta_n,\epsilon)}\mathcal{L}_n(\theta)-\mathcal{L}_n(\theta_n)>c_{\ref{A:GMM1},1}\epsilon^2$. Moreoever, there exists a $\delta$ such that $\text{dist}(\theta_n,\mathrm{Bd}(\Theta))>\delta$.
\end{enumerate}
\end{assumption}

Define $\widehat{\mu}_{n,a}^{\toc,\tL}$, $\widehat{\mu}_{n,c}^{\toc,\tL}=\sum_{a=1}^k c_a\widehat{\mu}_{n,a}^{\toc,\tL}$, and $ \V(\widehat{\mu}_{n,c}^{\toc,\tL})$ as in (\ref{qmle_linearized}) and (\ref{qmle_variance}), but with $\theta_n$ replaced by $\theta_n\equiv\arg\min_{\theta\in\Theta} \mathcal{L}_n(\theta)$ where $\mathcal{L}_n(\theta)$ is defined in (\ref{criterion_81}).
\begin{theorem}\label{Thm:OC}
Define $\widehat{\theta}_n$,  $\widehat{\mu}_{n,a}^{\toc}$, and $\widehat{\mu}_{n,c}^{\toc}$ as in Algorithm \ref{Opt-GR}. Under Assumptions \ref{A:BoundedFourthMoments},\ref{A:GMM1}, and \ref{A:GMM2new} and if $\viiii{\dmat}_1=O(1)$, $\sigma_{\max}((\widetilde{\dmat}\otimes \widetilde{\dmat})\circ \s)=o(n)$, and
$\sigma_{\max}\left(\dtildep\right)=O(1)$, and if there exists a positive constant $c_{\ref{Thm:OC},2}$ such that $n\V(\widehat{\mu}_{n,c}^{\toc,\tL})\geq c_{\ref{Thm:OC},2}$ uniformly for all large $n$, then following results hold: \\
\begin{enumerate}[label=(\roman*)]
   \item We have $\left(\widehat{\mu}^{\toc}_{n,c}-\widehat{\mu}_{n,c}^{\toc,\tL}\right)/\sqrt{\V(\widehat{\mu}_{n,c}^{\toc,\tL})}=o_p\left(1\right)$. 
   \item (Optimality) $ \V(\widehat{\mu}_{n,c}^{\toc,\tL})= \min_{\theta\in\Theta}\frac{1}{n}c' \onesmat'\diag\left(y-f(\theta)\right)\hspace{1pt}\dmat\hspace{1pt} \diag\left(y-f(\theta)\right)\onesmat c$.
\item Define the variance bound $\widetilde{\V}(\widehat{\mu}_{n,c}^{\toc,\tL})$  as in (\ref{qmle_vbound}), with an identified variance bound matrix $\widetilde{\dmat}$,  and the the plug-in variance-bound estimator $\widehat{\widetilde{\V}}(\widehat{\mu}_{n,c}^{\toc,\tL})$ as in (\ref{qmle_vbound_est}). The variance bound estimator is consistent
: $\widehat{\widetilde{\V}}(\widehat{\mu}_{n,c}^{\toc,\tL})/\widetilde{\V}(\widehat{\mu}_{n,c}^{\toc,\tL})\overset{p}{\to} 1$. 
\item If there exists a continuous function $q:(0,1)\to\mathbb{R}$ such that $\lim\sup_{n\to\infty}\textrm{P}_n\left(\left|\widehat{\mu}_{n,c}^{\toc}-\mu_{n,c}\right|\geq q(\alpha)\sqrt{\V(\widehat{\mu}_{n,c}^{\toc,\tL})}\right)\leq \alpha$ for all $\alpha\in (0,1)$, then, $\lim\sup_{n\to\infty}\textrm{P}_n\left(\left|\widehat{\mu}_{n,c}^{\toc}-\mu_{n,c}\right|\geq q\left(\alpha\right)\sqrt{  \widehat{\widetilde{\V}}(\hat{\mu}_{n,c}^{\toc,\tL})}\right)\leq \alpha $. 
\end{enumerate}
\end{theorem} 
\begin{remark}
Our notion of optimality in Theorem \ref{Thm:OC}-(ii) and (\ref{optimal-gr}) is weaker than the semiparametric efficiency bound commonly used in the i.i.d. sampling literature \citep{bickel1993efficient}. To our knowledge, this notion of semiparametric efficiency has not been developed in the design-based statistical literature. Instead, our definition of estimator optimality aligns with that of \cite{lin2013agnostic} and \cite{middleton2018unified}: we consider an estimator optimal if it achieves the smallest asymptotic variance, within a given class of adjustments.
\end{remark}
\begin{remark}
We now comment on Assumption \ref{A:GMM1}-(ii) extensively. Identifiability is a very important assumption here. In general, there can be multiple optimal solutions depending on the adjustment strategy (\cite{lin2013agnostic,middleton2018unified}). This is because some experimental designs will make the GR estimator invariant to certain parameter choices. For example, consider the setting of a two-arm  completely randomized design with one pretreatment covariate. Let 1 denote the treatment arm and 0 denote the control arm, and let $n_1$ and $n_0$ denote the number of treated and control units, respectively. Define $\Bar{x}=\frac{1}{n}\sum_{i}x_i=0$ (centered), $\widehat{\Bar{x}}_1 =\frac{1}{n_1}\sum_{i}\R_{1i}x_i$, and $\widehat{\Bar{x}}_0 =\frac{1}{n_0}\sum_{i}\R_{0i}x_i$, and note $\widehat{\Bar{x}}_0=-\frac{n_1}{n_0}\widehat{\Bar{x}}_1$. Suppose we use a separate-slope linear model to adjust $f^0(x_i,\beta)=\beta_0+\beta_1 x_i$ and $f^1(x_i,\beta)=\beta_2+\beta_3 x_i$ and $\frac{1}{n}\sum_{i=1}^n x_i=0$. The GR estimator for the ATE is
\begin{align}
   & (\beta_2-\beta_0)+\frac{1}{n}\sum_{i}x_i(\beta_3-\beta_1)-\left(\frac{1}{n_1}\sum_{i}\R_{1i}(\beta_2+x_i\beta_3)-\frac{1}{n_0}\sum_{i}\R_{0i}(\beta_0+x_i\beta_1)\right)\\
    & =0 - \widehat{\Bar{x}}_1\beta_3 +  \widehat{\Bar{x}}_0\beta_1 = -\widehat{\Bar{x}}_1\left(\beta_3 + \frac{n_1}{n_0}\beta_1\right).        
\end{align}
We observe that the intercepts $\beta_0$ and $\beta_2$ are canceled and that there are multiple pairs of $(\widetilde{\beta}_1,\widetilde{\beta}_3)$ that are equivalent to $(\beta_1,\beta_3)$, such as $\widetilde{\beta}_1=0$ and $\widetilde{\beta}_3=\beta_3 + \frac{n_1}{n_0}\beta_1$.\footnote{We note that the identification problem highlighted here is design specific. For example, this problem will also arise for pairwise randomized designs but not for Bernoulli designs.} 

In general, there are two sources of the weak/non-identification problem. The first is weak/non identification of arm-specific intercepts. This typically happens when the number of units in an arm is fixed in all possible random treatment allocations. Another source of weak/non-identification is that the design induces co-linearity or cancellation of covariates coefficients, for example, in pairwise randomized designs.

For linear models, both problems can be detected by inspecting the eigenvalues of $\frac{1}{n}\X\dmat^c\X$. Small eigenvalues correspond to the possible cancellation of covariate coefficients or arm-specific intercepts. To use the Opt-GR estimators, we recommend that researchers inspect these eigenvalues of before applying the estimators. If one of the eigenvalues is small, the researcher may want, for example, to avoid specifying an intercept for the treatment arms or use a model with same-slope adjustments instead of a model with separate-slope adjustments. 
    
\end{remark}

We end this section by noting that one can combine the insights of No-harm GR and Opt-GR estimators to form a class of Optimal-Imputed (Opt-I) GR estimators. This class of estimators is Opt-GR estimators with a linear imputation function. Instead of using the full set of covariates, the linear imputation function uses a single covariate imputed by a QMLE model (for example, a flexible ML model).\footnote{We note that the coefficient estimator for the optimal-GR estimator with linear imputation functions admits a closed-form solution.} We find that this class of estimators has good finite-sample performance in our simulations. We detail the constructions below. Recall the definition of the intercept matrix $\onesmat$ in equation (\ref{intercept}) and the stack of imputations $f(\theta)$ defined before Section \ref{Section:Nonlinear}. The assumptions, asymptotic theory, and variance-bound estimation for the Opt-I estimators are left to Appendix \ref{more_results_S5} due to space constraints.
\begin{algorithm}[h]\label{OCI-GR}
\caption{Opt-I estimator for $\mu_{n,c}=c'\frac{1}{n}\onesmat'y$ }
\begin{algorithmic}
\Require QMLE coefficient estimates $\widehat{\theta}_n$ from Algorithm \ref{QMLE-GR}
\begin{enumerate}
    \item Define $\widehat{X}=\begin{bmatrix}
        \onesmat &\vert f(\widehat{\theta}_n)
    \end{bmatrix}\in \mathbb{R}^{kn\times(k+1)}$.
    \item Define $\widehat{\beta}=[\widehat{\beta}^1,...,\widehat{\beta}^{k+1}]=\arg\min_{\beta\in \mathbb{R}^{k+1}}  \widehat{\mathcal{L}}^{\toc}_n(\beta)$, where
    \begin{align}
        \widehat{\mathcal{L}}^{\toc}_n(\beta)= & \allowdisplaybreaks \diag(\widehat{X}\beta)\dmat^c\R\bpiInv\diag(\widehat{X}\beta) - 2\diag(y)\bpiInv\R\dmat^c\diag(\widehat{X}\beta)
    \end{align}
    \item  Compute $\widehat{\mu}_{n,a}^{\toci}$, for $a\in [k]$, where
    {\small
    \begin{equation}\label{DR}
        \widehat{\mu}_{n,a}^{\toci} = \frac{1}{n}\sum_{i=1}^n \left(\widehat{\beta}^a + \widehat{\beta}^{k+1}f^a(x_i,\widehat{\theta}_n)\right) + \frac{1}{n}\sum_{i=1}^n\frac{\R_{ai}}{\bpi_{ai}}  \left(y_{ai}-\widehat{\beta}^a - \widehat{\beta}^{k+1}f^a(x_i,\widehat{\theta}_n)\right).
    \end{equation}}
    \item Output $\widehat{\mu}_{n,c}^{\toci}=\sum_{a=1}^k c_a\widehat{\mu}_{n,a}^{\toci}$. 
\end{enumerate}
\end{algorithmic}
\end{algorithm}

\bibliographystyle{ecta-fullname} % Style BST file
\bibliography{ref}  % Bibliography file (usually '*.bib')

%% Or include bibliography directly:
%%%%%%%%%%%%%%%%%%%%%%%%%%%%%%%%%%%%%%%%%%%%%%
%% Example with multiple Appendixes:        %%
%%%%%%%%%%%%%%%%%%%%%%%%%%%%%%%%%%%%%%%%%%%%%%
\newpage 
\begin{appendix}
\section{Data Application}\label{Section:Simulation}
We now demonstrate how the results from earlier sections can be applied in practice. We study a network experiment based on the data in \cite{cai2015social}. Section \ref{BackgroundandDataset} describes the background and the dataset. Section \ref{complexity} uses the comlexity metric $\sigma_{\max}\left(\Ome^c\right)$ to understand the strengths and weaknesses of different designs in this setting. Section \ref{subsecion_sim} introduces simulation designs and discusses simulation results.
\subsection{Background and Dataset}\label{BackgroundandDataset}
\indent \cite{cai2015social} examines how social networks influence weather insurance adoption in rural China. The unit of the experiment is a household and the primary outcome of interest is weather insurance adoption, a binary variable indicating whether a household purchases insurance after an information session. Some household characteristics are observed, including demographics, rice production, income, and past experiences with natural disasters. The social network information is collected through a friend-nomination survey.\footnote{For more detailed information, see Section II-B of \cite{cai2015social}.}
%, in which households are asked  "to list five close friends, either within or outside the village, with whom they most frequently discussed rice production or financial issues".\footnote{For more detailed information, please see Section II-B of \cite{cai2015social}.}

Households were randomly assigned to two rounds of information sessions. In each round, a simple information session and an intensive information session were held simultaneously. The two rounds of information sessions were three days apart.
Households were randomized into four treatment arms: first-round simple, first-round intensive, second-round simple, and second-round intensive. The network effects on the insurance take-ups are measured by the average differences in insurance purchase decisions among second-round participants with different numbers of friends who were invited to the first round. 

%The simple session took about 20 minutes with an introduction to the insurance contract. The intensive session took about 45 minutes and offered additional explanation on how insurance works and its expected benefits. 
%\footnote{The paper's experimental design features two household-level randomizations and two village-level randomizations. The two household-level randomizations are designed to study network effects and peer effects, and the two village-level randomizations are designed to evaluate the monetary value of the social network effects and to generate exogenous variations in the insurance take-up rates. For the purpose of this paper, we focus on the household-level randomization that studies network effects. See Section A, Figure 1.1 and Figure 1.2 in \cite{cai2015social} for detailed information on the experimental design. } 

The experimental design is a village-level stratified randomization. Households were stratified according to household sizes and areas of rice production per capita. The computer code used for randomization in \cite{cai2015social} is not immediately available.\footnote{A brief description of the randomization procedure is given in footnote 8 of \cite{cai2015social}.} As a result, we decide to construct the randomization procedures based on the description of the paper and the details can be found in Appendix \ref{A:treatment_assignments}. Any implications we draw below should not be related to the original paper.
%\footnote{The data used in the simulations below are available at \url{https://www.openicpsr.org/openicpsr/project/113593/version/V1/view}.} 

We briefly describe the social networks used in the data. The friend nomination graph is a directed graph with edges pointing from the nominators to the nominees. The network has 4806 households and 41 disjoint weak components.The average outdegree is 3.5, the maximum indegree is 18, and the average within-village path length is 2.7. We plot the second largest component of the network in Figure \ref{fig:network}. The network has clear community structures, as most friendship ties form within natural villages. 
{\small 
\begin{figure}
    \caption{The Second Largest Component of the Network in the Simulation}
    \centering
    \includegraphics[scale=0.5]{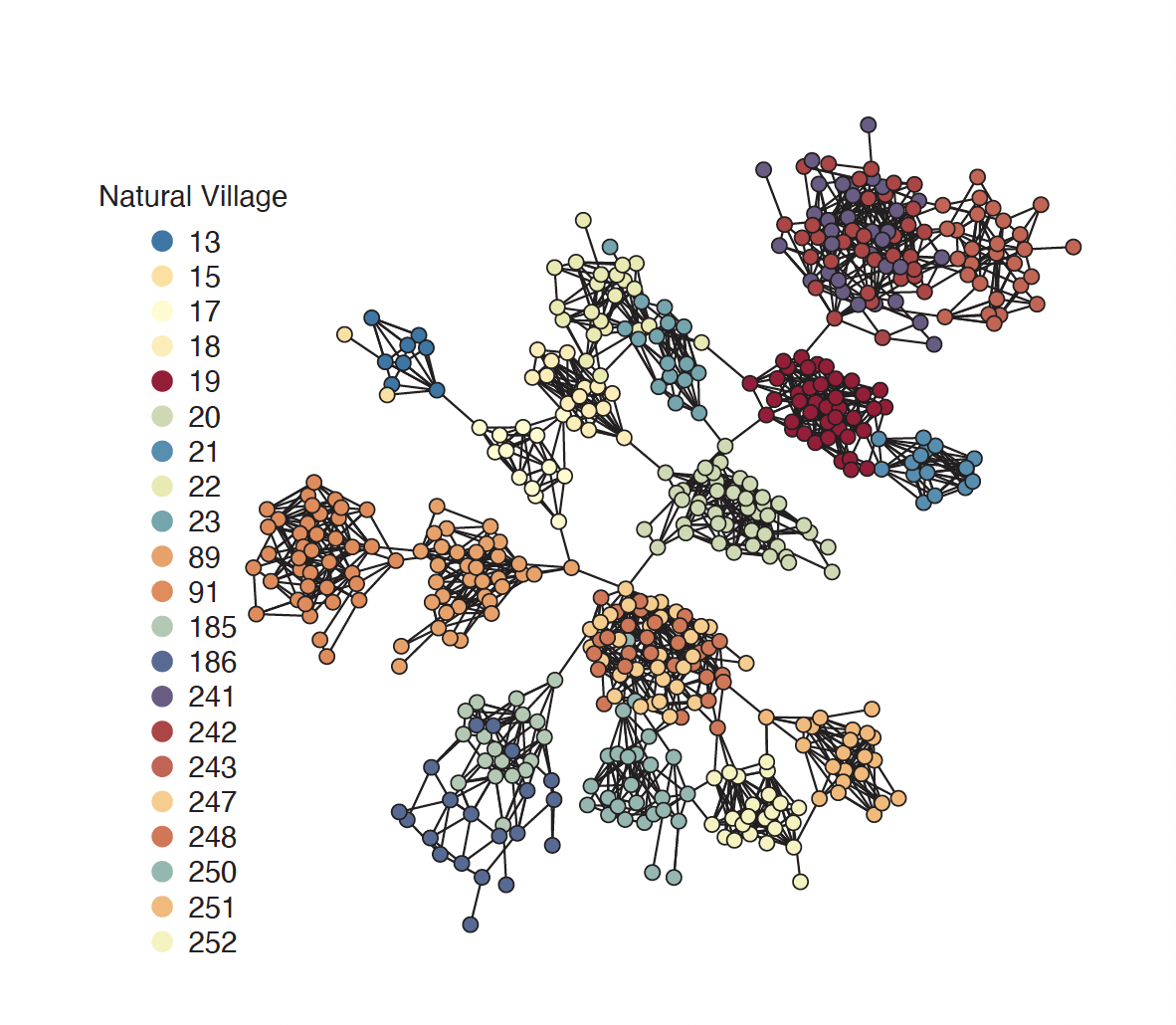}
    \label{fig:network}\\
    \legend{Figure \ref{fig:network} plots the second largest component used in the simulation. Each node represents a household. Two nodes are connected if one of the pair nominates another as a friend. Natural villages are groups of households and are the basis of randomization.}
\end{figure}
}
%\footnote{We note that this network is a subset of the complete network in the dataset. For our simulation, we dropped households with missing household size and rice production area information, the variables used for stratification. Our network statistics are consistent with those reported in  \cite{cai2015social}.}
%\footnote{A component in a directed graph is called a weak component if it is a connected component in the directed graph with directed edges replaced by non-directed ones.} 
%\footnote{The largest component is too large to fit on a page to show finer details.}  
There are 12 exposure mappings considered in the paper. We include a subset of them in Table \ref{exposure_mapping}. Exposures 8–12 are defined analogously to exposures 3–7, replacing the SRS with the SRI. Motivated by Column (5) in Table 2 of \cite{cai2015social}, we focus on the comparisons between exposures 3v4, 3v5, 3v6, 4v5 and 4v6, which we treat as the key parameters of interest to measure the effects of social networks. 
\begin{table}
\label{exposure_mapping}
    \caption{A List of Some Exposure Mappings Considered in \cite{cai2015social}}
    \begin{tabular}{c c c c  } \hline
    Index & Description & Index & Description\\
    \hline
    1   &  In FRS   &  5 & In SRS, with  friends in FRI\\
    2   & In FRI   & 6 & In SRS, with two friends in FRI   \\
    3  & In SRS, with no friends FRS or FRI  &   7 & In SRS, with more than two friends in FRI\\
 4  & In SRS, with friends only in FRS\\
    \hline
    \end{tabular}
\legend{ The abbreviations stand for the treatment arms. FRS: First Round Simple, FRI: First Round Intensive, SRS: Second Round Simple, and SRI: Second Round Intensive.}
\end{table}
\subsection{Use $\sigma_{\max}\left(\Omega^c\right)$ to understand different designs}\label{complexity}
In this section, we compare three designs using our proposed measure $\sigma_{\max}\left(\Omega^c\right)$ from Section \ref{Section:Estimation}. The purpose of this demonstration is to show how to use the  measure to understand the relative strengths and weaknesses of various designs. 
We consider multiple experimental designs: A) a finely-stratified village-level randomization, B) a village-level randomization, and C) Bernoulli designs with varying assignment probabilities.

The finely-stratified village-level randomization (Design A) is inspired by the experimental design in \cite{cai2015social}. We partition households in each natural village into four groups based on their household sizes and rice production areas. Unis within each stratum are completely randomized into four treatment arms. For details, see Appendix \ref{A:treatment_assignments}.
For the village-level randomization (Design B), we randomly assign households within each village to four arms with proportions (0.1,0.1,0.4,0.4). For Bernoulli designs, we consider different assignment probabilities to four treatment arms: (1/4,1/4,1/4,1/4), (1/5,1/5,3/10,3/10), (1/6,1/6,4/12,4/12) and (1/9,2/9,4/12,4/12). We label them Designs C.1, C.2, C.3 and C.4, respectively.

Table \ref{TableD} below documents $\sigma_{\max}\left(\Omega^c\right)$ values for different designs. We make two comments about the tables:
\begin{enumerate}
    \item We first observe that many $\infty$ symbols appear in the table. These arise in cases where some units have zero assignment probability to certain exposures. For instance, under the finely-stratified village-level randomization (Design A), 34 units have zero probability of being assigned to exposure 3, and 1,067 units have probabilities less than 0.01. These zeros and near-zeros are driven by small strata: if a unit and all its friends belong to a small stratum, it may never encounter a situation where none of its friends are assigned to the first round.\footnote{For example, if a unit and its five friends form a single stratum and at least one of them is always assigned to the first round, then the unit can never be in a state where none of its friends are treated in the first round.} A similar pattern occurs for exposure 6 under the same design, as well as under the village-level randomization (Design B). OLS estimators under these designs may fail to capture meaningful causal effects, and inverse-probability weighted estimators, even excluding units with zero assignment probabilities, can suffer from high variance.\footnote{For example, for Design A, if we exclude units with zero assignment probabilities, $\sigma_{\max}\left(\Ome^c\right)=924.96$ when comparing exposures 3 vs.4, $\sigma_{\max}\left(\Ome^c\right)=113.25$ when comparing exposures 4 vs.5, and $\sigma_{\max}\left(\Ome^c\right)=224.35$ when comparing exposures 4 vs.6. }
    \item Secondly, examining the columns for the Bernoulli designs (Designs C.1–C.3), we observe an inherent tradeoff: assigning units to the first round with high probabilities results in few units under exposure 3, which may compromise the statistical power for comparisons such as 3 vs. 4, 3 vs. 5, and 3 vs. 6; conversely, using low assignment probabilities leads to fewer units in exposure 6, potentially limiting our ability to detect nonlinear network effects. Based on these observations, we may fine-tune the first-round assignment probabilities to ensure that only a small share of units are assigned to the first round overall, with relatively more units assigned to the first-round intensive treatment than to the first-round simple treatment, arriving at Design C.4. If all exposure comparisons in Table \ref{TableD} are of primary interest and the ATEs are expected to be similar in magnitude, the proposed measure $\sigma_{\max}\left(\Ome^c\right)$ would recommend selecting Design C.4, among the six designs in the table.
\end{enumerate}
\begin{center}
\begin{table}
    \caption{Comparing $\sigma_{\max}\left(\Ome^c\right)$ for different designs}\label{TableD}
    \begin{tabular}{c |c c c c c c}
    \hline 
    \diagbox{Exposure Comparisons}{Designs}& A& B &C.1 &C.2 &C.3 &C.4 \\
     \hline
    Exposure 3 vs. Exposure 4     & $\infty$ & 43.16 &  314.84 & 110.32 & 59.69 & 59.80 \\
    Exposure 3 vs. Exposure 5  & $\infty$ & 42.43 & 314.94&  110.43 & 59.82 & 59.69 \\
     Exposure 3 vs. Exposure 6  & $\infty$ & $\infty$ &  315.57 & 110.91 & 146.16 & 83.07 \\
    Exposure 4 vs. Exposure 5 & $\infty$ & 46.02 &  38.74 & 31.14 & 32.10 & 48.51 \\
   Exposure 4 vs. Exposure 6 & $\infty$ & $\infty$ &  81.69 &  110.29 & 146.29 & 83.17\\
    \hline
    \end{tabular}
    \legend{This table reports $\sigma_{\max}\left(\Ome^c\right)$ for comparing various designs. Design A is a finely-stratified village-level randomization, Design B is a village-level randomization, and Design C.1-Design C.4 are Bernoulli designs with probabilities (1/4,1/4,1/4,1/4), (1/5,1/5,3/10,3/10), (1/6,1/6,1/3,1/3), and (1/9,2/9,1/3,1/3), respectively.  }
\end{table}
\end{center}
\subsection{Simulation Design}\label{subsecion_sim}
We report simulation results for comparing exposures 3 vs.4 and comparing exposures 3 vs.6 using Design C.2 and Design C.4. The goal of this simulation is to illustrate implications of the complexity metric $\sigma_{\max}\left(\Ome^c\right)$ as well as the behavior of the various estimators. 

We compare the following 7 estimators: 1) IPW estimator, 2) inverse-probability weighted WLS estimator (WLS), 3) QMLE-GR estimator with a logit model (Logit),\footnote{We set $\omega_{ai}=\pi_{ai}$ where $\omega_{ai}$ is defined in equation (\ref{M-pop}). This is to mimic the exercise where researchers estimate a logit model without any weighting.}, 4) Opt-GR estimator with a linear model (Opt Linear), 5) Opt-GR estimator with a logit model (Opt Logit), 6) Opt-I GR estimator with imputations using the WLS model (Opt-I WLS), and 7) Opt-I GR estimator with imputations using the logit model (Opt-I Logit). All adjustment models have a separate intercept for each arm and the same coefficients on the covariates. 

For each comparison of two exposures, we impute the potential outcomes in two ways. In the first simulation scenario, we impute the potential outcomes using a logistic model with coefficients estimated from the data. In this scenario, barring finite-sample issues, QMLE-GR, Opt-GR, and Opt-I GR estimators are expected to work similarly well and show improvement over the baseline IPW estimator. In the second simulation scenario, we impute the potential outcomes such that the Opt-GR estimators will have efficiency gains over the QMLE-GR estimators. This scenario is used to demonstrate the theoretical guarantee for the Opt-GR and Opt-I GR estimators.  We refer to the first simulation scenario as \textit{Sim-Impute} and the second simulation scenario as \textit{Sim-Optimal}. More details for data constructions, imputations, and implementations can be found in Appendix \ref{AppendixSimulationDetails}.
\subsection{Simulation Results}
Formal and complete simulation results are included in Appendix \ref{AppendixSimulationResults}. We high a subset of the results in Table \ref{WLS_MSE} and Table \ref{OPT_MSE} for discussion:
\begin{enumerate}
    \item When one compares across all tables, the variances of estimators generally increase as $\sigma_{\max}\left(\Ome^c\right)$ increases. For example, comparing the sample-size normalized MSE of WLS estimators in Table \ref{WLS_MSE} in the Sim-Impute Scenario, one finds a larger $\sigma_{\max}\left(\Ome^c\right)$ generically translates to a larger MSE for the WLS estimators.   
\begin{table}
    \caption{Normalized Mean Squared Errors of WLS Estimators in the Sim-Impute Scenario }
    \centering
    \begin{tabular}{cc cc c}
    Exposure Comparisons & \multicolumn{2}{c}{3 vs.4} & \multicolumn{2}{c}{3 vs.6}\\
    \hline
    Design& C.2 & C.4 & C.2 & C.4 \\
    $\sigma_{\max}\left(\Ome^c\right)$ & 110.32& 59.80 &110.91  &48.51 \\
    \hline 
    MSE (WLS)       & 8.95  &7.81 & 13.88 & 9.13 \\
    \end{tabular}
    \label{WLS_MSE}
    \legend{The numbers are selected from Table \ref{Table:34DesignC2Impute}, \ref{Table:36DesignC2Impute}, \ref{Table:34DesignC4Impute}, and \ref{Table:36DesignC4Impute}. All MSEs are multiplied by the sample sizes.}
\end{table}
\item  The Opt-GR estimators (Opt Linear and Opt Logit) bring variance reductions but also face bias-variance trade-offs in the finite sample. For example, in Table \ref{OPT_MSE}, the Opt-GR linear estimator has a variance 10$\sim$15\% lower than that of the WLS estimator. However, the Opt-GR estimators can also incur a finite sample bias. The practical performance of Opt-GR estimators tends to become worse as $\sigma_{\max}\left(\Ome^c\right)$ gets larger.

\item In our simulations, the Opt-I WLS and Opt-I Logit estimators perform reasonably well in all cases. The Opt-I WLS and Opt-I Logit estimators are less efficient compared with the Opt Linear and Opt Logit estimators in terms of theoretical asymptotic variance, but the loss of efficiency appears to be small and the Opt-I OLS and Opt-I Logit estimators have better finite sample performance. Taken together and based on our simulation results, we consider the Opt-I OLS and Opt-I Logit as viable alternatives for Opt Linear and Opt Logit in many practical settings.

\end{enumerate} 
\begin{table}
    \caption{Normalized Biases and Variances of WLS, OPT-GR (Linear) and OPT-I GR (Linear) Estimators in the Sim-Optimal Scenario }
    \centering
    \begin{tabular}{cc cc c| c cc cc }
    Comparisons & \multicolumn{2}{c}{3 vs.4} & \multicolumn{2}{c|}{3 vs.6}  & \multicolumn{2}{c}{3 vs.4} & \multicolumn{2}{c}{3 vs.6}\\
    \hline
    Design& C.2 & C.4 & C.2 & C.4 &    Design& C.2 & C.4 & C.2 & C.4 \\
    $\sigma_{\max}\left(\Ome^c\right)$ & 110.32& 59.80 &110.91  &48.51   
    &$\sigma_{\max}\left(\Ome^c\right)$ & 110.32& 59.80 &110.91  &48.51   \\
    \hline
   $\textrm{Bias}^2$ (WLS) & 0.06 &0.01  & 0.02  &0.03  &Var.(WLS)  & 13.27 & 17.66 & 10.87 & 12.52  \\
    $\textrm{Bias}^2$ (Opt-GR) & 1.45 &0.01  &0.00  &0.48  &Var.(Opt-GR)  & 11.38 & 15.65 & 9.85 & 11.32  \\
 $\textrm{Bias}^2$ (Opt-I GR) & 0.20 &0.01  &0.00  &0.00  &Var.(Opt-I GR)  & 11.19 & 15.47 & 10.25 & 11.29  \\
    \end{tabular}
        \label{OPT_MSE}
    \legend{The numbers are selected from Table \ref{Table:34DesignC2Optimal}, \ref{Table:36DesignC2Optimal}, \ref{Table:34DesignC4Optimal}, and \ref{Table:36DesignC4Optimal}. Squared biases and variances are multiplied by the sample sizes.}
\end{table}
\section{Additional Results in Section \ref{Section:Estimation} }\label{other_estimators}
We define a few more estimators.
\begin{example}[Hajek (HA) estimator]
\begin{equation}
    \hat{\mu}^{\thj}_n = \left(\onesmat'\bpiInv \R \onesmat\right)^{-1} \onesmat' \bpiInv \R y. 
\end{equation}
Its probability target is $\frac{1}{n}\onesmat' y$.
\end{example}
\begin{example}[Completely Imputed (CI) estimators]
\begin{equation}
    \hat{\mu}^{\tci}_n=\frac{1}{n} \onesmat'\X \left(\X' \Om \R \X \right)^{+} \X'\Om \R y  .
\end{equation}
Its probability target is $\frac{1}{n}\onesmat' \X b^{\twls}_n$.
\end{example}
\begin{example}[Missing Imputed (MI) estimators]
\begin{equation}
     \mu^{\tmi}_n=\frac{1}{n} \onesmat'\left(\I_{kn} - \left(\R-\I_{kn} \right) \X \left(\X' \Om \R \X \right)^{+} \X'\Om  \right) \R y.
\end{equation}
Its probability target is $\frac{1}{n}\onesmat' \bpi y +\frac{1}{n}\onesmat'(\I_{kn}-\bpi)\X b^{\twls}_n$. 
\end{example}

The following lemma gives conditions under which the CI and MI estimators are consistent estimators of the average potential outcomes. These results are stated in an unpublished work \cite{middleton2021unifying3}. See also \cite{brewer1979class} and \cite{wright1983finite}. 
\begin{lemma}\label{lemma:CIMI}
\begin{enumerate}[label=(\roman*)]
    \item If the columns of the matrix $\onesmat \in \mathbb{R}^{kn\times k}$ are in the column space of the matrix $\pi\Om\X$, $\nu_n^{\tci}=\frac{1}{n}\onesmat'y$.
    \item If the columns of the matrix $\onesmat \in \mathbb{R}^{kn\times k}$ are in the column space of the matrix $\left(\bpi^{-1}-I_{kn}\right)^{-1}\Om\X$, $\nu_n^{\tmi}=\frac{1}{n}\onesmat'y$.
\end{enumerate}
\end{lemma}
\begin{proof}
   We have $\nu^{\tci}_n=\frac{1}{n}\onesmat \X b^{\twls}_n$. The difference from average potential outcomes is:
    \begin{equation}\label{CIConditionA}
      \frac{1}{n}\onesmat' y- \frac{1}{n}\onesmat' \X b^{\twls}_n.
    \end{equation}
   Note,
    $$\onesmat' (y-\X b^{\twls}_n)=\onesmat'\bpiInv\Om^{-1}\Om\bpi(y-\X b^{\twls}_n)$$
    By the definition of WLS,    
    \begin{equation*}
        \X'\Om\bpi(y-\X b^{\twls}_n)=0.
    \end{equation*}
    Thus equation (\ref{CIConditionA}) is 0 if each column of $\Om^{-1}\bpiInv\onesmat$ is in the column space of $\X$, which is the stated condition after rearrangement.

        The MI estimators converge to $\frac{1}{n}\onesmat' \bpi y +\frac{1}{n}\onesmat'(\I_{kn}-\bpi)\X b^{\twls}_n$. The difference from average potential outcomes is:
    \begin{align}\label{MICondition1}
    \begin{split}
       \frac{1}{n}\onesmat' y- \frac{1}{n}\onesmat' \pi y -\frac{1}{n}\onesmat'(\I_{kn}-\bpi)\X b^{\twls}_n   =  \frac{1}{n}\onesmat'(\I_{kn}-\bpi)(y-\X b^{\twls}_n ). 
    \end{split}
    \end{align}
   Note,
   \begin{align}
    & \onesmat'(\I_{kn}-\bpi) (y-\X b^{\twls}_n)=\onesmat'(\I_{kn}-\bpi)\bpiInv\Om^{-1}\Om\bpi(y-\X b^{\twls}_n) \\
    =&  \onesmat'(\bpiInv-\I_{kn})\Om^{-1}\Om\bpi(y-\X b^{\twls}_n).
   \end{align}

    By the definition of WLS, the quantity  (\ref{MICondition1}) is 0 if each column of $\Om^{-1}(\bpiInv-\I_{kn})\onesmat$ is in the column space of $\X$, which is the stated condition after rearrangement.
\end{proof}

The following lemma gives conditions under which a WLS estimator is algebraically equivalent to a GR estimator:

\begin{lemma}\label{lemma:WLS_GR}
If $\frac{1}{n}\sum_{i=1}^n x_{si}=0$ for all $s\in[p]$ and columns of the matrix $\R\bpiInv\onesmat \in \mathbb{R}^{kn\times k}$ are in the column space of the matrix $\R\Om\X$, $\widehat{\nu}_n^{\twls}=\widehat{\nu}_n^{\tgr}$.
In particular, $\widehat{\nu}_n^{\twls}=\widehat{\nu}_n^{\tgr}$ if $\frac{1}{n}\sum_{i=1}^n x_{si}=0$ for all $s\in[p]$ and $\Om=\bpiInv$.
\end{lemma}
\begin{proof}
Notice that if $\frac{1}{n}\sum_{i=1}^n x_{si}=0$ for all $s\in [p]$, we have $ \frac{1}{n}\onesmat'X =  \begin{bmatrix}
 \I_k \vert \0_{k\times p}
    \end{bmatrix}$. Hence $ \begin{bmatrix}
\I_k \vert \0_{k\times p}
\end{bmatrix}\widehat{b}^{\twls}=\onesmat'X \widehat{b}^{\twls}$. We then have:
\begin{equation}
\widehat{\nu}^{\twls}_n -\widehat{\nu}^{\tgr}_n = \frac{1}{n}\onesmat'\bpiInv\R\left(y-X \widehat{b}^{\twls}\right). 
\end{equation}
By the definition of WLS,
\begin{equation}
    X'\Om \R \left(y-X\widehat{b}^{\twls}\right)=0
\end{equation}
Hence $\widehat{\nu}^{\twls}_n -\widehat{\nu}^{\tgr}_n=0$ if $\R\bpiInv\onesmat$ is in the column space of $\R\Om X$. 
\end{proof}

\section{Additional Results in Section \ref{Section:Nonlinear}}
This section contains additional results for the Opt-I estimators. Define $X=\begin{bmatrix}
     \onesmat & \vert f(\theta_n)
\end{bmatrix}\in\mathbb{R}^{kn\times(k+1)}$ and $\dmat^c= \diag(c'\onesmat)\dmat \diag(c'\onesmat)\in \mathbb{R}^{kn\times kn}$.
\begin{assumption}\label{A:OCI} There exists a positive integer $N$ and a positive constant $c_{\ref{A:OCI},1}$ such that
    \begin{equation}
        \frac{1}{n}X'\dmat^c X \succeq c_{\ref{A:OCI},1}\ones{k+1},
    \end{equation}
    uniformly for all $n\geq N$.
\end{assumption}
We define:
\begin{equation}
 \beta^{\toci}_c =\begin{bmatrix}
    \beta^1,..., \beta^{k+1}
\end{bmatrix}=  \left(X'\dmat^c X\right)^{-1}X'\dmat^c y,   
\end{equation}
\begin{equation}
    \hat{\mu}_{n,a}^{\toci,\tL}=\frac{1}{n}\sum_{i=1}^n (\beta^a+\beta^{k+1}f^a(x_i,\theta_n)) + \frac{1}{n}\sum_{i=1}^n\frac{\R_{ai}}{\bpi_{ai}}  (y_{ai}-\beta^a-\beta^{k+1}f^a(x_i,\theta_n))
\end{equation}
 for $a\in[k]$ and $\theta_n\equiv\arg\min_{\theta\in\Theta} \mathcal{L}_n(\theta)$ defined in (\ref{M-pop}). Define $\hat{\mu}_{n,c}^{\toci,\tL}=\sum_{a=1}^k c_a\hat{\mu}_{n,a}^{\toci,\tL}$. The variance of $\hat{\mu}_{n,c}^{\toci,\tL}$ can be expressed as
\begin{equation*}
    \V(\hat{\mu}_{n,c}^{\toci,\tL}) =\frac{1}{n^2} c'\onesmat'\diag(y-X^{\toci}\beta)\hspace{1pt}\dmat\hspace{1pt} \diag(y-X^{\toci}\beta)\onesmat c.
\end{equation*}
\begin{theorem}\label{Thm:OCI}
Define $\hat{\mu}_{n,a}^{\toci}$, $\hat{\mu}_{n,c}^{\toci}$ and $\widehat{\beta}^{\toci}_c$ as in Algorithm \ref{OCI-GR}.  Under Assumptions \ref{A:BoundedFourthMoments}, \ref{A:OCI}, \ref{A:VanillaConsistency},  and \ref{A:Imputation}, and if $\viiii{\dmat}_1=O(1)$, $\sigma_{\max}((\widetilde{\dmat}\otimes \widetilde{\dmat})\circ \s)=o(n)$, $\sigma_{\max}\left(\dtildep\right)=O(1)$ and there exists a positive constant $c_{\ref{Thm:OCI},2}$ such that $n\V(\widehat{\mu}_{n,c}^{\toci,\tL})\geq c_{\ref{Thm:OCI},2}$ uniformly for all large $n$, then following results hold:
\begin{enumerate}[label=(\roman*)]
   \item We have $\left(\widehat{\mu}^{\toci}_{n,c}-\hat{\mu}_{n,c}^{\toci,\tL}\right)/\sqrt{\V(\widehat{\mu}_{n,c}^{\toci,\tL})}=o_p\left(1\right)$. 
   \item We have:
   \begin{equation}
    \V(\hat{\mu}_{n,c}^{\toci,\tL}) \leq \V(\hat{\mu}_{n,c}^{\tqmle}) =\frac{1}{n^2} c'\onesmat'\diag\left(y-f(\theta_n)\right)\hspace{1pt}\dmat\hspace{1pt} \diag\left(y-f(\theta_n)\right)\onesmat c.   
   \end{equation}
\item Define the variance bound 
\begin{equation}
  \tilde{\V}(\hat{\mu}_{n,c}^{\toci,\tL})= \frac{1}{n^2}c'\onesmat'\diag\left(y-X^{\toci}\beta^{\toci}_c\right)\tilde{\dmat}\diag\left(y-X^{\toci}\beta^{\toci}_c\right)\onesmat c\in \mathbb{R} 
\end{equation}
with an identified variance bound matrix $\tilde{\dmat}$. The plug-in variance-bound estimator
\begin{equation}
  \tilde{\V}(\hat{\mu}_{n,c}^{\toci,\tL})= \frac{1}{n^2}c'\onesmat'\diag\left(y-X^{\toci}\widehat{\beta}^{\toci}_c\right)\R\dtildep\R\diag\left(y-X^{\toci}\widehat{\beta}^{\toci}_c\right)\onesmat c\in \mathbb{R} 
\end{equation}
is consistent: $\widehat{\widetilde{\V}}(\hat{\mu}_{n,c}^{\toci,\tL})/\widetilde{\V}(\hat{\mu}_{n,c}^{\toci,\tL})\overset{p}{\to} 1$.
\item If there exists a continuous function $q:(0,1)\to\mathbb{R}$ such that $\lim\sup_{n\to\infty}\textrm{P}_n\left(\left|\widehat{\mu}_{n,c}^{\toci}-\mu_{n,c}\right|\geq q(\alpha)\sqrt{\V(\widehat{\mu}_{n,c}^{\toci,\tL})}\right)\leq \alpha$ for all $\alpha\in (0,1)$, then, $\lim\sup_{n\to\infty}\textrm{P}_n\left(\left|\widehat{\mu}_{n,c}^{\toci}-\mu_{n,c}\right|\geq q\left(\alpha\right)\sqrt{  \widehat{\widetilde{\V}}(\hat{\mu}_{n,c}^{\toci,\tL})}\right)\leq \alpha $. 
\end{enumerate}
\end{theorem} 

\section{Additional Results for the Network Experiment in Section \ref{Section:NetworkExperiments} }\label{network_experiments}
We provide lower-level conditions for the network experiments presented in Section \ref{Section:NetworkExperiments}. These conditions lead to the properties $\sigma_{\max}\left(\Ome\right)=O(1)$, $\sigma_{\max}((\widetilde{\dmat}\otimes \widetilde{\dmat})\circ \s)=o(n)$ and  $\sigma_{\max}\left(\dtildep\right)=O(1)$, which are required by Corollary \ref{C:Consistency} for $\sqrt{n}$-estimation and Corollary \ref{C:ConsistentVariance} for consistent variance bound estimation. 

\begin{assumption}\label{A:networkprobability} For all $a,b\in[k]$ and for all large $n$, there exists a positive $c_{\ref{A:networkprobability},1}\in(0,1)$, $\pi_i(a)>c_{\ref{A:networkprobability},1}>0$ for all $i\in [n]$. For a positive $c_{\ref{A:networkprobability},2}\in(0,1)$, $\pi_{ij}(a,b)>c_{\ref{A:networkprobability},2}>0$ for all $i\in[n]$ whenever $\pi_{ij}(k,l)\not =0$.
\end{assumption}
\begin{assumption}\label{A:dependencegraph}
The dependency graph of the random assignment vector $\R\ones{kn}$ has bounded degrees uniformly in $n$.
\end{assumption}
One can verify that the Assumptions \ref{A:networkprobability} and \ref{A:dependencegraph} are satisfied for the experimental design considered in Example \ref{Example:Network} when the network has a bounded degree uniformly in $n$.\footnote{ In some social network settings, there may exist nodes of very large degrees. This may cause two problems: 1) the assignment probability becomes too small, and/or 2) the dependence between exposure maps becomes too strong. In these cases, one may wish to restrict the parameters of interest to subgroups (e.g., people with low degrees) and consider specific exposure mappings for precise estimates (e.g., the sample average treatment effects condition on the event that all the high-degree nodes are assigned to treatment). }

These two assumptions are standard in the literature. Under the two assumptions, both the first-order design matrix $\dmat$ and the AS bound matrix $\tilde{\dmat}^{\tas}$ have bounded $l_1$-induced norms, because each column contains only finitely many nonzero entries with bounded magnitudes uniformly in $n$. The fact that $\sigma_{\max}((\widetilde{\dmat}\otimes \widetilde{\dmat})\circ \s)=o(n)$ and  $\sigma_{\max}\left(\dtildep\right)=O(1)$  are also satisfied by Proposition 6.2 in \cite{aronow2017estimating} and by Lemma \ref{LemmA:VanillaRootn}.

Recall that $p$ is the dimension of the pretreatment covariates. We consider the case of a same-slope adjustment and consider two types of imputation functions: linear models $f^a(x_i,\theta)=\gamma^a+x_i'\beta$, $a\in[k]$ and logistic models $f^a(x_i,\theta)=\frac{\exp(\gamma^a+x_i'\beta)}{1+\exp(\gamma^a+x_i'\beta)}$, $a\in [k]$.  For the QMLE-GR estimator with the linear model, the finite population criterion and the sample equivalent are 
\begin{equation*}
   \mathcal{L}_n^{\textnormal{ln}}(\theta)= \frac{1}{n}\sum_{a=1}^k\sum_{i=1}^n (y_{ai}-\gamma^a-x_i'\beta)^2, \hspace{4mm}   \widehat{\mathcal{L}}_n^{\textnormal{ln}}(\theta)= -\frac{1}{n}\sum_{a=1}^k\sum_{i=1}^n \frac{\R_{ai}}{\pi_{ai}} (y_{ai}-\gamma^a-x_i'\beta)^2,
\end{equation*}
where $\theta=(\{\gamma^a\}_{a=1}^k,\beta)\in\mathbb{R}^{k+p}$.
For the QMLE-GR estimator with the logistic regression model, they are
\begin{equation*}
 \mathcal{L}_n^{\textnormal{lg}}(\theta)=-\frac{1}{n}\sum_{a=1}^k\sum_{i=1}^n \left[y_{ai}(\gamma^a+x_i'\beta) - \log(1+\exp{(\gamma^a+x_i'\beta)})\right]
\end{equation*}
and
\begin{equation*}
\widehat{\mathcal{L}}_n^{\textnormal{lg}}(\theta)=-\frac{1}{n}\sum_{a=1}^k\sum_{i=1}^n \frac{\R_{ai}}{\pi_{ai}}\left[y_{ai}(\gamma^a+x_i'\beta) - \log(1+\exp{(\gamma^a+x_i'\beta)})\right],
\end{equation*}
where $\theta=(\{\gamma^a\}_{a=1}^k,\beta)\in\mathbb{R}^{k+p}$. Note we can also define 
\begin{equation*}
 \mathcal{L}_n^{\textnormal{lg},2}(\theta)=-\frac{1}{n}\sum_{a=1}^k\sum_{i=1}^n \pi_{ai} \left[y_{ai}(\gamma^a+x_i'\beta) - \log(1+\exp{(\gamma^a+x_i'\beta)})\right]
\end{equation*}
and
\begin{equation*}
\widehat{\mathcal{L}}_n^{\textnormal{lg},2}(\theta)=-\frac{1}{n}\sum_{a=1}^k\sum_{i=1}^n \R_{ai}\left[y_{ai}(\gamma^a+x_i'\beta) - \log(1+\exp{(\gamma^a+x_i'\beta)})\right].
\end{equation*}
For the Opt-GR estimators, we form the criterion using the setup in Section \ref{Opt-GMM} with the first-order design matrix $\dmat$. We choose the AS bound for a bounding matrix. The plug-in variance bound estimators are constructed using the formulae in Theorems \ref{Thm:QMLE}, \ref{Thm:NOHARM}, and \ref{Thm:OC}. 
For the Opt-GR estimator with the logistic model, we need more moment assumptions in order to satisfy Assumption \ref{A:GMM2new}-(i)-(b).
\begin{assumption}[Bounded 8th moments]\label{A:Bounded8thmoment}
For all $n$ and $z_i\in\{x_{1i},...,x_{pi}\}$, $\frac{1}{n}\sum_{i=1}^n z_i^8<C_{\ref{A:Bounded8thmoment}}<\infty$, 
where $C_{\ref{A:Bounded8thmoment}}$ is a finite constant.
\end{assumption}
The following theorems specialize Theorems \ref{Thm:QMLE}, \ref{Thm:NOHARM}, and \ref{Thm:OC} in the network experiment setting. 
\begin{theorem}\label{Theorem:NetworkLinearAdj}
For linear models: 
\begin{enumerate}
    \item Under Assumptions \ref{A:BoundedFourthMoments}, \ref{A:Invertibility},  \ref{A:networkprobability}, and \ref{A:dependencegraph}, Theorem \ref{Thm:QMLE} holds for the QMLE-GR estimator. In addition, under Assumption \ref{A:NOHARM}, Theorem \ref{Thm:NOHARM} holds for the No-harm-GR estimator, and, under Assumption \ref{A:OCI}, Theorem \ref{Thm:OCI} holds for the Opt-I GR estimator.
    \item Under Assumptions \ref{A:BoundedFourthMoments}, \ref{A:GMM1}-(ii), \ref{A:networkprobability}, and \ref{A:dependencegraph}, Theorem \ref{Thm:OC} holds for the Opt-GR estimator.
\end{enumerate}
For logistic models: 
\begin{enumerate}
    \item Under Assumptions \ref{A:BoundedFourthMoments}, \ref{A:networkprobability}, \ref{A:dependencegraph}, and \ref{A:Bounded8thmoment} and \ref{A:VanillaConsistency}-(i), (ii) and (ix), Theorem \ref{Thm:QMLE} holds for the QMLE-GR estimator. In addition, under Assumption \ref{A:NOHARM}, Theorem \ref{Thm:NOHARM} holds for the No-harm-GR estimator, and, under Assumption \ref{A:OCI}, Theorem \ref{Thm:OCI} holds for the Opt-I GR estimator.
    \item Under Assumptions \ref{A:BoundedFourthMoments}, \ref{A:GMM1}, \ref{A:networkprobability}, \ref{A:dependencegraph}, \ref{A:Bounded8thmoment} and \ref{A:GMM2new}-(iv), Theorem \ref{Thm:OC} holds for the Opt-GR estimator.
\end{enumerate}
\end{theorem}
Proofs for this theorem can be found in Appendix \ref{proof_network_experiment}.
\begin{remark}
    With a CLT under local dependence \cite{chen2004normal}, one can construct a confidence interval using a normal approximation if the asymptotic variance $ \frac{1}{n} c'\onesmat'\diag(y-f(\theta_n))\hspace{1pt}\dmat\hspace{1pt} \diag(y-f(\theta_n))\onesmat c$ is uniformly bounded above and below by positive constants for large $n$. That is, there exist positive constants $c$ and $C$ such that the asymptotic variance is bounded between $c$ and $C$ uniformly for large $n$. The upper bound is satisfied under our assumptions. The lower bound is a typical assumption. If the assumption on the lower bound is violated, the estimator will converge to the true parameter with a faster-than-$\sqrt{n}$ rate, and the confidence interval using a normal approximation may not have the asymptotically correct coverage. 
\end{remark}
\section{Auxiliary Lemmas}
This section proves several auxiliary lemmas instrumental in the proofs below. Let $\|\cdot\|_2$ denote the Frobenius norm if applied to a matrix and the $l_2$ vector norm if applied to a vector.
\begin{lemma}{(Well-behaved WLS Design Matrix)}\label{LemmaDesign}
Under Assumption \ref{A:Invertibility}, and if there exists a positive $c$ such that $\lambda_{\min}(\Om\bpi)>c>0$ for all $n$, there exists an $n_0$ such that for all $n\geq n_0$, $\frac{1}{n}\X'\Om\pi\X\in\mathbb{R}^{(k+p)\times (k+p)}$ is invertible. Moreover, $\lambda_{\min}(\frac{1}{n}\X'\Om\bpi\X)$ is bounded away from 0 uniformly for $n\geq n_0$.
\end{lemma}
\begin{proof}
By Assumption \ref{A:Invertibility}, there exists an $n_0$ such that for all $n\geq n_0$, $\lambda_{\min}(\frac{1}{n}\X'\X)\geq c_{\ref{A:Invertibility}}$ for all $n\geq n_0$. Since $\Om$ and $\pi$ are positive diagonal matrices, $\frac{1}{n}\X'\Om\pi\X\in\mathbb{R}^{(k+p)\times (k+p)}$ is a positive semidefinte matrix. We only need to check its smallest eigenvalue. For any $t\in\mathbb{R}^{k+p}$,
\begin{align*}
    t'\frac{1}{n}\X'\Om\bpi\X t & \geq c \frac{1}{n}\|\X t\|_2^2 =c\times \left(t'\frac{1}{n}\X'\X t\right) \geq cc_{\ref{A:Invertibility}}\|t\|_2^2 >0,
\end{align*}
where the first inequality is justified by   $\lambda_{\min}(\Om\bpi)\geq c$. 
This inequality holds for all $n\geq n_0$, proving the statement. 
\end{proof}
\begin{lemma}{(Bounded WLS Coefficients)}\label{LemmaBoundedCoefs}
Let $b^{\twls}_n=(\X\Om\bpi \X)^{+}(\X\Om\bpi y)$. Under Assumptions \ref{A:BoundedFourthMoments} and \ref{A:Invertibility} and if  there exists a positive $C$ such that $\lambda_{\max}(\Om\bpi)<C$ for all $n$,
$\|b^{\twls}_n\|_2=O(1)$
\end{lemma}
\begin{proof}
We showed in Lemma \ref{LemmaDesign} that  $\X'\Om\bpi \X$ is invertible for large $n$, so we shall assume $b^{\twls}_n=(\X\Om\bpi \X)^{-1}(\X\Om\bpi y)$. For all large enough $n$,
\begin{align}\label{xybound}
   & \|\frac{1}{n}\X\Om\bpi y\|_2^2  =  \frac{1}{n^2}y'\Om\bpi\X'\X\Om\bpi y \leq \lambda_{\max}(\frac{1}{n}\X'\X) \times \frac{1}{n}\|\Om\pi y\|_2^2\\
    & \leq \lambda_{\max}(\frac{1}{n}\X'\X) \times \frac{C^2}{n}\|y\|_2^2\leq \|\frac{1}{n}\X'\X\|_2 \times\frac{C^2}{n}\|y\|_2^2 <\infty.    
\end{align}
by Assumptions \ref{A:BoundedFourthMoments} and \ref{A:Invertibility} and the fact that $\lambda_{\max}(\Om\bpi)<C$.
Then,
\begin{align}
    & b^{\twls'}_nb^{\twls}_n
     = (\frac{1}{n}\X'\Om\bpi y)'(\frac{1}{n}\X'\Om\bpi \X)^{-1}(\frac{1}{n}\X'\Om\bpi \X)^{-1}(\frac{1}{n}\X'\Om\bpi y)\\
    & \leq \left(\lambda_{\min}(\frac{1}{n}\X'\Om\pi\X)\right)^{-2}\times   \|\frac{1}{n}\X'\Om\bpi y\|_2^2 =O(1)
\end{align}
by Lemma \ref{LemmaDesign}.
\end{proof}
Let $z\in\mathbb{R}^{kn}$ be an arbitrary vector. We define the IPW estimator:
\begin{equation*}
    \widehat{\delta}^{\tht} = \frac{1}{n}\onesmat'\bpiInv \R z \in \mathbb{R}^k
\end{equation*}
\begin{lemma}[Rate of the HT Estimators]\label{HTcp1}
If $\frac{1}{n}\|z\|_2^2=O(1)$, $n\V\left(\widehat{\delta}^{\tht}\right)=O\left(\sigma_{\max}\left(\dmat\right)\right)$.
\end{lemma}
\begin{proof}
Notice we can write:
\begin{align*}
\begin{split}
    \widehat{\delta}^{\tht}  = \frac{1}{n}\onesmat'\bpiInv \R \diag(z)\ones{kn} = \frac{1}{n}\onesmat'\diag(z)\bpiInv \R \ones{kn}
\end{split}
\end{align*}
We have for any $t\in\mathbb{R}^k$
\begin{align}
\begin{split}
      &  n\V(t'\widehat{\delta}^{\tht}) =  n \left( t'\frac{1}{n}\onesmat'\diag(z)\V(\bpiInv \R \ones{kn}) \diag(z)\onesmat t\frac{1}{n}\right)\\
     & = \frac{1}{n}(t'\onesmat'\diag(z))\dmat(\diag(z)\onesmat t)\leq \sigma_{\max}\left(\dmat\right) \frac{1}{n}\|t'\onesmat'\diag(z)\|_2^2\\
     & = \sigma_{\max}\left(\dmat\right) \times \frac{1}{n}t'\onesmat'\diag(z)\diag(z)\onesmat t\\
     & \leq \sigma_{\max}\left(\dmat\right) \times \|t\|_2^2 \times \frac{1}{n} \lambda_{\max}(\onesmat'\diag(z)\diag(z)\onesmat )\\
     & \leq \sigma_{\max}\left(\dmat\right) \times \|t\|_2^2 \times \frac{1}{n} \Tr(\onesmat'\diag(z)\diag(z)\onesmat)\\
     & = \sigma_{\max}\left(\dmat\right) \times \|t\|_2^2 \times  \frac{1}{n} \|z\|_2^2  =  O\left( \|t\|_2^2 \sigma_{\max}\left(\dmat\right) \right),
\end{split}
\end{align}
where $\Tr$ is the trace operator of a matrix and the last equality holds by our assumption on $\frac{1}{n}\|z\|_2^2$. We conclude that the largest eigenvalue of the positive semidefinite matrix $\V(\widehat{\delta}^{HT})$ is of the order $O( \sigma_{\max}\left(\dmat\right)/n)$. $\|n\V(\widehat{\delta}^{\tht})\|_2\leq \sqrt{k}\lambda_{\max}\left(n\V(\widehat{\delta}^{\tht})\right) =O(\sigma_{\max}\left(\dmat\right))$, proving the statement.
\end{proof}

\begin{lemma}\label{LemmaWLS}
Under Assumptions \ref{A:BoundedFourthMoments} and \ref{A:Invertibility}, and if there exist positive $c$ and $C$ such that $0<c<\lambda_{\min}(\Om\bpi)<\lambda_{\max}(\Om\bpi)<C$ for all $n$, and if $\sigma_{\max}\left(\dmat\right)/n=o(1)$, $\widehat{b}^{\twls}_n-b^{\twls}_n=O_p(\sqrt{\sigma_{\max}\left(\dmat\right)}n^{-\frac{1}{2}}).$
\end{lemma}
\begin{proof}
    We first show the "numerator" vector $\frac{1}{n}\X'\Om\R y$ is consistent for $\frac{1}{n}\X'\Om\bpi y$. Notice $\frac{1}{n}\X'\Om\R y$ is an unbiased estimator for $\frac{1}{n}\X'\Om\bpi y$. We only need to show that the variance is of the order $O\left(\sigma_{\max}\left(\dmat\right)/n\right)$.
       Let $\X_i$ be the column vector created from the $i$th column of $\X$. Then the $i$th element of $\frac{1}{n} \X' \Om \R y $ can be written,    
    \begin{align*}
        & \{ \frac{1}{n} \X' \Om \R y \}_{i } = \frac{1}{n} \X_i' \Om \R y = \frac{1}{n} \ones{kn}' \diag(\X_i) \Om \R y =  \frac{1}{n} \ones{kn}' \R \Om \diag(\X_i)  y\\
        = & \ones{k}' \frac{1}{n}\onesmat' \bpiInv \R  \bpi \Om \diag(\X_i) y=  \ones{k}' \frac{1}{n}\onesmat' \bpiInv\R  \bpi \Om (\X_i\circ y).
    \end{align*}
  Under Assumption \ref{A:BoundedFourthMoments}, we have 
  \begin{align*}
      \frac{1}{n}\|\bpi\Om (\X_i\circ y)\|_2^2 & \leq C^2\times \frac{1}{n} \|(\X_i\circ y)\|_2^2\leq C^2 \sqrt{\frac{1}{n} \|\X_i\|_4^4} \times \sqrt{\frac{1}{n} \|y\|_4^4} <\infty,
  \end{align*}
   by the fact that $\lambda_{\max}\left(\bpi\Om\right)<C$, the Cauchy-Schwartz inequality and Assumption \ref{A:BoundedFourthMoments}. Thus by Lemma \ref{HTcp1}, we have 
   \begin{equation}
\sqrt{\frac{\sigma_{\max}\left(\dmat\right)}{n}}\| \frac{1}{n} \X' \Om \R y - \frac{1}{n}\X\Om\bpi y\|_2 = O_p\left(1\right).
   \end{equation}
  \indent Similarly, the $(i,j)$ element of the WLS "denominator" matrix, can be written as:
    \begin{align*}
        \{ \frac{1}{n} \X' \Om \R \X \}_{ij} = & \ones{k}' \frac{1}{n}\onesmat' \bpiInv\R  \bpi \Om (\X_i\circ \X_j).
    \end{align*}
    Following the same argument as above, we can show that
   \begin{equation}
      \sqrt{\frac{\sigma_{\max}\left(\dmat\right)}{n}}\| \frac{1}{n} \X' \Om \R 
      \X- \frac{1}{n}\X\Om\bpi \X\|_2 = O_p\left(1\right).
   \end{equation}
Note by Weyl's inequality, the smallest eigenvalue of $\frac{1}{n} \X' \Om \R \X $ converges to the smallest eigenvalue of $\frac{1}{n} \X' \Om \bpi \X $ in probability. Thus for a positive and sufficiently small $\epsilon$ and by Lemma \ref{LemmaDesign}, $\mathbf{P}(\lambda_{\min}(\frac{1}{n} \X' \Om\R \X)>\epsilon)\to 1$. Thus we have 
\begin{equation}
    \sqrt{\frac{\sigma_{\max}\left(\dmat\right)}{n}}\| (\frac{1}{n} \X' \Om \R \X)^{+} - (\frac{1}{n}\X\Om\bpi \X)^{-1}\|_2 = O_p(1).
\end{equation}
Finally note the algebraic decomposition that for $\widehat{A},A\in\mathbb{R}^{k_1\times k_2}$ and $\widehat{B},B\in\mathbb{R}^{k_2\times k_3}$
    \begin{equation*}
        \widehat{A}\widehat{B}-AB = (\widehat{A}-A)(\widehat{B}-B) + (\widehat{A}-A)B + A(\widehat{B}-B),
    \end{equation*}
    Let $\widehat{A}=(\frac{1}{n} \X' \Om \R \X)^{+}$, $\widehat{B}=\frac{1}{n} \X' \Om \R y$, $A=(\frac{1}{n}\X\Om\bpi \X)^{-1}$ and $B=\frac{1}{n}\X\Om\bpi y$. We have:
    \begin{equation*}
\sqrt{\frac{\sigma_{\max}\left(\dmat\right)}{n}}\|\widehat{b}^{\twls}_n-b^{\twls}_n\|_2= O_p(1).
    \end{equation*}
\end{proof}
The following lemma shows that, in order to establish convergence for a symmetric matrix estimator (of a fixed dimension), it suffices to consider its bilinear forms.
\begin{lemma}\label{quad}
Consider a sequence of symmetric matrices $\widehat{A}_n\in\mathbb{R}^{k\times k},n=1,2..,$ and a symmetric matrix $A_n\in \mathbb{R}^{k\times k}$. If for every $t\in\mathbb{R}^k$, $t'\widehat{A}_nt - t'A_nt\overset{p}{\to}0$, then $\widehat{A}_n-A_n\overset{p}{\to}\0_{k\times k}$.
\begin{proof}
 Using the standard basis vectors in $\mathbb{R}^k$, one can show that the difference in diagonal entries converges in probability to 0. Then, looking at all the two-by-two principal submatrices, one can show the off-diagonal entries converge in probability to 0 as well.
\end{proof}
\end{lemma}
\begin{remark}
Lemma \ref{quad} is not true for asymmetric matrices. For example, if 
\begin{equation}
    \widehat{A}_n=\begin{bmatrix} 0 & 1\\ -1 & 0  
\end{bmatrix}, \text{and, } A_n=\begin{bmatrix} 0 &\hspace{2mm} 0\\ 0 &\hspace{2mm} 0  
\end{bmatrix},
\end{equation}
then $t'\widehat{A}_nt=t'A_nt=0$ for all $t\in\mathbb{R}^{2}$ but $\widehat{A}_n-A_n$ does not converge to 0.
\end{remark}

We now state a tensor inequality. Consider a fourth-order $n$-dimensional tensor. We denote it as $\mathbf{A}=[a_{ijkl}]\in\mathbb{R}^{n\times n\times n \times n}$.  We shall understand it as a multi-linear function:
\begin{equation*}
 \mathbf{A}:\mathbb{R}^n \times \mathbb{R}^n \times \mathbb{R}^n \times \mathbb{R}^n \to \mathbb{R}, 
\end{equation*}
where for $x,y,z,a\in\mathbb{R}^n$,
\begin{equation}\label{optim}
 \mathbf{A}(w,x,y,z)= \sum_{i=1}^{n}\sum_{j=1}^{n}\sum_{k=1}^{n}\sum_{l=1}^{n}a_{ijkl}w_ix_jy_kz_l
\end{equation}
Consider the following maximization problem:
\begin{align*}
    & \max_{w,x,y,z}\mathbf{A}(w,x,y,z)\\
    &\text{  subject to }     \sum_{i=1}^n w_i^4=\sum_{i=1}^n x_i^4=\sum_{i=1}^n y_i^4=\sum_{i=1}^n z_i^4=1.
\end{align*}
We denote its optimal value as $\sigma_{\max}(\mathbf{A})$. This optimal value exists because we are optimizing a continuous function over a compact set. Note further the problem 
\begin{align*}
    & \max_{w,x,y,z}|\mathbf{A}(w,x,y,z)|\\
    &\text{  subject to }     \sum_{i=1}^n w_i^4=\sum_{i=1}^n x_i^4=\sum_{i=1}^n y_i^4=\sum_{i=1}^n z_i^4=1
\end{align*}
has the same solution as the problem above because we can always take the negative of one of the vectors. 

Note for any vectors $w,x,y,z\in\mathbb{R}^n$, we have:
\begin{align}
        &\left|\mathbf{A}\left(w,x,y,z\right)\right|= \|w\|_4\|x\|_4\|y\|_4\|z\|_4 \times \left|\mathbf{A}\left(\frac{w}{\|w\|_4},\frac{x}{\|x\|_4},\frac{y}{\|y\|_4},\frac{z}{\|z\|_4}\right)\right|\\
        & 
        \leq  \sigma_{\max}\left(\mathbf{A}\right) \|w\|_4\|x\|_4\|y\|_4\|z\|_4
\end{align}
The following lemma bounds $\sigma_{\max}(\mathbf{A})$. We define quantities:
\begin{equation*}
\| \mathbf{A}\|_{-i}=\max_i\sum_{j=1}^n\sum_{k=1}^n\sum_{l=1}^n |a_{ijkl}|
\end{equation*}
and analogously,
\begin{equation*}
\| \mathbf{A}\|_{-j}=\max_j\sum_{i=1}^n\sum_{k=1}^n\sum_{l=1}^n |a_{ijkl}|
\end{equation*}
and similarly for $\|A\|_{-k}$ and $\|A\|_{-l}$. We define:
\begin{equation*}
    \| \mathbf{A}\|_{\infty} =\max\{\| \mathbf{A}\|_{-i},\| \mathbf{A}\|_{-j},\| \mathbf{A}\|_{-k},\| \mathbf{A}\|_{-l}\}.
\end{equation*}
\begin{lemma}\label{LemmaTensor}
Consider a fourth-order $n$-dimensional tensor $\mathbf{A}=[a_{ijkl}]\in\mathbb{R}^{n\times n\times n \times n}$,  
\begin{equation*}
    \sigma_{\max}(\mathbf{A})\leq \| \mathbf{A}\|_{\infty}.
\end{equation*}
\end{lemma}
\begin{proof}
First note if $a_{ijkl}=0$ for all $ijkl$ indices, this inequality is trivially satisfied. Thus we assume there exists at least one $a_{ijkl}\not=0$ for an $ijkl$ index. \\
We observe that the objective function (\ref{optim}) is continuous and the feasible set is compact, so at least one optimal solution exists. Moreover, 0 is not in the feasible set. As a result, we conclude that a solution for the optimization problem exists and the local independence constraint qualification is satisfied at each solution (\cite{nocedal2006numerical}, P320). Let $\{w^*,x^*,y^*,z^*\}\subset\mathbb{R}^n$ denotes one of the optimal solutions. For the vector $w^*$, the KKT condition states that there exists a $\lambda_w^*$ such that
\begin{equation*}
    \frac{\partial}{\partial w_i}\mathbf{A}(w^*,x^*,y^*,z^*) - \lambda_w^* 4w_i^{*3}=0, \text{ for all } i\in[n].
\end{equation*}
We have
\begin{equation*}
      \frac{\partial}{\partial w_i}\mathbf{A}(w^*,x^*,y^*,z^*)= \sum_{j=1}^{n}\sum_{k=1}^{n}\sum_{l=1}^{n}a_{ijkl}x^*_jy^*_kz^*_l
\end{equation*}
We further have
\begin{align}
   & 4\sum_{i=1}^nw_i^*( \lambda_w^*w_i^{*3}) =\sum_{i=1}^nw_i^*\frac{\partial}{\partial w_i}\mathbf{\A}(w^*,x^*,y^*,z^*)=\sum_{i=1}^n\sum_{j=1}^{n}\sum_{k=1}^{n}\sum_{l=1}^{n}a_{ijkl}w^*_ix^*_jy^*_kz^*_l\\
   & = \mathbf{\A}(w^*,x^*,y^*,z^*).
\end{align}
\indent Given the constraint $\sum_{i=1}^n (w_i^*)^4=1$, we have the following equality:
\begin{equation*}
    \lambda_w^* =\frac{1}{4}\A(w^*,x^*,y^*,z^*)
\end{equation*}
Now notice:
\begin{align*}
& 4\lambda^*_w \|w^*\|_\infty^3  = 4\lambda^*_w \max_i |w_i^*|^3 = \max_i |\sum_{j=1}^{n}\sum_{k=1}^{n}\sum_{l=1}^{n}a_{ijkl}x^*_jy^*_kz^*_l| \\
& \leq \max_i \{\sum_{j=1}^{n}\sum_{k=1}^{n}\sum_{l=1}^{n}|a_{ijkl}|\} \times \|x^*\|_{\infty}\times \|y^*\|_{\infty}\times \|z^*\|_{\infty}\\
        & \leq \|\A\|_{\infty}\times \|x^*\|_{\infty}\times \|y^*\|_{\infty}\times \|z^*\|_{\infty},
\end{align*}
Notice a similar argument also works for $x^*$ and $y^*$ and $z^*$, and $\lambda^*_w=\lambda^*_x=\lambda^*_y=\lambda^*_z=\frac{1}{4}\A(w^*,x^*,y^*,z^*)$. Define $\lambda^*=\frac{1}{4}\A(w^*,x^*,y^*,z^*)$. We have the following four equations:
\begin{align*}
    4\lambda^* \|w^*\|^3_{\infty}&\leq \|\A\|_{\infty}\times \|x^*\|_{\infty}\times \|y^*\|_{\infty}\times \|z^*\|_{\infty}\\
    4\lambda^* \|x^*\|^3_{\infty}&\leq \|\A\|_{\infty}\times \|w^*\|_{\infty}\times \|y^*\|_{\infty}\times \|z^*\|_{\infty} \\  
    4\lambda^* \|y^*\|^3_{\infty}&\leq \|\A\|_{\infty}\times \|w^*\|_{\infty}\times \|x^*\|_{\infty}\times \|z^*\|_{\infty} \\  
    4\lambda^* \|z^*\|^3_{\infty}&\leq \|\A\|_{\infty}\times \|w^*\|_{\infty}\times \|x^*\|_{\infty}\times \|y^*\|_{\infty}  
\end{align*}
Define $v^*=\max\{\|w^*\|_{\infty},\|x^*\|_{\infty},\|y^*\|_{\infty},\|z^*\|_{\infty}\}$, we have the inequality:
\begin{equation*}
    4\lambda^*\times (v^*)^3\leq \|\A\|_{\infty}\times (v^*)^3
\end{equation*}
Since $v^*\not=0$, We then have the inequality $ \A(w^*,x^*,y^*,z^*)=4\lambda^* \leq \|\A\|_{\infty}$.
\end{proof}

Denote $\|z\|_4^4=\sum_{i,j}z^4_{ij}$, $\|z\|_1=\sum_{i,j}|z_{ij}|$ and $\|z\|_{\infty}=\max_{i,j}|z_{ij}|$ for a matrix $z$.
\begin{lemma}{(Convergence of the infeasible variance estimator)}\label{Infeasible} Let $\widetilde{\dmat}$ be a valid variance bound matrix, and $\dtildep$ be the inverse probability weighted version of the bounding matrix $\widetilde{\dmat}$. Let $\mathbf{Q}$ denote the fourth-order tensor $(\widetilde{\dmat}\bigotimes\widetilde{\dmat})\circ \mathbf{S} $, as defined in (\ref{Smat}). Let $z\in\mathbb{R}^{kn\times k}$. Consider the estimator:
\begin{equation*}
   \widehat{\widetilde{\V}}=\frac{1}{n^2}z'\R\dtildep\R z\in\mathbb{R}^{k\times k}
\end{equation*}
for the quantity 
\begin{equation*}
    \widetilde{\V}=\frac{1}{n^2}z'\widetilde{\dmat} z\in\mathbb{R}^{k\times k}.
\end{equation*}
Then, 
\begin{equation*}
   \widehat{\widetilde{\V}}-\widetilde{\V}=O_p\left( \sqrt{  \frac{1}{n^3}\sigma_{\max}(\mathbf{Q})}\times \sqrt{\frac{1}{n}\|z\|_4^4 }\right)
\end{equation*}
\end{lemma}
\begin{proof}
For an arbitrary $t\in\mathbb{R}^k$, consider the quadratic form        $t'(\widehat{\widetilde{\V}}-\widetilde{\V})t$. We are interested in upper-bounding its convergence rate.

Note $t'\left(\widehat{\widetilde{\V}}\right)t$ is unbiased for $t'\left(\widetilde{\V}\right)t$. To upper-bound the convergence rate, we study its variance:\\
\begin{align}\label{(9)}
     \V \bigg ( t'\widehat{\widetilde{\V}}t \bigg )& =  \frac{1}{n^4}\E \left[ \left(t'z\R \dtildep \R z t - t'z \widetilde{\dmat}zt\right)^2\right]
\end{align}
Use $\R_i,i\in [kn]$ to denote the $i$th diagonal element of $\R$. Some algebra shows that (\ref{(9)}) has the form:
\begin{equation*}
   \frac{1}{n^4}\mathbf{Q}(t'z,t'z,t'z,t'z)= \frac{1}{n^4}\sum_{i,j,k,l=1}^{kn}\mathbf{COV}(\R_i\R_j,\R_k\R_l)\frac{\widetilde{d}_{ij}\widetilde{d}_{kl}}{\pi_{ij}\pi_{kl}}(t'z)_i(t'z)_j(t'z)_k(t'z)_l,
\end{equation*}
where $\widetilde{d}_{ij}$ denotes the $ij$th entry of the matrix $\widetilde{\dmat}$, $\pi_{ij}$ denotes the $ij$th entry of the matrix $\matp$ defined in Definition \ref{SO}, and $(t'z)_i$ denotes the $i$th entry of the vector $t'z$. Note that for those entries where $\pi_{ij}=0$ or $\pi_{kl}=0$, we also have $\tilde{d}_{ij}=0$ or $\tilde{d}_{kl}$ since $\tilde{\Ome}$ is a valid variance bound. Hence the quantity is well-defined. 

Note $\mathbf{Q}(\cdot,\cdot,\cdot,\cdot)$ is a fourth-order $kn$ dimensional tensor. Then we have,
\begin{align*}
&  \frac{1}{n^4}\E \left[ \left(t'z\R \dtildep \R z t - t'z \widetilde{\dmat} zt\right)^2\right] = \frac{1}{n^4}\mathbf{Q}(t'z,t'z,t'z,t'z)\\ 
& \leq\frac{1}{n^4}\sigma_{\max}(\mathbf{Q})\|t'z\|_4^4 \leq  \frac{1}{n^4}\sigma_{\max}(\mathbf{Q})k^{3}\|z\|_4^4 \times \|t\|_\infty^4  \leq  \frac{1}{n^4}\sigma_{\max}(\mathbf{Q})k^{3}\|z\|_4^4 \times \|t\|_2^4.
\end{align*}
Thus we have
\begin{equation*}
     t'(\widehat{\widetilde{\V}}-\widetilde{\V})t= \sqrt{  \frac{1}{n^3}\sigma_{\max}(\mathbf{Q})\frac{1}{n}\|z\|_4^4  }\times k^3\|t\|_2^2.
\end{equation*}
\end{proof}
\begin{lemma}{(Feasible Estimators Converging to  Infeasible Estimators)}\label{Feasible}
Let $\dtildep$ be the inverse probability weighted version of the bounding matrix $\widetilde{\dmat}$. Consider plug-in estimator of the form
\begin{equation*}
\frac{1}{n^2}\widehat{z}\dtildep\widehat{z} \in \mathbb{R}^{k\times k},
\end{equation*}
where $\widehat{z}\in\mathbb{R}^{kn}$. 
We have 
\begin{align*}
      & \frac{1}{n^2}\widehat{z}\R\dtildep\R\widehat{z}- \frac{1}{n^2}z\R\dtildep\R z \\
      =& O\left(\frac{1}{n^2}\|\widehat{z}-\R z\|_2^2 \times \sigma_{\max}\left(\dtildep\right) +  \frac{1}{n^2}\|\widehat{z}-\R z \|_2\times \|z\|_2 \times\sigma_{\max}\left(\dtildep\right)\right).
\end{align*}
\end{lemma}
\begin{proof}
We have the following algebraic manipulation:
\begin{align}\label{(15)}
        & \frac{1}{n^2}\widehat{z}\dtildep\widehat{z}-    \frac{1}{n^2}z\R\dtildep \R z \\
        &=  \frac{1}{n^2}(\widehat{z}-\R z)'\dtildep\widehat{z} + \frac{1}{n^2}z'\R\dtildep\widehat{z}-   \frac{1}{n^2}z'\R\dtildep\R z\\
   & =  \frac{1}{n^2}(\widehat{z}-\R z)'\dtildep (\widehat{z}-\R z) + \frac{1}{n^2}(\widehat{z}-\R z)\dtildep\R z + \frac{1}{n^2}z\R\dtildep (\widehat{z}-\R z)
\end{align}
Using the notation $\viii{\cdot}$  as a shorthand notation for $\sigma_{\max}()$ and by the submultiplicativity of a matrix norm,\footnote{Note $\sigma_{\max}()$ is a matrix norm.} we have the quantity upper bounded by
\begin{align*}
    & \viii{\frac{1}{n^2}(\widehat{z}-\R z)'\dtildep (\widehat{z}-\R z)+ \frac{1}{n^2}(\widehat{z}-\R z)\dtildep\R z + \frac{1}{n^2}z'\R\dtildep(\widehat{z}-\R z)}\\
  \leq &  \frac{1}{n^2}\viii{ (\widehat{z}-\R z)\dtildep(\widehat{z}-\R z)} + \frac{1}{n^2}\viii{(\widehat{z}-\R z)\dtildep\R z}  +  \frac{1}{n^2}\viii{z'\R\dtildep(\widehat{z}-\R z)}\\
   \leq & \frac{1}{n^2}\viii{\widehat{z}-\R z}^2\times \viii{\dtildep}+ \frac{2}{n^2}\viii{\widehat{z}-\R z}\times\viii{\R}\times\viii{\dtildep}\times\viii{z}  \\
  \leq & \frac{1}{n^2}\|\widehat{z}-\R z\|_2^2\times \viii{\dtildep} +  \frac{2}{n^2}\|\widehat{z}-\R z\|_2\times \|z\|_2 \times\viii{\R}\times \viii{\dtildep} 
\end{align*}
Thus we have for (\ref{(15)})
\begin{equation*}
    (\ref{(15)}) =O \left(\frac{1}{n^2}\|\widehat{z}-\R z\|_2^2 \times \sigma_{\max}\left(\dtildep\right) +  \frac{1}{n^2}\|\widehat{z}-\R z\|_2\times \|z\|_2 \times\sigma_{\max}\left(\dtildep\right)\right) ,
\end{equation*}
where we used the fact that $\viii{\R}\leq1$.
\end{proof}
Let $\viiii{\mathbf{A}}_1$ denote the $l_1$-induced matrix norm, where $\viiii{\mathbf{A}}_1=\max_{j\in[n]}\{\sum_{i=1}^n |a_{ij}|\} $, for an arbitrary $\mathbf{A}\in\mathbb{R}^{m\times n}$.

\begin{lemma}\label{Lemma:bound}
Let $\Ome\in\mathbb{R}^{n\times n}$ be a symmetric matrix, $y\in \mathbb{R}^n$ and $x\in \mathbb{R}^n$ be two arbitrary column vectors. The quantity $x'\Ome'\diag(y)\diag(y)\Ome x$ can be upper bounded by :
\begin{equation}
  x'\Ome'\diag(y)\diag(y)\Ome x  \leq \viiii{\Ome}_1^2 \sqrt{\sum_{i=1}^n x_i^4} \sqrt{\sum_{i=1}^n y_i^4} 
\end{equation}
\end{lemma}
\begin{proof}
Notice the stated quantity is the $l_2$ norm of the vector $x'\Ome'\diag(y)\in\mathbb{R}^n$. Denote the $j$th column of $\Ome$ by $\Ome_j$, and the $(i,j)$ th entry of $\Ome$ by $\Ome_{ij}$. The $i$th entry of the vector $x'\Ome'\diag(y)$ is $x'\Ome_iy_i$. The $l_2$ norm thus can be written as
\begin{align*}
            & \sum_{i=1}^n \left(x'\Ome_iy_i\right)^2 = \sum_{i=1}^n \left(\sum_{j=1}^n x_j\Ome_{ji}y_i\right)^2\leq \sum_{i=1}^n \left(\sum_{j=1}^n |x_j\|\Ome_{ji}\|y_i|\right)^2\\
            = &  \sum_{i=1}^n \left(\sum_{j=1}^n |x_j|\sqrt{|\Ome_{ji}|}|y_i| \sqrt{|\Ome_{ji}| }\right)^2 \leq  \sum_{i=1}^n \left(\sum_{j=1}^n |x_j|^2|y_i|^2 |\Ome_{ji}|\right) \times \left(\sum_{j=1}^n |\Ome_{ji}|\right) \\
            \leq & \viiii{\Ome}_1 \sum_{i=1}^n \sum_{j=1}^n |x_j|^2|y_i|^2 |\Ome_{ji}|
\end{align*}
Now notice the expression $\sum_{i=1}^n \sum_{j=1}^n |x_j|^2|y_i|^2 |\Ome_{ji}|$ is the expanded expression for the quadratic form $(x\circ x)' |\Ome| (y\circ y)$, where $|\Ome|$ replaces entries in $\Ome$ with their absolute values and $\circ$ denotes the Hadamard product. Thus we can continue the inequality 
\begin{align*}
   &  \viiii{\Ome}_1 \sum_{i=1}^n \sum_{j=1}^n |x_j|^2|y_i|^2 |\Ome_{ij}|\leq  \viiii{\Ome}_1 \sigma_{\max}\left(\left|\Ome\right|\right) \times \sqrt{\sum_{i=1}^n x_i^4} \sqrt{\sum_{i=1}^n y_i^4}\\
   \leq &  \viiii{\Ome}_1^2 \sqrt{\sum_{i=1}^n x_i^4} \sqrt{\sum_{i=1}^n y_i^4},
\end{align*}
where for the last line we used the fact that for a symmetric matrix,
\begin{equation*}
  \sigma_{\max}\left(\left|\Ome\right|\right)\leq \viiii{\left|\Ome\right|}_1=\viiii{\Ome}_1.
\end{equation*}
\end{proof}
%\begin{remark}
%Taking $\Ome=\begin{bmatrix}
%     1 & 1 & ... & 1\\
%     1 & 0 & ... & 0 \\
%     \vdots & \vdots & \ddots & \vdots\\
%     1 & 0 & ... & 0
%\end{bmatrix}$, $y=e_1 \in \mathbb{R}^n$ and %$x=\onesmat\in\mathbb{R}^n$, one can show %the upper bound $\||\Ome\||_1 \||%(|\Ome|)\||_2 \times \sqrt{\sum_{i=1}^n x_i^4} \sqrt{\sum_{i=1}^n y_i^4}$ is tight in terms of the order of $n$. 
%\end{remark}
Let $\Theta$ be a compact set in a finite dimensional Euclidean space, and $B(\theta,\delta)$ denote a closed ball in $\Theta$ of radius $\delta \geq 0$ (with the $l_2$ norm) centered at $\theta$. The following lemma adapts Theorem 1 \cite{andrews1992generic} to our setting. \begin{lemma}\label{Lemma:SC}
 Let $\Theta$ be a compact set in a finite dimensional Euclidean space and $\mathbf{P}_n$ be the probability measure induced by the random assignments. Consider a sequence of continuous deterministic functions $Q_n(\cdot):\Theta\to\mathbb{R}$ and continuous stochastic functions $\widehat{Q}_n(\cdot):\Theta\to\mathbb{R}$.\footnote{We assume the criterion functions are continuous to avoid measurability issues. This condition is satisfied by all models considered in this paper.  } If $|\widehat{Q}_n(\theta)-Q_n(\theta)|=o_p(1)$ pointwise on $\theta\in\Theta$ and the stochastic function $\widehat{Q}_n(\cdot)-Q_n(\cdot)$ is uniformly stochastically equicontinuous: for all $\epsilon>0$, there exists a $\delta>0$ such that
\begin{align*}
    \mathbf{P}_n\left(\sup_{\theta\in\Theta}\sup_{\theta'\in B(\theta,\delta)}|\widehat{Q}_n(\theta)-Q_n(\theta)-\left(\widehat{Q}_n(\theta')-Q_n(\theta')\right)|>\epsilon \right)<\epsilon
\end{align*}
uniformly for large $n$. Then $\sup_{\theta\in\Theta} \left|\widehat{Q}_n(\theta)-Q_n(\theta)\right|=o_p(1)$. 
\end{lemma}
\begin{proof}
Fix a given $\epsilon>0$ and let $\delta$ to be the corresponding radius in the stochastic equicontinuity condition. Because $\Theta$ is compact, we can find a finite cover of $\Theta$, $\{B(\theta_j,\delta)\}_{j=1,...,J}$. We then have:
\begin{align*}
  & \mathbf{P}_n \left(\sup_{\theta\in\Theta} \left|\widehat{Q}_n(\theta)-Q_n(\theta)\right|>2\epsilon\right)\\
  \leq &   \mathbf{P}_n\left(\sup_{\theta\in\Theta}\sup_{\theta'\in B(\theta,\delta)} \left|\widehat{Q}_n(\theta)-Q_n(\theta)-\left(\widehat{Q}_n(\theta')-Q_n(\theta')\right)\right|>\epsilon\right)\\
  & + \mathbf{P}_n\left(\max_{j\in [J]}\left|\widehat{Q}_n(\theta_j)-Q_n(\theta_j)\right|>\epsilon\right)<2\epsilon,
\end{align*}
for large $n$. Then for each $\epsilon>0$ and any $\widetilde{\epsilon}<\epsilon$, we can find an $n$ large enough such that $ \mathbf{P}_n\left(\sup_{\theta\in\Theta} |\widehat{Q}_n(\theta)-Q_n(\theta)|>\epsilon\right)\leq \mathbf{P}_n\left(\sup_{\theta\in\Theta} |\widehat{Q}_n(\theta)-Q_n(\theta)|>\widetilde{\epsilon}\right)<2\widetilde{\epsilon}$. Thus $\sup_{\theta\in\Theta} \left|\widehat{Q}_n(\theta)-Q_n(\theta)\right|=o_p(1)$.
\end{proof}
The following lemma is the standard consistency proof for GMM estimators. 
\begin{lemma}\label{Lemma:UC}
 Let $\Theta$ be a compact set in a finite dimensional Euclidean space and $\mathbf{P}_n$ be the probability measure induced by the random assignments. Consider a sequence of continuous deterministic functions $Q_n(\cdot):\Theta\to\mathbb{R}$ and continuous stochastic functions $\widehat{Q}_n(\cdot):\Theta\to\mathbb{R}$. Let $\theta_n\equiv\arg\inf_{\theta\in\Theta}Q_n(\theta)$ and define $\widehat{\theta}_n\equiv\arg\min_{\theta\in\Theta}\widehat{Q}_n(\theta)$. If there exists a positive $c$ such that $\inf_{\theta\in \Theta\backslash B(\theta_n,\epsilon)}Q_n(\theta)-Q_n(\theta)>c \epsilon^2$ uniformly for all large $n$ and $\sup_{\theta\in\Theta}|\widehat{Q}_n(\theta)-Q_n(\theta)|=o_p(1)$, then $\widehat{\theta}_n-\theta_n=o_p(1)$.
\end{lemma}
\begin{proof}
    We have
    \begin{align*}
        &  \mathbf{P}_n\left(\|\widehat{\theta}_n-\theta_n\|_2>2\epsilon\right) \leq         \mathbf{P}_n\left(Q_n(\widehat{\theta}_n)-Q_n(\theta_n)>4c\epsilon^2\right)\\
         \leq & \mathbf{P}_n\left( Q_n(\widehat{\theta}_n)-Q_n(\theta_n)-\widehat{Q}_n(\widehat{\theta}_n)+\widehat{Q}_n(\theta_n)>2c\epsilon^2\right) + \mathbf{P}_n\left(\widehat{Q}_n(\widehat{\theta}_n)-\widehat{Q}_n(\theta_n)>2c\epsilon^2\right)\\
         \leq& 2\mathbf{P}_n(\sup_{\theta\in\Theta}|Q_n(\theta)-Q_n(\theta)|>c\epsilon^2)\to 0,
    \end{align*}
as $n\to\infty$, where $\mathbf{P}_n\left(\widehat{Q}_n(\widehat{\theta}_n)-\widehat{Q}_n(\theta_n)>c2\epsilon^2\right)=0$ because $\widehat{\theta}_n$ is the minimzer of $\widehat{Q}_n(\theta)$. 
\end{proof}
The following lemma uses notations in Section \ref{Section:Nonlinear}.
\begin{lemma}\label{Lemma:Equivalence}
 Define a GR estimator for arm a as $\widehat{\mu}_{n,a}=\frac{1}{n}\sum_{i}f^a(x_i,\widehat{\theta}_n)+\frac{1}{n}\sum_{i}\frac{\R_{ai}}{\pi_{ai}}\left(y_{ai}-f^a(x_i,\widehat{\theta}_n)\right)$, and $\widehat{\mu}_{n,a}^L=\frac{1}{n}\sum_{i}f^a(x_i,\theta_n)+\frac{1}{n}\sum_{i}\frac{\R_{ai}}{\pi_{ai}}\left(y_{ai}-f^a(x_i,\theta_n)\right)$. If $\widehat{\theta}_n-\theta_n=O_p\left(n^{-\frac{1}{2}}\right)$, and there exists a positive integer $N$ such that the following conditions for $f^a$ are satisfied uniformly for all $n\geq N$:
 \begin{enumerate}[label=(\roman*)]
    \item $f^a(x_i,\theta)$ is two times differentiable in $\theta$ for all $x_i$ values, $i\in[n]$.
    \item There exists a $C_{\ref{Lemma:Equivalence},1}$ and an $\epsilon_{\ref{Lemma:Equivalence},1}>0$ such that $\frac{1}{n}\sum_{a,i} \|\nabla_{\theta}f^a(x_i,\theta)\|^2_2<C_{\ref{Lemma:Equivalence},1}$ for all $\theta\in B(\theta_n,\epsilon)$.
    \item There exists a $C_{\ref{Lemma:Equivalence},1}$ and an $\epsilon_{\ref{Lemma:Equivalence},2}>0$ such that $\frac{1}{n}\sum_{a,i} \sup_{\theta\in B(\theta_n,\epsilon)} \|\nabla_{\theta\theta}f^a(x_i,\theta)\|_1<C_{\ref{Lemma:Equivalence},2}$,
    \item $\sigma_{\max}\left(\Ome\right)=O(1)$,
\end{enumerate}
then $\frac{1}{n} \sum_{i}\left(f^a(x_i,\widehat{\theta}_n)-f^a(x_i,\theta_n)\right)^2=O_p(n^{-1})$ and  $\widehat{\mu}_{n,a}-\widehat{\mu}_{n,a}^L=o_p(n^{-\frac{1}{2}})$.
\end{lemma}
\begin{proof}
We first show $\frac{1}{n} \sum_{i}(f^a(x_i,\widehat{\theta}_n)-f^a(x_i,\theta_n))^2=O_p(n^{-1})$. We have two useful facts:
\begin{enumerate}
    \item For any $\theta_1,\theta_2,\widetilde{\theta}\in B(\theta_n,\epsilon)$ and $\|\theta_1-\theta_2\|_2\leq \delta$, by (ii),
    \begin{equation}\label{14}
        \frac{1}{n}\sum_{i=1}^n \left(\nabla_\theta f^a(x_i,\widetilde{\theta})'(\theta_1-\theta_2)\right)^2\leq \left(\frac{1}{n}\sum_{i=1}^n \|\nabla_\theta f^a(x_i,\widetilde{\theta})\|_2^2\right) \times \|\theta_1-\theta_2\|_2^2 \leq C_1\delta^2
    \end{equation}
    \item For any $\theta_1,\theta_2\in B(\theta_n,\delta)$ with $\delta<\epsilon_1$, we have
    \begin{align*}
          & \frac{1}{n}\sum_{i=1}^n\left(f^a(x_i,\theta_2)-f^a(x_i,\theta_1)\right)^2\\
          =&   \frac{2}{n}\sum_{i=1}^n \left(f^a(x_i,\widetilde{\theta}_2)-f^a(x_i,\theta_1)\right)\nabla_\theta f^a(x_i,\widetilde{\theta}_2)'(\theta_2-\theta_1)\\
          \leq &  2\sqrt{\frac{1}{n}\sum_{i=1}^n \left(f^a(x_i,\widetilde{\theta}_2)-f^a(x_i,\theta_1)\right)^2} \times \sqrt{\frac{1}{n}\sum_{i=1}^n \left(\nabla_\theta f^a(x_i,\widetilde{\theta}_2)'(\theta_1-\theta_2)\right)^2},
        \end{align*}
     where $\widetilde{\theta}_2$ is between $\theta_1$ and $\theta_2$. Taking supremum on both sides over $\theta_2,\theta_1,\widetilde{\theta}_2$, we arrive at:
     \begin{align*}
         & \sup_{\theta_1,\theta_2\in B(\theta_n,\delta)} \frac{1}{n}\sum_{i=1}^n(f^a(x_i,\theta_2)-f^a(x_i,\theta_1))^2 \\
          \leq & 2\sqrt{ \sup_{\theta_1,\theta_2\in B(\theta_n,\delta)} \frac{1}{n}\sum_{i=1}^n(f^a(x_i,\theta_2)-f^a(x_i,\theta_1))^2} \times \sqrt{C_{\ref{Lemma:Equivalence},1}\delta^2},   
     \end{align*}
     which yields:
     \begin{equation}\label{16}
         \sup_{\theta_1,\theta_2\in B(\theta_n,\delta)}\frac{1}{n}\sum_{i=1}^n(f^a(x_i,\theta_2)-f^a(x_i,\theta_1))^2 \leq 4C_{\ref{Lemma:Equivalence},1}\delta^2
     \end{equation}
    \end{enumerate}
Because $\widehat{\theta}_n-\theta_n=O_p(n^{-\frac{1}{2}})$, 
\begin{align*}
& \frac{1}{n}\sum_{i=1}^n(f^a(x_i,\widehat{\theta}_n)-f^a(x_i,\theta_n))^2\leq   \sup_{\theta_1,\theta_2\in B(\theta_n,\|\widehat{\theta}_n-\theta_n\|)}\frac{1}{n}\sum_{i=1}^n(f^a(x_i,\theta_2)-f^a(x_i,\theta_1))^2 \\
\leq & 4C_1\|\widehat{\theta}_n-\theta_n\|^2_2= O_p(n^{-1}),
\end{align*}
with probability going to one and this proves the first claim. Now the second claim follows by:
\begin{align*}
    & \widehat{\mu}_{n,a}-\widehat{\mu}_{n,a}^\tL =\frac{1}{n}\sum_{i=1}^n \left[f^a(x_i,\widehat{\theta}_n) -  f^a(x_i,\theta_n)\right] + \frac{1}{n}\sum_{i=1}^n\frac{\R_{ai}}{\bpi_{ai}} (f^a(x_i,\theta_n)-f^a(x_i,\widehat{\theta}_n))\\
     = & \frac{1}{n}\sum_{i=1}^n\nabla_{\theta}f^a(x,\theta_n)'(\widehat{\theta}_n-\theta_n) + (\widehat{\theta}_n-\theta_n)' \{\frac{1}{n}\sum_{i=1}^n \nabla_{\theta\theta}f^a(x,\widetilde{\theta}_n)\}(\widehat{\theta}_n-\theta_n) \\
     & - \frac{1}{n}\sum_{i=1}^n\frac{\R_{ai}}{\bpi_{ai}}\nabla_{\theta}f^a(x,\theta_n)'(\widehat{\theta}_n-\theta_n) - (\widehat{\theta}_n-\theta_n)' \{\frac{1}{n}\sum_{i=1}^n\frac{\R_{ai}}{\bpi_{ai}} \nabla_{\theta\theta}f^a(x,\widetilde{\theta}_n)\}(\widehat{\theta}_n-\theta_n)\\
     = & o_p\left(\frac{1}{\sqrt{n}}\right),
\end{align*}
where $\widetilde{\theta}_n$ is between $\widehat{\theta}_n-\theta_n$.
The final line follows by noticing
\begin{align*}
          & \frac{1}{n}\sum_{i=1}^n\nabla_{\theta}f^a(x,\theta_n)'(\widehat{\theta}_n-\theta_n) -  \frac{1}{n}\sum_{i=1}^n\frac{\R_{ai}}{\bpi_{ai}}\nabla_{\theta}f^a(x_i,\theta_n)'(\widehat{\theta}_n-\theta_n)\\
         = &   - \frac{1}{n}\sum_{i=1}^n\frac{\R_{ai}-\bpi_{ai}}{\bpi_{ai}}\nabla_{\theta}f^a(x_i,\theta_n)'(\widehat{\theta}_n-\theta_n)\\
         = & O_p\left(\frac{\sigma_{\max}\left(\Ome\right)}{\sqrt{n}}\right)O_p\left(\frac{1}{\sqrt{n}}\right)=o_p(n^{-\frac{1}{2}})
\end{align*}
    by (ii) and Lemma \ref{HTcp1}. 
 Because $\widetilde{\theta}_n \in B(\theta_n,\epsilon)$ with probability approaching one, the middle term can be upper bounded by
    \begin{align}
\begin{split}
        & \|\frac{1}{n}\sum_{i=1}^n\frac{\R_{ai}-\pi_{ai}}{\pi_{ai}}\nabla_{\theta\theta}f^a(x_i,\widetilde{\theta}_n)\|_1\\
        & \leq \frac{1}{n} \sum_{i=1}^n\frac{\R_{ai}}{\pi_{ai}}\sup_{\theta\in B(\theta_n,\epsilon)}\|\nabla_{\theta\theta}f^a(x_i,\theta)\|_1 + \frac{1}{n} \sum_{i=1}^n \sup_{\theta\in B(\theta_n,\epsilon)}\|\nabla_{\theta\theta}f^a(x_i,\theta)\|_1. 
\end{split}
\end{align}
The upper bound is of order $O_p(1)$ by (iii) and the Markov inequality. Together with $\widehat{\theta}_n-\theta_n=O_p(n^{-\frac{1}{2}})$, this implies the $O_p(n^{-1})$ rate for the second-order remainder term. 
\end{proof}
    
\section{Proofs for the results in Section \ref{Section:Estimation}}
\subsection{Proof of Lemma \ref{lemma:WLS}}

Notice that if $\frac{1}{n}\sum_{i=1}^n x_{si}=0$ for all $s\in [p]$, we have $ \frac{1}{n}\onesmat'X =  \begin{bmatrix}
 \I_k \vert \0_{k\times p}
    \end{bmatrix}$. Hence $ \begin{bmatrix}
\I_k \vert \0_{k\times p}
\end{bmatrix}b^{\twls}=\onesmat'Xb^{\twls}$.
Hence $\frac{1}{n}\onesmat'\left(y -  Xb^{\twls}\right)=0$ if columns of $\onesmat$ are in the column space of $\bpi\Om\X$ since
\begin{equation}
   \onesmat'\left(y -  Xb^{\twls}\right)=t'\X'\Om\bpi \left(y-\X b^{\twls}\right), 
\end{equation}
for some $t\in\mathbb{R}^{(k+p)\times k}$ such that $\onesmat=\pi\Om \X t$, by the first-order conditions of a WLS estimator.

\subsection{Proof of Theorem \ref{Thm:Consistency}}
\begin{proof}
We have by definition $\E[\widehat{m}_n^s]=m_n^s$ for $s\in [l_1]$. It can be seen, as in Lemma \ref{HTcp1}, $\V(\widehat{m}_n^s)=\frac{1}{n^2}\phi^s\dmat \phi^s\leq \frac{1}{n}\|\phi^s\|^2_2 \times \frac{1}{n}\sigma_{\max}\left(\dmat\right)$. Thus we have $\widehat{m}_n^s-m_n^s=O_p(n^{-\frac{1}{2}}\sqrt{\sigma_{\max}\left(\dmat\right)})=o_p(1)$ for $s\in [l_1]$.  By Assumption \ref{A:MomentEstimator2}-(i), $\|\widehat{\mu}_n-\mu_n\|_2\leq C_{\ref{A:MomentEstimator2}},1 \sqrt{\sum_{s=1}^{l_1} (\widehat{m}_n^s-m_n^s)^2}$. Thus $\|\widehat{\mu}_n-\mu_n\|_2=O_p(n^{-\frac{1}{2}}\sqrt{\sigma_{\max}\left(\dmat\right)})=o_p(1)$.
\end{proof}
\subsection{Proof of Corollary \ref{C:Consistency}}
\begin{proof}
We prove only for the GR estimators. Other estimators can be proved analogously. To verify Assumption \ref{A:MomentEstimator2} for the GR estimators, notice:
\begin{align*}
    \widehat{\mu}^{\tgr}_n & = \underbrace{\frac{1}{n}\onesmat'\bpiInv\R y}_{(1)} - \underbrace{(\frac{1}{n}\onesmat'\bpiInv(\R-\bpiInv)'\X)}_{(2)} \times \underbrace{(\X'\m\R\X)^{-1}}_{(3)} \times \underbrace{(\X'\m\R y)}_{(4)}\\
    & = F(k \text{ moments in (1)}, k+p \text{ moments in (2)},\frac{(k+p)(k+p+1)}{2} \\
    & \text{ moments in (3)}, k+p \text{ moments in (4)} \hspace{2pt})
\end{align*}
We rewrite the GR estimator using the moment-type estimator formulation as $F(\widetilde{y}_n,\widetilde{x}_n,\widetilde{A}_n,\widetilde{b}_n)$ where $\widetilde{y}_n\in\mathbb{R}^k$ denotes the $k$ entries corresponding to the $k$ moments in (1), $\widetilde{x}_n\in\mathbb{R}^{k+p}$  denotes the k+p entries corresponding to the k+p moments in (2), $\widetilde{A}_n\in\mathbb{R}^{(k+p)\times (k+p)}$ is a symmetric matrix corresponding to the $\frac{(k+p)(k+p+1)}{2} \text{ moments in (3)}$ and $\widetilde{b}_n\in\mathbb{R}^{k+p}$ denotes the k+p entries corresponding to the k+p moments in (4).\footnote{Some moments in (3) can appear in $\widetilde{A}_n$ twice.} Let $y_n=\frac{1}{n}\onesmat'y$, $x_n=0$, $A_n=\frac{1}{n}\X'\Om\bpi\X$ and $b=\frac{1}{n}\X\Om\bpi y$. By Assumption \ref{A:Invertibility}, the fact that $\lambda_{\min}\left(\omega\pi\right)>c>0$ and Lemma \ref{LemmaDesign}, $\lambda_{\min}(A_n)$ is bounded away from 0 uniformly for $n\geq n_0$. For a sufficient small $\epsilon$ and by Corollary 6.3.8 in \cite{horn2012matrix}, $\lambda_{\min}(\widetilde{A}_n)$ is uniformly bounded away from 0 for all $\|\widetilde{A}_n-A_n\|_2<\epsilon$.
For $\widetilde{y}_n,\widetilde{x}_n,\widetilde{A}_n$ and $\widetilde{b}_n$ that are sufficiently close to ${y}_n,0 ,A_n$ and $b_n$, we have:
\begin{align}
    &\|F(\widetilde{y}_n,\widetilde{x}_n,\widetilde{A}_n,\widetilde{b}_n)-  F({y}_n,0 ,A_n,b_n)\|_2 = \|y_n-\widetilde{y}_n-\widetilde{x}_n \widetilde{A}_n^{-1}\widetilde{b}_n\|_2\\
    &  \leq \|y_n-\widetilde{y}_n \|_2 + \|\widetilde{x}_n \widetilde{A}_n^{-1}\widetilde{b}_n\|_2\\
    & \leq \|y_n-\widetilde{y}_n\|_2 + \|\widetilde{x}_n\|_2 \times \|\widetilde{A}_n^{-1}\widetilde{b}_n\|_2
  \leq \|y_n-\widetilde{y}_n\|_2 + C\|\widetilde{x}_n\|_2,
\end{align}
where $C$ is a constant independent of $n$ and we also use Lemma \ref{LemmaBoundedCoefs} to bound $\|A_n^{-1}\widetilde{b}_n\|_2$ by a constant $C$ that is independent of $n$. Assumption \ref{A:MomentEstimator2}-(ii) is satisfied by Assumption \ref{A:BoundedFourthMoments}.
\end{proof}

\subsection{Proof of Theorem \ref{Thm:FirstOrderExpansion}}
\begin{proof}
By an argument similar to Lemma \ref{HTcp1}, we have $\widehat{m}_n^s-m_n^s=O_p(\sqrt{\sigma_{\max}\left(\dmat\right)}n^{-\frac{1}{2}})=o_p(1)$ for $s\in[l_1]$.  By Assumption \ref{A:MomentEstimator2}, we have, with probability approaching one, $\|\widehat{\mu}_n-\mu_n^\tL\|_2=O_p(\sqrt{\sigma_{\max}\left(\dmat\right)}n^{-\frac{1}{2}})$.

For any linearized estimator, it has the form:
\begin{equation*}
    \mu_n + dF_{m_n}(\widehat{m}_{n,r}-m_{n,r}) \in\mathbb{R}^k
\end{equation*}    
Note the $s$th entry of the vector $\widehat{m}_{n,r}-m_{n,r}$ has the form $\frac{1}{n}1_{\scriptscriptstyle kn}'(\R-\bpi) \phi^s$. Thus the random term can be written as:
\begin{align*}
    &dF_{m_n}(\widehat{m}_{n,r}-m_{n,r})  = \sum_{s=1}^{l_1}dF_{m_n}^s(\widehat{m}^s_{n}-m^s_{n}) = \sum_{s=1}^{l_1}dF_{m_n}^s\frac{1}{n}1_{\scriptscriptstyle kn}'\bpiInv(\R-\bpi) \phi^s\\
    &  = \sum_{s=1}^{l_1}dF_{m_n}^s\frac{1}{n}1_{\scriptscriptstyle kn}'\bpiInv(\R-\bpi) \diag(\phi^s)1_{\scriptscriptstyle kn} = \frac{1}{n}\sum_{s=1}^{l_1}dF_{m_n}^s1_{\scriptscriptstyle kn}' \diag(\phi^s)\bpiInv(\R-\bpi)1_{\scriptscriptstyle kn}\\
    & = \underbrace{\left(\frac{1}{n}\sum_{s=1}^{l_1}dF_{m_n}^s1_{\scriptscriptstyle kn}' \diag(\phi^s)\right)}_{z'}\bpiInv(\R-\bpi)1_{\scriptscriptstyle kn}
\end{align*}
Notice that $z$ involves only fixed quantities, thus we have:
\begin{equation*}
    \V(dF_{m_n}(\widehat{m}_{n,r}-m_{n,r}) ) = z'\V(\bpiInv(\R-\bpi)1_{\scriptscriptstyle kn})z = z'\dmat z
\end{equation*}
\end{proof}

\subsection{Proof of Corollary \ref{C:FirstOrderExpansion}}
\begin{proof}
We prove only for the GR estimators. The theorem can be proved in an analogous way for other estimators. Recall the form of a GR estimator:
\begin{equation*}
    \widehat{\mu}_n^{\tgr} =  \frac{1}{n}\onesmat'\bpiInv \R y -  \frac{1}{n}\onesmat' \bpiInv\R\X \widehat{b}^{\twls}_n+  \frac{1}{n}\onesmat' \X \widehat{b}^{\twls}_n,
\end{equation*}
where $\widehat{b}^{\twls}_n=(\X'\Om\R\X)^{+}(\X'\Om\R y)$ with a diagonal weighting matrix $\Om$ with strictly positive entries. The linear expansion of the GR estimator can be shown to have the form $\frac{1}{n}\onesmat'y+\frac{1}{n}\onesmat'\bpiInv(\R-\pi)(y-\X'b^{\twls}_n)$.

Note the estimators depends on five sets of moments, $\widehat{m}_{y,n}=\frac{1}{n}\onesmat'\bpiInv \R y$, $\widehat{m}_{x,n}=\frac{1}{n}\onesmat' \bpiInv\R\X$, $m_{x,n}=\frac{1}{n}\onesmat' \X $, $\widehat{m}_{xy,n}=\frac{1}{n}\X'\Om\R y$ and $\widehat{m}_{xx,n}=\frac{1}{n}\X'\Om\R \X$. We also define $m_{y,n}=\frac{1}{n}\onesmat'y$, $m_{xx,n}=\frac{1}{n}\X'\Om\bpi \X$ and $m_{xy,n}=\frac{1}{n}\X'\Om\bpi y$.
Consider moments $\widetilde{m}_{xx,n}$ and $\widetilde{m}_{xy,n}$ that are in a sufficiently small local neighborhood of $m_{xx,n}$ and $m_{xy,n}$. $\widetilde{m}_{xx,n}$ is invertible and has the smallest eigenvalue bounded away from 0 uniformly for large $n$. Note we have the bound:
\begin{align*}
\widetilde{b}^{\twls}_n-b^{\twls}_n&=\widetilde{m}_{xx,n}^{-1}\widetilde{m}_{xy,n}- m_{xx,n}^{-1}m_{xy,n}\\
&= (\widetilde{m}_{xx,n}^{-1}-m_{xx,n}^{-1})\widetilde{m}_{xy,n}+ (m_{xx,n})^{-1}(\widetilde{m}_{xy,n}-m_{xy,n})\\
& =\widetilde{m}_{xx,n}^{-1}\left(m_{xx,n}-\widetilde{m}_{xx,n}\right)m_{xx,n}^{-1}\widetilde{m}_{xy,n} + (m_{xx,n})^{-1}(\widetilde{m}_{xy,n}-m_{xy,n}) 
\end{align*}
Then,
\begin{equation*}
\|\widetilde{b}_n^{\twls}-b^{\twls}_n\|_2\leq C( \|m_{xx,n}-\widetilde{m}_{xx,n}\|_2 +\|\widetilde{m}_{xy,n}-m_{xy,n}\|_2),
\end{equation*}
where $C$ is a constant independent of $n$ for large $n$ by by Assumptions \ref{A:BoundedFourthMoments} and \ref{A:Invertibility}. Thus for a local linear approximation we have, for moments $\widetilde{m}_y$, $\widetilde{m}_x$, $\widetilde{m}_{xx,n}$ and $\widetilde{m}_{xy,n}$ in a sufficiently small local neighborhood of $m_y$, $m_x$, $m_{xx,n}$ and $m_{xy,n}$, 
\begin{align*}
F(\widetilde{m}_n)=&m_{x,n}\widetilde{m}_{xx,n}^{-1}\widetilde{m}_{xy,n} + \widetilde{m}_{y,n}-\widetilde{m}_{x,n}\widetilde{m}_{xx,n}^{-1}\widetilde{m}_{xy,n} \\
F(m_n)=& m_{y,n}\\
dF_{m_n}(\widetilde{m}_n^r-m_n^r)=&(\widetilde{m}_{y,n}-m_{y,n})- (\widetilde{m}_{x,n}-m_{x,n})m_{xx,n}^{-1}m_{xy,n}
\end{align*}
\begin{align*}
 & \|F(\widetilde{m}_n)-F(m_n)-dF_{m_n}(\widetilde{m}_n^r-m_n^r)\|_2\\
    & \leq   \|-(m_{x,n} - \widetilde{m}_{x,n})\widetilde{m}_{xx,n}^{-1}\widetilde{m}_{xy,n} - (m_{x,n} - \widetilde{m}_{x,n})m_{xx,n}^{-1}m_{xy,n}\|_2\\
    & \leq   \|(m_{x,n} - \widetilde{m}_{x,n})(\widetilde{b}^{\twls}_n - b^{\twls}_n)\|_2\\
    & \lesssim \|m_{x,n} - \widetilde{m}_{x,n}\|_2 + \|\widetilde{b}^{\twls}_n - b^{\twls}_n\|_2^2  \\
    & \lesssim \|m_{x,n} - \widetilde{m}_{x,n}\|_2^2 + \|m_{xx,n}-\widetilde{m}_{xx,n}\|^2_2+\|\widetilde{m}_{xy,n}-m_{xy,n}\|_2,
\end{align*}
where the constant is independent of $n$ by our assumption. This verifies Assumption \ref{A:MomentEstimators3} for the GR estimators.\\
%As commented in the Remark \ref{R:Simple}, the equivalence can be proved more directly:
%\begin{align*}
%    & (\onesmat' \onesmat)^{-1}\onesmat'\bpiInv \R y -  (\onesmat' \onesmat)^{-1}\onesmat' \bpiInv\R\X \widehat{b}^{\twls}+  (\onesmat' \onesmat)^{-1}\onesmat' \X \widehat{b}^{\twls}\\
%    & -\frac{1}{n}\onesmat' y - \frac{1}{n}\onesmat' \bpiInv(\R-\bpi)(y-\X' b^{\twls}_n )\\
%    = & (\onesmat' \onesmat)^{-1}\onesmat' \bpiInv(\R-\bpi)\X (b^{\twls}_n-\widehat{b}^{\twls}) \\
 %   = & O_p(\sqrt{\frac{\||\dmat\||_2}{n}})O_p(\sqrt{\frac{\||\dmat\||_2}{n}})\\
 %   = & O_p(\frac{\||\dmat\||_2}{n})
%\end{align*}
\end{proof}
\subsection{Proof of Theorem \ref{Thm:ConsistentVariance}}
\begin{proof}
    They are proved by using Lemma \ref{Infeasible} and Lemma \ref{Feasible}.
\end{proof}
\subsection{Proof of Corollary \ref{C:ConsistentVariance}}
We prove only for the GR estimators. The corollary can be proved in an analogous way for other estimators. The plug-in variance bound estimator for a GR estimator is: 
\begin{equation*}
    \widehat{\widetilde{\V}}(\widehat{\mu}^{\tht})= \frac{1}{n}\onesmat'\diag(y-\X \widehat{b}^{\twls}_n)\R\dtildep\R\diag(y-\X \widehat{b}^{\twls}_n)\onesmat\frac{1}{n}.
\end{equation*}
Using the notation in Theorem \ref{Thm:ConsistentVariance}, we identify $z=\onesmat'\diag(y-\X b^{\twls}_n)\in\mathbb{R}^{kn\times k}$ and $\widehat{z}=\R \onesmat'\diag(y-\X \widehat{b}^{\twls}_n)\in\mathbb{R}^{kn\times k}$. Note $\X(\widehat{b}^{\twls}_n-b^{\twls}_n)\in\mathbb{R}^{kn}$.  We have
\begin{align*}
    &\frac{1}{n}\|\widehat{z}-\R z\|_2^2\leq \frac{1}{n}\|\onesmat'\diag(\X(\widehat{b}^{\twls}_n-b^{\twls}_n))\|_2^2=\frac{1}{n}\|\X(\widehat{b}^{\twls}_n-b^{\twls}_n)\|_2^2\\
    & \leq \frac{1}{n}\lambda_{\max}(\X'\X) \|\widehat{b}^{\twls}_n-b^{\twls}_n\|_2^2  \leq  \frac{1}{n}\|\X\|_2^2\times  \|\widehat{b}^{\twls}_n-b^{\twls}_n\|_2^2  =O_p\left(\frac{\sigma_{\max}\left(\dmat\right)}{n}\right),
\end{align*}
where $\frac{1}{n}\|\X\|_2^2=O(1)$ by Assumption \ref{A:BoundedFourthMoments} and $\|\widehat{b}^{\twls}_n-b^{\twls}_n\|_2^2=O_p\left(\sqrt{\sigma_{\max}\left(\dmat\right)}n^{-\frac{1}{2}}\right)$ by Lemma \ref{LemmaWLS}. $\frac{1}{n}\|z\|_2^2=O(1)$ by Assumption \ref{A:BoundedFourthMoments} and Lemma \ref{LemmaBoundedCoefs}.

\subsection{Proof of Theorem \ref{plug-in-inference}}

Based on the premises, we have $\widehat{\nu}_n-\widehat{\nu}_n^{\tL}=o_p\left(\sqrt{\sigma_{\max}\left(\Ome\right)}/n^{0.5}\right)$. Because $\V\left(c'\widetilde{\nu}_n^{\tL}\right)\succeq \sigma_{\max}\left(\Ome\right)/n$, we have $\left(c'\widehat{\nu}_n-c’\widehat{\nu}^{\tL}_n\right)/\sqrt{\V\left(c'\widehat{\nu}_n^{\tL}\right)}=o_p(1)$. 

We have
\begin{align}
 &\left|\frac{c'\widehat{\nu}_n-c'\nu_n} 
{\sqrt{\widehat{\widetilde{\V}}\left(c'\nu_n^{\tL}\right)}}\right|= \left|\frac{c'\widehat{\nu}_n-c'\nu_n} 
{\sqrt{\widetilde{\V}\left(c'\nu_n^{\tL}\right)}}\times \sqrt{\frac{\widetilde{\V}\left(c'\nu_n^{\tL}\right)}{\widehat{\widetilde{\V}}\left(c'\nu_n^{\tL}\right)}}\right|\\
\leq & \left|\frac{c'\widehat{\nu}_n^\tL-c'\nu_n} 
{\sqrt{\V\left(c'\nu_n^{\tL}\right)}}\right|\times \sqrt{\frac{\V\left(c'\nu_n^{\tL}\right)}{\widetilde{\V}\left(c'\nu_n^{\tL}\right)}}\sqrt{\frac{\widetilde{\V}\left(c'\nu_n^{\tL}\right)}{\widehat{\widetilde{\V}}\left(c'\nu_n^{\tL}\right)}}\times \left(1+o_p(1)\right)\\
\leq &  \left|\frac{c'\widehat{\nu}_n^\tL-c'\nu_n} 
{\sqrt{\V\left(c'\nu_n^{\tL}\right)}} \right|\left(1+o_p(1)\right).
\end{align}

Hence:
\begin{align}
& \lim\sup_n \textrm{P}_n \left( \left|c'\widehat{\nu}_n-c'\nu_n\right|\geq q(\alpha) \sqrt{\widehat{\widetilde{\V}}\left(c'\nu_n^{\tL}\right)} \right)\\
\leq & \lim\sup_n\textrm{P}_n\left( \left|c'\widehat{\nu}_n^\tL-c'\nu_n\right|(1+o_p(1))\geq q(\alpha) \sqrt{\V\left(c'\nu_n^{\tL}\right)} \right) \\
\leq & \lim\sup_n \textrm{P}_n\left( \left|c'\widehat{\nu}_n^\tL-c'\nu_n\right|\geq (1+\epsilon)^{-1}q(\alpha) \sqrt{\V\left(c'\nu_n^{\tL}\right)} \right)\\
 &  + \lim\sup_n \textrm{P}_n\left( 1+o_p(1) \ge 1+\epsilon \right),
\end{align}
for any positive $\epsilon$. The second term converges to 0. Since $q(\alpha)$ is continuous, for sufficiently small $\epsilon$, there exists a $\delta$ such that $q(\alpha-\delta)\leq  (1+\epsilon)^{-1}q(\alpha)$. Hence we have,
\begin{align}
     & \lim\sup_n \textrm{P}_n\left( \left|c'\widehat{\nu}_n^\tL-c'\nu_n\right|\geq (1+\epsilon)^{-1}q(\alpha) \sqrt{\V\left(c'\nu_n^{\tL}\right)} \right)\\
     \leq  &  \lim\sup_n \textrm{P}_n\left( \left|c'\widehat{\nu}_n^\tL-c'\nu_n\right|\geq q(\alpha-\delta) \sqrt{\V\left(c'\nu_n^{\tL}\right)} \right) \leq \alpha-\delta. 
\end{align}
Since the $\epsilon$ can be made arbitrarily small, the $\delta$ can be made arbitrarily small as well, and we have
\begin{equation}
\lim\sup_n \textrm{P}_n \left( \left|c'\widehat{\nu}_n-c'\nu_n\right|\geq q(\alpha) \sqrt{\widehat{\widetilde{\V}}\left(c'\nu_n^{\tL}\right)} \right)\leq \alpha.
\end{equation}

\section{Additional Assumptions}\label{regularity_assumptions}
We use $\nabla$ to denote the total differentiation operator. For a function $f:\mathcal{X}\subset \mathbb{R}^k\to\mathbb{R}$, $\nabla_x f$ denotes the gradient function of $f$ (if it exists), $\nabla_{xx'}f$ denotes the Hessian function of $f$ and so on. We use the partial derivative notation $\frac{\partial}{\partial x}f$ to denote the partial derivative $f$ with respect to a particular argument $x$. For a set in $\Theta\subset \mathbb{R}^s$, we use the notation $\operatorname{Bd}(\Theta)$ to denote its boundary with respect to the standard topology of a Euclidean space.

We make the following assumptions to study the probabilistic properties of QMLE-GR estimators. Let $\Theta$ be a set in a finite-dimensional Euclidean space. We define the distance from a point $\theta$ to a set $\Theta$ as $d(\theta,\Theta)=\inf_{\tilde{\theta}\in \Theta}\|\tilde{\theta}-\theta\|_2$. The following assumptions are used to obtain consistency and $\sqrt{n}$ convergence of $\hat{\theta}_n$. For a vector/matrix/tensor $A$, we use $\|A\|_1$ to denote the sum of the absolute values of the entries in $A$.
\begin{assumption}{(QMLE Criterion)}\label{A:VanillaConsistency} There exists a positive integer $N$ such that for all $n\geq N$, the following conditions are satisfied:
\begin{enumerate}[label=(\roman*)]   
    \item Let $s$ be a positive integer. The parameter space $\Theta$ is a compact set in $\mathbb{R}^{s}$ with a nonempty interior. 
    \item There exists a $\theta_n\in\Theta$ and a positive $c_{\ref{A:VanillaConsistency},2}$ such that, for any $\epsilon>0$, $\inf_{\theta\in \Theta\backslash B(\theta_n,\epsilon)}\mathcal{L}_n(\theta)-\mathcal{L}_n(\theta_n)> c_{\ref{A:VanillaConsistency},2}\epsilon^2$. There exists a $\delta>0$ such that $d(\theta_n,\mathrm{Bd}(\Theta))>\delta$.
    \item $g^a(y_{ai},x_i,\theta) $ is three-times differentiable in $\theta$ for all $x_i$ and $y_{ai}$ values.
    \item For all $\theta\in\Theta$, there exists a positive $C_{\ref{A:VanillaConsistency},4}$ such that $\frac{1}{n}\sum_{a,i} \left[ g^a(y_{ai},x_i,\theta)\right]^2<C_{\ref{A:VanillaConsistency},4}$.
    \item $|g^a(y_{ai},x_i,\theta_2)-g^a(y_{ai},x_i,\theta_1)|\leq D^a(y_{ai},x_i)h(d(\theta_1,\theta_2))$ for all $\theta_1,\theta_2\in\Theta$. $h$ is a function that does not depend on $n$ and satisfies $h(t)\to 0 $ as $t\to 0$, and $D^a(\cdot,\cdot)$ is a non-negative function of $\left(y_{ai},x_{i}\right)$. There exists a $C_{\ref{A:VanillaConsistency},5}$ such that $\frac{1}{n}\sum_{a,i} D^a(y_{ai},x_i)<C_{\ref{A:VanillaConsistency},5}$.
    \item At $\theta_n$, there exists a $C_{\ref{A:VanillaConsistency},6}$ such that  $\frac{1}{n}\sum_{a,i} \|\nabla_{\theta} g^a(y_{ai},x_i,\theta_n)\|^2_2< C_{\ref{A:VanillaConsistency},6}$ and \\ $\frac{1}{n}\sum_{a,i} \|\nabla_{\theta\theta} g^a(y_{ai},x_i,\theta_n)\|^2_2< C_{\ref{A:VanillaConsistency},6}$. 
    \item There exists a $C_{\ref{A:VanillaConsistency},8}$ and an $\epsilon>0$ such that $\frac{1}{n}\mathlarger{\sum}_{a,i}$  {\small $ \sup_{\theta\in B(\theta_n,\epsilon)}$ }$\|\nabla_{\theta\theta\theta}g^a(y_{ai},x_i,\theta)\|_1<C_{\ref{A:VanillaConsistency},8}$.\footnote{$\nabla_{\theta\theta\theta}g^a(y_{ai},x_i,\theta)$ denotes the tensor of third-order derivatives of $g^a(y_{ai},xi,\theta)$ with respect to the parameter vector $\theta$. }
    \item There exist a constant $0<C_{\ref{A:VanillaConsistency},9}<\infty$ such that  $\omega_{ai}\leq C_{\ref{A:VanillaConsistency},9}$ for all $a\in[k]$, and $i\in[n]$.
    \item There exists a positive scalar $\epsilon$ such that the smallest absolute eigenvalue of the matrix \\$\frac{1}{n}\sum_{a,i}\omega_{ai} \nabla_{\theta\theta}g^a(y_{ai},x_i,\theta_n)$ is greater than $\epsilon$.
\end{enumerate}
\end{assumption}
The following assumptions on $f^a(\cdot,\theta)$, $a\in [k]$ are used to obtain the $\sqrt{n}$ equivalence of the QMLE-GR estimators to their asymptotic linear expansions.
\begin{assumption}\label{A:Imputation}
Let $N$ be a positive integer. Uniformly for all $n\geq N$, the following conditions are satisfied for all $a\in[k]$:
\begin{enumerate}[label=(\roman*)] 
    \item $f^a(x_i,\theta)$ is two-times differentiable in $\theta$ for all $x_i$ values, $i\in[n]$.
    \item There exists a $C_{\ref{A:Imputation},2}$ such that $\frac{1}{n}\sum_{a,i}(y_{ai}-f^a(x_i,\theta_n))^2\leq C_{\ref{A:Imputation},2}$. 
    \item There exists a $C_{\ref{A:Imputation},3}$ and an $\epsilon>0$ such that $\frac{1}{n}\sum_{a,i} \|\nabla_{\theta}f^a(x_i,\theta)\|^2_2<C_{\ref{A:Imputation},3}$ for all $\theta\in B(\theta_n,\epsilon)$.
    \item There exists a $C_{\ref{A:Imputation},4}$ and an $\epsilon>0$ such that $\frac{1}{n}\sum_{a,i} \sup_{\theta\in B(\theta_n,\epsilon)} \|\nabla_{\theta\theta}f^a(x_i,\theta)\|_1<C_{\ref{A:Imputation},4}$.
    \item There exists a $C_{\ref{A:Imputation},5}$ such that $\frac{1}{n}\sum_{a,i}(y_{ai}-f^a(x_{i},\theta_n))^4\leq C_{\ref{A:Imputation},5}$.
\end{enumerate}
\end{assumption}
\begin{remark}
    These conditions are standard in the literature and are discussed in \cite{andrews1992generic}. See also \cite{newey1994large}. Assumption \ref{A:VanillaConsistency}-(i) assumes the parameter space is finite dimensional, compact, and independent of $n$. Assumption \ref{A:VanillaConsistency}-(ii) is a unique identification assumption, and it can sometimes be checked by inspecting the convexity of the criterion function. Assumptions \ref{A:VanillaConsistency}-(iii), (iv), and (v) are conditions for the uniform convergence of the criterion function. They can be checked by inspecting the Taylor expansion of the criterion function coupled with appropriate moment conditions. Assumptions \ref{A:VanillaConsistency}-(vi), (vii), and (viii) are conditions for the rate of convergence of $\hat{\theta}_n$. Assumptions\ref{A:VanillaConsistency}-(vii) is a local identification condition, which requires the curvature around $\theta_n$ to be non-vanishing.  Assumption \ref{A:Imputation} is required for $\sqrt{n}$-equivalence and asymptotic variance bound estimation. It can be checked by a Taylor expansion of imputation functions coupled with appropriate moment conditions.

    We note that our conditions are more complicated than those of \cite{guo2021generalized}. \cite{guo2021generalized} considers the case of a two-arm completely randomized experiment for which exponential inequalities and stochastic equicontinuity conditions are available. To our knowledge, such conditions are not available in our setting.
\end{remark}

\begin{assumption}\label{A:GMM2new} Let $\theta_n$ be defined in Assumption \ref{A:GMM1}. Let $N$ be a positive integer.  Uniformly for all $n\geq N$, the following conditions are satisfied:
\begin{enumerate}[label=(\roman*)] 
\item (Moments)
\begin{enumerate}
   \item (Criterion Moments) For all $\theta\in\Theta$, there exists a $C_{\ref{A:GMM2new},1}$ such that $\frac{1}{n}\sum_{a,i} \left( y_{ai}-f^a(x_i,\theta)\right)^4<C_{\ref{A:GMM2new},1}$.
    \item (Derivative Moments) For all $\theta\in\Theta$,  there exists a $C_{\ref{A:GMM2new},2}$ such that $\frac{1}{n}\sum_{a,i}\left(\frac{\partial}{\partial \theta_t}f^a(x_i,\theta)\right)^4\leq C_{\ref{A:GMM2new},2}$ and  $\frac{1}{n}\sum_{a,i}\left(\frac{\partial^2}{\partial \theta_t\partial \theta_u}f^a(x_i,\theta)\right)^4\leq C_{\ref{A:GMM2new},2}$ for all $t,u\in [s]$. 
\end{enumerate}
    \item  (Lipschitz Continuity for the Function and Its Derivatives) For all $a\in [k]$.
\begin{enumerate}
    \item  $f^a(x_i,\theta)$ is two-times differentiable in $\theta$ for all $x_i$, $i\in [n]$.
    \item  (Function) For all $\theta_1,\theta_2\in\Theta$, $|f^a(x_i,\theta_2)-f^a(x_i,\theta_1)|\leq D^{a}_{\ref{A:GMM2new},3}(x_i)\|\theta_2-\theta_1\|_2$. There exists a $C_{\ref{A:GMM2new},3}$ such that $\frac{1}{n}\sum_{a,i} (D^{a}_{\ref{A:GMM2new},3}(x_i))^2\leq C_{\ref{A:GMM2new},3}$.
     \item (First Derivatives)  For all $\theta_1,\theta_2\in\Theta$, $|\frac{\partial}{\partial \theta_t}f^a(x_i,\theta_2)- \frac{\partial}{\partial \theta_t}f^a(x_i,\theta_1)|\leq D^{a}_{\ref{A:GMM2new},4}(x_i) \|\theta_2-\theta_1\|_2$ for all $t\in [s]$. There exists a $C_{\ref{A:GMM2new},4}$ such that $\frac{1}{n}\sum_{a,i}(D^{a}_{\ref{A:GMM2new},4}(x_i))^2 \leq C_{\ref{A:GMM2new},4}$. 
\end{enumerate}    
    \item (Taylor Approximations) For all $a\in [k]$,
\begin{enumerate}
    \item (First-Order Approximation for the Function) There exists an $\epsilon$ such that for all $\theta\in B(\theta_n,\epsilon)$, $|f^a(x_i,\theta)- f^a(x_i,\theta_n)-\nabla_{\theta}f^a(x_i,\theta_n)'(\theta-\theta_n)|\leq D^{a}_{\ref{A:GMM2new},6}(x_i) \|\theta-\theta_n\|^2_2$.  There exists a $C_{\ref{A:GMM2new},6}$ such that $\frac{1}{n}\sum_{a,i}(D^{a}_{\ref{A:GMM2new},6}(x_i))^2 \leq C_{\ref{A:GMM2new},6}$.
    \item (First-Order Approximation for the First Derivative) There exists an $\epsilon$ such that for all $\theta\in B(\theta_n,\epsilon)$,  $|\frac{\partial}{\partial \theta_t}f^a(x_i,\theta)- \frac{\partial}{\partial \theta_t}f^a(x_i,\theta_n)-\nabla_{\theta} (\frac{\partial}{\partial \theta_t}f^a(x_i,\theta_n))'(\theta-\theta_n)|\leq D^{a}_{\ref{A:GMM2new},7}(x_i) \|\theta-\theta_n\|^2_2$ for all $t\in [s]$. There exists a $C_{\ref{A:GMM2new},7}$ such that $\frac{1}{n}\sum_{a,i}(D^{a}_{\ref{A:GMM2new},7}(x_i))^2 \leq C_{\ref{A:GMM2new},7}$.
\end{enumerate}
\item Denote $c'\onesmat' \diag\left(f(\theta_n)\right)$ by $f^c(\theta_n)\in\mathbb{R}^{kn}$ and $c'\onesmat\diag(y)$ by $y^c\in\mathbb{R}^{kn}$. Define the matrix $\mathcal{H}_n\in\mathbb{R}^{s\times s}$ to be the column stack of the vectors:
\begin{equation}
    \frac{1}{n}\left(y^c-f^c(\theta_n)\right)\dmat \nabla_{\theta}\left(\frac{\partial}{\partial \theta_t}f^c(\theta_n)\right) \in \mathbb{R}^{s}, t\in [s]. 
\end{equation}
There exists a positive scalar $\epsilon$ such that the smallest absolute value of the eigenvalues of the matrix $    \nabla_{\theta\theta}\mathcal{L}\left(\theta_n\right) \in \mathbb{R}^{s\times s}$, defined as,
\begin{equation}
    \nabla_{\theta\theta}\mathcal{L}\left(\theta_n\right) = -\mathcal{H}_n + \frac{1}{n}\left(\nabla_{\theta} f^c\left(\theta_n\right)\right)'\dmat \nabla f^c\left(\theta_n\right),
\end{equation}
is greater than $\epsilon$.
\end{enumerate}
\end{assumption}
\begin{remark}
Assumption \ref{A:GMM2new} can be checked by a Taylor expansion of the imputation functions coupled with appropriate moment conditions. 
\begin{comment}
2) Always inspect the eigenvalues of $\X'\Ome\X$; if the eigenvalues are very close to zero, try removing some covariate columns;  3) One could add some penalty terms onto the criterion, for example $\mathcal{L}_n(\theta)+\|\theta\|^2_2$. Solution to this problem is no longer optimal. However, one may gain identification and improved finite-sample performance in exchange for a loss of efficiency.
\end{comment}
\end{remark}
\section{Proofs for the results in Section \ref{Section:Nonlinear}}\label{more_results_S5}
\subsection{Proof of Theorem \ref{Thm:QMLE}}
We first check the stochastic equicontinuity for large $n$. Notice by Assumption \ref{A:VanillaConsistency}-(iv), $\sigma_{\max}\left(\Ome\right)=O(1)$ and Lemma \ref{HTcp1} we have pointwise convergence for the criterion function: for each $\theta\in\Theta$,
\begin{equation}
    \frac{1}{n}\sum_{a=1}^k\sum_{i=1}^n\frac{\R_{ai}}{\pi_{ai}}\omega_{ai}g^a( y_{i}(a),x_i,\theta) -   \frac{1}{n}\sum_{a=1}^k\sum_{i=1}^n \omega_{ai} g^a( y_{i}(a),x_i,\theta)=o_p(1).
\end{equation}
By Assumption \ref{A:VanillaConsistency}-(v), $\mathcal{L}_n(\theta)$ is continuous uniformly over $\theta\in\Theta$ and for large $n$. Then,
\begin{align}
    & \mathbf{P}_n\left(\sup_{\theta\in\Theta}\sup_{\theta'\in B(\theta,\delta)}\left| \widehat{\mathcal{L}}_n(\theta)-\mathcal{L}_n(\theta)-\left( \widehat{\mathcal{L}}_n(\theta')-\mathcal{L}_n(\theta')\right)\right|>3\epsilon\right) \\
     \leq  & \mathbf{P}_n\left(\sup_{\theta\in\Theta}\sup_{\theta'\in B(\theta,\delta)}\left|\widehat{\mathcal{L}}_n(\theta)-\widehat{\mathcal{L}}_n(\theta')\right|>2\epsilon\right) \\
     & + \mathbf{P}_n\left(\sup_{\theta\in\Theta}\sup_{\theta'\in B(\theta,\delta)}\left|\mathcal{L}_n(\theta)-\mathcal{L}_n(\theta')\right|>\epsilon\right).    
\end{align}
First note that the second term is a degenerate probability event: it happens with probability 0 for sufficiently small $\delta$. For the first term, we take a $2\delta$-covering $\mathcal{N}_{2\delta}$ of $\Theta$. Note $\mathcal{N}_{2\delta}$ is a finite set by Assumption \ref{A:VanillaConsistency}-(i). By the triangle inequality, for each $\theta$ there exists a $\tilde{\theta}\in\mathcal{N}_{2\delta}$ such that $d(\theta',\tilde{\theta})<2\delta$ for all $\theta'\in B(\theta,\delta)$. Thus we can bound the first term:
\begin{align}\label{SC}
     &\mathbf{P}_n\left(\sup_{\theta\in\Theta}\sup_{\theta'\in B(\theta,\delta)}\left|\widehat{\mathcal{L}}_n(\theta)-\widehat{\mathcal{L}}_n(\theta')\right|>2\epsilon\right) \\
    =   &  \mathbf{P}_n\left(\sup_{\theta\in\Theta}\sup_{\theta'\in B(\theta,\delta)}\left|\widehat{\mathcal{L}}_n(\theta)-\widehat{\mathcal{L}}_n(\tilde{\theta})+\widehat{\mathcal{L}}_n(\tilde{\theta})-\widehat{\mathcal{L}}_n(\theta')\right|>2\epsilon\right)\\
    \leq & 2\mathbf{P}_n\left(\frac{1}{n}\sum_{a=1}^k\sum_{i=1}^n \frac{\R_{ai}}{\pi_{ai}}\omega_{ai}D( y_{i}(a),x_i)h(2\delta)>\epsilon\right)\\
    \leq &  \frac{2C_{\ref{A:VanillaConsistency},9}}{\epsilon}\frac{1}{n}\sum_{a,i}D( y_{i}(a),x_i)h(2\delta) 
\end{align}
where the last line is by the Markov inequality and Assumption \ref{A:VanillaConsistency}-(ix). By setting $\delta$ small enough, we prove the desired inequality with Assumption \ref{A:VanillaConsistency}-(v). With Lemma \ref{Lemma:SC} and Lemma \ref{Lemma:UC}, we conclude $\widehat{\theta}_n-\theta_n=o_p(1)$. We now show $\widehat{\theta}_n-\theta_n=o_p(n^{-\frac{1}{2}})$.

By Assumption \ref{A:VanillaConsistency}-(ii), $\widehat{\theta}_n$ is in the interior of $\Theta$ with probability approaching one. Hence with probability approaching one, $\widehat{\theta}_n$ satisfies:
\begin{equation}
    \nabla_{\theta}\widehat{\mathcal{L}}_n(\widehat{\theta}_n)=\frac{1}{n}\sum_{a=1}^k\sum_{i=1}^n \frac{\R_{ai}}{\bpi_{ai}}\omega_{ai}\nabla_{\theta}g^a( y_{i}(a),x_i,\widehat{\theta}_n)=0
\end{equation}
A Taylor expansion around $\theta=\theta_n$ yields:
\begin{align}
    0=&\frac{1}{n}\sum_{a=1}^k\sum_{i=1}^n \frac{\R_{ai}}{\bpi_{ai}}\omega_{ai}\nabla_{\theta}g^a( y_{i}(a),x_i,\theta_n) \label{Taylor1}\\
    +& \frac{1}{n}\sum_{a=1}^k\sum_{i=1}^n \frac{\R_{ai}}{\bpi_{ai}}\omega_{ai}\nabla_{\theta\theta}g^a( y_{i}(a),x_i,\theta_n)(\widehat{\theta}_n-\theta_n) 
     + o\left(\|\widehat{\theta}_n-\theta_n\|\right)\label{Taylor2}.
\end{align}    
The last term is justified by bounding the higher-order reminder terms as follows: take the $t$th entry of the gradient $ \nabla_{\theta}\widehat{\mathcal{L}}_n$. This row corresponds to the partial derivative of $\widehat{\mathcal{L}}_n$ with respect to the $t$th parameter $\theta_t$.  Its Taylor expansion has the form:
\begin{align}
       0=&\frac{1}{n}\sum_{a=1}^k\sum_{i=1}^n \frac{\R_{ai}}{\bpi_{ai}}\omega_{ai}\frac{\partial}{\partial \theta_t}g^a( y_{i}(a),x_i,\theta_n) \\
       & + \frac{1}{n}\sum_{a=1}^k\sum_{i=1}^n \frac{\R_{ai}}{\bpi_{ai}}\omega_{ai}\nabla_{\theta}'\frac{\partial}{\partial \theta_t}g^a( y_{i}(a),x_i,\theta_n)(\widehat{\theta}_n-\theta_n)\\
       &+ \frac{1}{n}(\widehat{\theta}_n-\theta_n)'\sum_{a=1}^k\sum_{i=1}^n\frac{\R_{ai}}{\bpi_{ai}}\omega_{ai}\nabla_{\theta\theta}\frac{\partial}{\partial \theta_t}g^a( y_{i}(a),x_i,\tilde{\theta}_n)(\widehat{\theta}_n-\theta_n)    
\end{align}
where $\tilde{\theta}_n$ is between $\theta_{n}$ and $\hat{\theta}_{n}$. As $\widehat{\theta}_n$ enters $B(\theta_n,\epsilon)$ with probability one and by the subadditivity of the spectral norm $\viii{\cdot}$, we have,
\begin{align}\label{15}
        & \viii{\frac{1}{n}\sum_{a=1}^k\sum_{i=1}^n\frac{\R_{ai}}{\bpi_{ai}}\omega_{ai}\nabla_{\theta\theta}\frac{\partial}{\partial \theta_t}g^a( y_{i}(a),x_i,\tilde{\theta}_n)}\\
        \leq &   \frac{1}{n}\sum_{a=1}^k\sum_{i=1}^n\frac{\R_{ai}}{\bpi_{ai}}\omega_{ai}\viii{\nabla_{\theta\theta}\frac{\partial}{\partial \theta_t}g^a( y_{i}(a),x_i,\tilde{\theta}_n)}\\
        \leq &   \frac{1}{n}\sum_{a=1}^k\sum_{i=1}^n\frac{\R_{ai}}{\bpi_{ai}}\omega_{ai}\|\nabla_{\theta\theta}\frac{\partial}{\partial \theta_t}g^a( y_{i}(a),x_i,\tilde{\theta}_n)\|_2 \label{115}\\        
    \leq &  \frac{1}{n}\sum_{a=1}^k\sum_{i=1}^n\frac{\R_{ai}}{\bpi_{ai}}\omega_{ai} \sup_{\theta\in B(\theta_n,\epsilon)}\|\nabla_{\theta\theta\theta}g^a( y_{i}(a),x_i,\tilde{\theta}_n)\|_1 =O_p(1)\label{116}
\end{align}
where for (\ref{115}) we use the fact that the spectral norm is dominated by the Frobenius norm, and for (\ref{116}) we used the the equivalence of the $l_1$ and $l_2$ norm, Markov inequality, Assumption \ref{A:VanillaConsistency}-(vii), Assumption \ref{A:VanillaConsistency}-(viii) and $\sigma_{\max}\left(\Ome\right)=O(1)$.  Also by Assumption \ref{A:VanillaConsistency}-(vi), Assumption \ref{A:VanillaConsistency}-(viii) and Lemma \ref{HTcp1}, $\frac{1}{n}\sum_{a,i} \omega_{ai}\frac{\R_{ai}}{\bpi_{ai}}\nabla_{\theta\theta}g^a( y_{i}(a),x_i,\theta_n)$ converges in probability to $\frac{1}{n}\sum_{a,i}\omega_{ai} \nabla_{\theta\theta}g^a( y_{i}(a),x_i,\theta_n)$. Hence it is invertible with probability approaching one by Assumption \ref{A:VanillaConsistency}-(ix). Using the above argument and the fact that $\widehat{\theta}_n-\theta_n=o_p(1)$, we rearrange (\ref{Taylor1}) and (\ref{Taylor2}) to get to
\begin{align}
    &\left(\I_{s}+o_p(1)\right)\left(\widehat{\theta}_n-\theta_n\right) \\
    & = -\left(\frac{1}{n}\sum_{a,i}\frac{\R_{ai}}{\bpi_{ai}}\omega_{ai} \nabla_{\theta\theta}g^a( y_{i}(a),x_i,\theta_n)\right)^{-1}\frac{1}{n}\sum_{a=1}^k\sum_{i=1}^n \frac{\R_{ai}}{\bpi_{ai}}\omega_{ai}\nabla_{\theta}g^a( y_{i}(a),x_i,\theta_n) \\
    = & O_p\left(\frac{1}{\sqrt{n}}\right),    
\end{align}
where for the last line we used the first order condition $\frac{1}{n}\sum_{i=1}^n\omega_{ai}\nabla_\theta g^a( y_{i}(a),x_i,\theta_n)=0$, Assumption \ref{A:VanillaConsistency}-(vi), Assumption \ref{A:VanillaConsistency}-(viii), Lemma \ref{HTcp1} and $\sigma_{\max}\left(\Ome\right)=O(1)$. 

Because $\sigma_{\max}((\widetilde{\dmat}\otimes \widetilde{\dmat})\circ \s)=o(n)$,
$\sigma_{\max}\left(\dtildep\right)=O(1)$ and $n\V(\widehat{\mu}_{n,c}^{\tqmle,\tL})\geq c_{\ref{Thm:QMLE}}$ uniformly for all large $n$, by Assumption \ref{A:Imputation}, Lemma \ref{Lemma:Equivalence}, Lemma \ref{Infeasible} and Lemma \ref{Feasible}, and the fact that
\begin{equation}
    \|\R(y-f(\theta_n)) - \R(y-f(\widehat{\theta}_n))\|_2 \leq \| f(\theta_n)-f(\widehat{\theta}_n)\|_2,
\end{equation}
where $f(\theta)$ is the vector of imputed outcomes of all arms using coefficient $\theta$, we have,
\begin{equation}
\frac{\widehat{\widetilde{\V}}\left(\widehat{\mu}_{n,c}^{\tqmle,\tL}\right)-\widetilde{\V}\left(\widehat{\mu}_{n,c}^{\tqmle,\tL}\right)  }{\widetilde{\V}\left(\widehat{\mu}_{n,c}^{\tqmle,\tL}\right)  }=o_p(1),   
\end{equation}
and 
\begin{equation}
    \frac{\widehat{\mu}_{n,c}^{\tqmle}-\widehat{\mu}_{n,c}^{\tqmle,\tL}}{\sqrt{\widetilde{\V}\left(\widehat{\mu}_{n,c}^{\tqmle,\tL}\right)}} = \frac{1}{\sqrt{n\widetilde{\V}\left(\widehat{\mu}_{n,c}^{\tqmle,\tL}\right)}}\sqrt{n}\left(\widehat{\mu}_{n,c}^{\tqmle,\tL}-\widehat{\mu}_{n,c}^{\tqmle}\right)=o_p(1).
\end{equation}
The remainder of the proofs is similar to that for Theorem \ref{plug-in-inference}.
\subsection{Proof of Theorem \ref{Thm:NOHARM}}
We first show that
\begin{equation}
 \alpha_n^c= \frac{\frac{1}{n}c'\onesmat'\diag\left(y\right)\dmat\diag\left(f({\theta}_n)\right)\onesmat c}{\frac{1}{n}c'\onesmat'\diag\left(f({\theta}_n)\right)\dmat\diag\left(f({\theta}_n)\right)\onesmat c} 
\end{equation}
The numerator is upper bounded by
\begin{align}
     & \frac{1}{n}c'\onesmat'\diag(y)\dmat\diag(f({\theta}_n)\onesmat c\\
     \leq & \sigma_{\max}\left(\dmat\right) \times \sqrt{\frac{1}{n}\sum_{a,i} \left(c_{a}y_{i}(a)\right)^2} \times \sqrt{\frac{1}{n}\sum_{a,i} \left(c_{a}f^a(x_i,\theta_n)\right)^2} ,  
\end{align}
is bounded above uniformly in $n$ by Assumption \ref{A:BoundedFourthMoments}, $\sigma_{\max}\left(\Ome\right)=O(1)$, and Assumption \ref{A:Imputation}-(ii). The denominator is bounded below uniformly in $n$ by Assumption \ref{A:NOHARM}. Thus the imputation functions $\alpha f^a(,\theta)$, $a\in [k]$, satisfy Assumption \ref{A:Imputation}.\\
Estimator for the denominator converges at a $\sqrt{n}-$rate by Lemma \ref{Lemma:bound},
\begin{align*}
   & \frac{1}{n}c'\onesmat'f\left(\widehat{\theta}_n\right)\dmat f\left(\widehat{\theta}_n\right)\onesmat c -     \frac{1}{n}c'\onesmat'f\left(\theta_n\right)\dmat f\left(\theta_n\right)\onesmat c \\
  \leq   & 2\frac{1}{n}c'\onesmat'\left(f(\widehat{\theta}_n)-f(\theta_n)\right)\dmat f(\theta_n)\onesmat c \\
     & +  \frac{1}{n}c'\onesmat'\left(f(\widehat{\theta}_n)-f(\theta_n)\right)\dmat\left(f(\widehat{\theta}_n)-f(\theta_n)\right)\onesmat c \\
   \leq  & \|c\|_{\infty}\times \sigma_{\max}\left(\dmat\right) \times \frac{1}{n}\sum_{a,i}\left(f^a(x_i,\widehat{\theta}_n)-f^a(x_i,\theta_n)\right)^2 \\
    &+  \|c\|^2_{\infty} \times \sigma_{\max}\left(\dmat\right)\times \sqrt{\frac{1}{n}\sum_{a,i}(f^a(x_i,\theta_n))^2} \times \sqrt{\frac{1}{n}\sum_{a,i}\left(f^a(x_i,\widehat{\theta}_n)-f^a(x_i,\theta_n)\right)^2}\\
    = & O_p\left(\frac{1}{\sqrt{n}}\right),
\end{align*}
by Lemma \ref{Lemma:Equivalence}. The numerator also converges at a $\sqrt{n}-$rate. We first show:
\begin{align*}
   & \frac{1}{n}c'\onesmat'f\left(\widehat{\theta}_n\right)\dmat \bpiInv \R \diag(y)\onesmat c -     \frac{1}{n}c'\onesmat'f(\theta_n)\Ome\diag(y)\onesmat c \\
   =& \frac{1}{n}c'\onesmat'\left(f(\widehat{\theta}_n)-f(\theta_n)\right)\dmat \bpiInv \R \diag(y)\onesmat c \\
   & + \frac{1}{n}c'\onesmat'f(\theta_n)\dmat  \bpiInv \R\diag(y)\onesmat c  -\frac{1}{n}c'\onesmat'f(\theta_n)\dmat \diag(y)\onesmat c\\
   \leq &  \|c\|^2_{\infty} \times \sigma_{\max}\left(\dmat\right)\times \sqrt{\frac{1}{n}\sum_{a,i} y_{i}(a)^2} \times \sqrt{\frac{1}{n}\sum_{a}\left(f^a(x_i,\widehat{\theta}_n)-f^a(x_i,\theta_n)\right)^2} \\
   & +  \frac{1}{n}c'\onesmat'f(\theta_n)\dmat  \bpiInv (\R-\bpi)\diag(y)\onesmat c.
\end{align*}
Denote $c'\onesmat' f(\theta_n)$ by $f^c(\theta_n)\in\mathbb{R}^{kn}$ and $c'\onesmat\diag(y)$ by $y^c\in\mathbb{R}^{kn}$. The term $\frac{1}{n}c'\onesmat'f(\theta_n)\dmat  \bpiInv (\R-\bpi)\diag(y)\onesmat c$ is $O_p(n^{-\frac{1}{2}})$ by noticing:
\begin{align*}
& \frac{1}{n^2}\V\left(f^c(\theta_n)\dmat  \bpiInv (\R-\bpi)y^c\right)=\frac{1}{n^2}\V\left(\ones{kn}\diag(y^c)\bpiInv \R\dmat f^c(\theta_n)\right)\\
=&  \frac{1}{n^2}\V\left(\ones{kn}\bpiInv \R\diag(y^c)\dmat f^c(\theta_n)\right)\\
=& \frac{1}{n^2} f^c(\theta_n)'\dmat \diag(y^c)\dmat\diag(y^c)\dmat f^c(\theta_n)\\
\leq &  \frac{\viiii{\dmat}_1^2\sigma_{\max}\left(\dmat\right)}{n} \sqrt{\frac{1}{n}\sum_{a,i} (c^a y_{i}(a))^4}\times  \sqrt{\frac{1}{n}\sum_{a,i} (c^af^a(x_i,\theta_n))^4}
\end{align*}
and by Lemma \ref{Lemma:bound}, Assumption \ref{A:BoundedFourthMoments}, and Assumption \ref{A:Imputation}. Together with Assumption \ref{A:NOHARM} and the fact that $\viiii{\dmat}_1=O(1)$, this implies $\hat{\alpha}_n^c-\alpha_n^c=O_p(n^{-\frac{1}{2}})$. The remaining proof is the same as the proof in Theorem 
\ref{Thm:QMLE}.

\subsection{Proof of Theorem \ref{Thm:OC}}
Notice first:
\begin{itemize}
    \item $\left(1-\pi_{ai}\right)/\pi_{ai}\leq \sigma_{\max}\left(\dmat\right)$ for all $a\in[k]$ and $i\in[n]$ because $\dmat$ is positive-semidefinite and $\V\left(\R_{ai}/\pi_{ai}\right)=\left(1-\pi_{ai}\right)/\pi_{ai}$ is on the diagonal of $\dmat$. This implies $\max_{a,i}\{1/\pi_{ai}\}\leq \sigma_{\max}\left(\dmat\right)+1$. 
    \item $\sigma_{\max}\left(\dmat\right)\leq\viiii{\dmat}_1$ by Lemma 5.6.10 in \cite{horn2012matrix}.
\end{itemize}
We only need to check the $\sqrt{n}$-consistency $\widehat{\theta}_n$. The rest proofs are identical to those of Theorem \ref{Thm:QMLE}.
We first check the pointwise convergence of the criterion, namely
\begin{equation*}
    \widehat{\mathcal{L}}_n(\theta)- \mathcal{L}_n(\theta)=o_p(1),
\end{equation*}
pointwise in $\theta\in\Theta$.
We adopt the notation $y^c$ and $f^c\left(\theta\right)$ in the proof of Theorem \ref{Thm:NOHARM}. 
\begin{equation}
    n\widehat{\mathcal{L}}_n(\theta) = -2\left(y^c\R\bpiInv\right)'\dmat  f^c(\theta) + \left(f^c(\theta)\right)'\dmat f^c(\theta).
\end{equation}
\begin{equation}
    n \mathcal{L}_n(\theta) = -2\left(y^c\right)'\dmat f^c(\theta) + \left(f^c(\theta)\right)'\dmat f^c(\theta).   
\end{equation}
First note that $  \widehat{\mathcal{L}}_n(\theta)$ is an unbiased estimator of $ \mathcal{L}_n(\theta)$. 
The variance of $  \widehat{\mathcal{L}}_n(\theta)$ is: 
\begin{align*}
    &\frac{4}{n^2}\V\left( \left(y^c\R\bpiInv\right)' \dmat f^c(\theta)   \right)= \frac{4}{n^2} \V\left(1_{\scriptscriptstyle kn}'\diag\left(y^c\right)\R\bpiInv\Ome  f^c(\theta)\right)\\
    = &   \frac{4}{n^2} \V\left(1_{\scriptscriptstyle kn}'\bpiInv\R\diag(y^c)\Ome f^c(\theta)\right)=  \frac{4}{n^2}f^c(\theta)'\Ome'\diag\left(y^c\right)\dmat \diag\left(y^c\right)\Ome  f^c(\theta)\\
    \leq &   \sigma_{\max}\left(\dmat\right) \times \left(\frac{4}{n^2} f^c(\theta)'\Ome'\diag\left(y^c\right)\diag\left(y^c\right)\Ome f^c(\theta)\right)\\
     \leq &  \frac{4\viiii{\Ome}_1^2\sigma_{\max}\left(\dmat\right)}{n} \sqrt{\frac{1}{n}\sum_{a,i} \left(c_a y_{i}(a)\right)^4}\times  \sqrt{\frac{1}{n}\sum_{a,i} (c_a f^a(x_i,\theta))^4}
\end{align*}
where for the last inequality we use Lemma \ref{Lemma:bound}. Because $\viiii{\Ome}_1=O(1)$ and $\sigma_{\max}\left(\dmat\right)=O(1)$ and under Assumption \ref{A:BoundedFourthMoments} and Assumption \ref{A:GMM2new}-(i)-(a), the term above is $o_p(1)$. \\
Next we check stochastic equicontinuity as in the proof of Theorem \ref{Thm:QMLE}:
\begin{align*}
 &\frac{1}{n}\left(y^c\bpiInv\R\right)'\dmat f^c(\theta_1) - \frac{1}{n}\left(y^c\bpiInv\R\right)'\dmat  f^c(\theta_2) \\
 = & \frac{1}{n}\left(y^c\bpiInv\R\right)'\dmat \left(f^c(\theta_1) -  f^c(\theta_2)\right) \\
 \leq &  \sigma_{\max}\left(\dmat\right) \times \max_{a,i}\{\frac{1}{\pi_{ai}}\} \sqrt{\frac{1}{n}\sum_{a,i}\left(c_ay_i(a)\right)^2} \times \sqrt{\frac{1}{n}\sum_{a,i}(c_af^a(x_i,\theta_1)-c_af^a(x_i,\theta_2))^2}\\
\leq &  C \|\theta_1-\theta_2\|_2
\end{align*}
for a constant $C$ bounded above uniformly for large $n$ by Assumption \ref{A:BoundedFourthMoments} and Assumption
\ref{A:GMM2new}-(ii)-(b). Using the same steps as in Theorem \ref{Thm:QMLE}, we can prove stochastic equicontinuity upon noticing
\begin{align}
    & \textrm{P}_n\left(|\widehat{\mathcal{L}}_n(\theta_1)-\widehat{\mathcal{L}}(\theta_2)|>2\epsilon\right)  \leq  \textrm{P}_n\left(C\|\theta_2-\theta_1\|_2>2\epsilon \right) \to 0 
\end{align}
as $\|\theta_1-\theta_2\|_2\to 0$. Further with Assumption \ref{A:GMM1}-(ii) and by Lemma \ref{Lemma:SC} and Lemma \ref{Lemma:UC}, we have $\widehat{\theta}_n-\theta_n=o_p(1)$.

Now we establish the rate of convergence. We define $\frac{\partial}{\partial \theta_t}f^c(\theta)\equiv c'\onesmat'\diag(\frac{\partial}{\partial \theta_t} f(\theta))\in \mathbb{R}^{kn}$.  Let $\nabla_{\theta} f^c(\theta)\in \mathbb{R}^{kn\times s}$ denote the column stacks of $\frac{\partial}{\partial \theta_t}f^c(\theta), t\in [s]$ .
By Assumption \ref{A:GMM1}-(ii), with probably approaching one, $\widehat{\theta}_n$ is in the interior of $\Theta$ and satisfies the first order condition:\footnote{We omit the constant factor $2$.}
\begin{equation}\label{eq:25}
  -\frac{1}{n} \left(y^c\R\bpiInv\right)'\dmat \nabla_{\theta} f^c\left(\widehat{\theta}_n\right) + \frac{1}{n} f^c\left(\widehat{\theta}_n\right)\dmat   \nabla_{\theta} f^c\left(\widehat{\theta}_n\right)=0
\end{equation}
For the $t$th entry of the equalities above, it can be written as:
\begin{equation}
  -\frac{1}{n}  \left(y^c\R\bpiInv\right)'\dmat \frac{\partial}{\partial \theta_t} f^c\left(\widehat{\theta}_n\right) + \frac{1}{n} f^c\left(\widehat{\theta}_n\right)\dmat \frac{\partial}{\partial \theta_t} f^c\left(\widehat{\theta}_n\right)=0 \label{entry_t}
\end{equation}
We expand:
\begin{align*}
    0 = &   -\frac{1}{n}  \left(y^c\R\bpiInv\right)'\dmat \frac{\partial}{\partial \theta_t} f^c\left(\theta_n \right) + \frac{1}{n}  f^c\left(\theta_n\right)\dmat \frac{\partial}{\partial \theta_t}f^c\left(\theta_n\right)\\
     & -    \underbrace{\frac{1}{n}\left(y^c\R\bpiInv\right)'\dmat \left(\frac{\partial}{\partial \theta_t} f^c\left( \widehat{\theta}_n \right)- \frac{\partial}{\partial \theta_t} f^c\left(\theta_n \right)\right) }_{(A)}\\
    &+ \underbrace{ \frac{1}{n}   
\left(f^c\left(\theta_n\right)\right)'\dmat \left(\frac{\partial}{\partial \theta_t} f^c\left(\widehat{\theta}_n\right)- \frac{\partial}{\partial \theta_t}f^c\left(\theta_n\right)\right) }_{(B)} + \underbrace{   \frac{1}{n} \left(f^c\left(\widehat{\theta}_n\right)-f^c\left(\theta_n\right)\right)'\dmat \frac{\partial}{\partial \theta_t} f^c\left(\theta_n\right) }_{(C)}\\   
    & + \underbrace{   \frac{1}{n} \left(f^c\left(\widehat{\theta}_n\right)-f^c\left(\theta_n\right)\right)'\dmat \left(\frac{\partial}{\partial \theta_t} f^c\left( \widehat{\theta}_n\right)-\frac{\partial}{\partial \theta_t} f^c\left( \theta_n\right) \right) }_{(D)}
\end{align*}
We study (A), (B), (C), and (D). For (A), we have the decomposition:
\begin{align}
   &  \frac{1}{n}\left(y^c\R\bpiInv\right)'\dmat \left(\frac{\partial}{\partial \theta_t} f^c\left( \widehat{\theta}_n \right)- \frac{\partial}{\partial \theta_t} f^c\left(\theta_n \right)\right)\\
  =  & \frac{1}{n}\left(y^c\R\bpiInv\right)'\dmat \left(\frac{\partial}{\partial \theta_t} f^c\left( \widehat{\theta}_n \right)- \frac{\partial}{\partial \theta_t} f^c\left(\theta_n \right)-\left(\nabla_{\theta}\frac{\partial}{\partial t}f^c\left(\theta_n\right)\right)\left(\widehat{\theta}_n -\theta_n\right) \right) \\
  & +  \frac{1}{n}\left(y^c\R\bpiInv\right)'\dmat\left( \left(\nabla_{\theta}\frac{\partial}{\partial \theta_t}f^c\left(\theta_n\right)\right)\left(\widehat{\theta}_n -\theta_n\right)\right)\\
  & = O\left( \|\widehat{\theta}_n-\theta_n\|_2^2\right) + \frac{1}{n}\left(y^c\R\bpiInv\right)'\dmat\left( \left(\nabla_{\theta}\frac{\partial}{\partial \theta_t}f^c\left(\theta_n\right)\right)\left(\widehat{\theta}_n -\theta_n\right)\right),
\end{align}
where the last line is by Assumption \ref{A:BoundedFourthMoments} , Assumption \ref{A:GMM2new}-(iii)-(b) and that $\sigma_{\max}\left(\Ome\right)=O(1)$. 

To study (B), we have the decomposition:
\begin{align}
& \frac{1}{n}   \left(f^c\left(\theta_n\right)\right)'\dmat \left(\frac{\partial}{\partial \theta_t} f^c\left(\widehat{\theta}_n\right)- \frac{\partial}{\partial \theta_t}f^c\left(\theta_n\right)\right) \\
 = & \frac{1}{n}  \left(f^c\left(\theta_n\right)\right)'\dmat \left(\frac{\partial}{\partial \theta_t} f^c\left(\widehat{\theta}_n\right)- \frac{\partial}{\partial \theta_t}f^c\left(\theta_n\right)-\left(\nabla_{\theta}\frac{\partial}{\partial \theta_t}f^c\left(\theta_n\right)\right)\left(\widehat{\theta}_n-\theta_n\right)\right)\\
 & + \frac{1}{n}   \left(f^c\left(\theta_n\right)\right)'\dmat \left(\left(\nabla_{\theta}\frac{\partial}{\partial \theta_t}f^c\left(\theta_n\right)\right)\left(\widehat{\theta}_n-\theta_n\right)\right)\\
  = &  O\left( \|\widehat{\theta}_n-\theta_n\|_2^2\right) +  \frac{1}{n}   \left(f^c\left(\theta_n\right)\right)'\dmat \left(\left(\nabla_{\theta}\frac{\partial}{\partial \theta_t}f^c\left(\theta_n\right)\right)\left(\widehat{\theta}_n-\theta_n\right)\right),
\end{align}
by Assumption \ref{A:BoundedFourthMoments}, Assumption \ref{A:GMM2new}-(i)-(a), Assumption \ref{A:GMM2new}-(iii)-(b) and $\sigma_{\max}\left(\Ome\right)=O(1)$.

To study (C), we have the decomposition,
\begin{align}
  &   \frac{1}{n} \left(\frac{\partial}{\partial \theta_t} f^c\left(\theta_n\right)\right)'\dmat \left(f^c\left(\widehat{\theta}_n\right)-f^c\left(\theta_n\right)\right) \\
 = &   \frac{1}{n} \left(\frac{\partial}{\partial \theta_t} f^c\left(\theta_n\right)\right)'\dmat \left(f^c\left(\widehat{\theta}_n\right)-f^c\left(\theta_n\right)-\left(\nabla_{\theta} f^c(\theta_n)\right)\left(\widehat{\theta}_n-\theta_n\right)\right) \\
  & +  \frac{1}{n} \left(\frac{\partial}{\partial \theta_t} f^c\left(\theta_n\right)\right)'\dmat \left(\nabla_{\theta} f^c(\theta_n)\right)\left(\widehat{\theta}_n-\theta_n\right)\\
 = & O\left( \|\widehat{\theta}_n-\theta_n\|_2^2\right) + \frac{1}{n} \left(\frac{\partial}{\partial \theta_t} f^c\left(\theta_n\right)\right)'\dmat \left(\nabla_{\theta} f^c(\theta_n)\right)\left(\widehat{\theta}_n-\theta_n\right),
\end{align}
by Assumption \ref{A:GMM2new}-(i)-(b) and Assumption \ref{A:GMM2new}-(iii)-(a) and $\sigma_{\max}\left(\Ome\right)=O(1)$.

For $(D)$, we have:
\begin{align}
     & \frac{1}{n} \left(f^c\left(\widehat{\theta}_n\right)-f^c\left(\theta_n\right)\right)'\dmat \left(\frac{\partial}{\partial \theta_t} f^c\left( \widehat{\theta}_n\right)-\frac{\partial}{\partial \theta_t} f^c\left( \theta_n\right) \right)\\
    = & O\left( \|\widehat{\theta}_n-\theta_n\|_2^2\right),
\end{align}
by Assumption \ref{A:GMM2new}-(ii)-(b), Assumption \ref{A:GMM2new}-(ii)-(c) and $\sigma_{\max}\left(\Ome\right)=O(1)$.
 
To summarize, for the $t$th column of the equalities in equation (\ref{entry_t}), we have:
\begin{align}
    0 = &  -\frac{1}{n}  \left(y^c\R\bpiInv\right)'\dmat \frac{\partial}{\partial \theta_t} f^c\left(\theta_n \right) + \frac{1}{n}  f^c\left(\theta_n\right)\dmat \frac{\partial}{\partial \theta_t}f^c\left(\theta_n\right)\\
     & - \frac{1}{n}\left(y^c\R\bpiInv\right)'\dmat\left( \left(\nabla_{\theta}\frac{\partial}{\partial \theta_t}f^c\left(\theta_n\right)\right)\left(\widehat{\theta}_n -\theta_n\right)\right)\\
     & +  \frac{1}{n}   \left(f^c\left(\theta_n\right)\right)'\dmat \left(\left(\nabla_{\theta}\frac{\partial}{\partial \theta_t}f^c\left(\theta_n\right)\right)\left(\widehat{\theta}_n-\theta_n\right)\right)\\
    & + \frac{1}{n} \left(\frac{\partial}{\partial \theta_t} f^c\left(\theta_n\right)\right)'\dmat \left(\nabla_{\theta} f^c(\theta_n)\right)\left(\widehat{\theta}_n-\theta_n\right)+O\left( \|\widehat{\theta}_n-\theta_n\|_2^2\right).
\end{align}
Define $\widehat{\mathcal{H}}_n$ to be the column stack of the vectors $\frac{1}{n}\left(y^c\R\bpiInv-f^c\left(\theta_n\right)\right)\dmat \nabla_{\theta}\frac{\partial}{\partial \theta_t}f^c\left(\theta_n\right)$, $t\in [s]$.

Stacking the previous expressions together, we have the expression:
\begin{align*}
0 =&-\frac{1}{n}  \left(y^c\R\bpiInv\right)'\dmat \nabla_{\theta} f^c\left(\theta_n \right) + \frac{1}{n}  f^c\left(\theta_n\right)\dmat \nabla_{\theta} f^c\left(\theta_n\right)\\
 &- \widehat{\mathcal{H}}_n \left(\widehat{\theta}-\theta_n\right)\\
& + \frac{1}{n}\left(\nabla_{\theta} f^c\left(\theta_n\right) \right)'\dmat\left(\nabla_{\theta} f^c\left(\theta_n\right) \right)  \left(\widehat{\theta}-\theta_n\right)+O\left( \|\widehat{\theta}_n-\theta_n\|_2^2\right).
\end{align*}
Define,
\begin{align}
     &\widehat{A}_n = -\widehat{\mathcal{H}}_n + \frac{1}{n}\left(\nabla_{\theta} f^c\left(\theta_n\right) \right)'\dmat\left(\nabla_{\theta} f^c\left(\theta_n\right) \right), \\
   &  A_n = -\mathcal{H}_n + \frac{1}{n}\left(\nabla_{\theta} f^c\left(\theta_n\right) \right)'\dmat\left(\nabla_{\theta} f^c\left(\theta_n\right) \right), \\
   & \widehat{b}_n = -\frac{1}{n}  \left(y^c\R\bpiInv\right)'\dmat \nabla_{\theta} f^c\left(\theta_n \right) + \frac{1}{n}  f^c\left(\theta_n\right)\dmat \nabla_{\theta} f^c\left(\theta_n\right),
\end{align}
where $\mathcal{H}_n$ is defined in Assumption \ref{A:GMM2new}-(iv).  Note that $\E\left[\widehat{b}_n\right]=0$ by Assumption \ref{A:GMM1}-(ii), and $\widehat{b}_n=O_p(n^{-\frac{1}{2}})$ by Assumption \ref{A:BoundedFourthMoments}, Assumption \ref{A:GMM2new}-(i)-(a), Lemma \ref{Lemma:bound} and the fact that $\viiii{\dmat}_1=O(1)$.

We have $\widehat{A}_n-A_n=o_p(1)$ by Assumption \ref{A:BoundedFourthMoments}, Assumption \ref{A:GMM2new}-(i), Lemma \ref{Lemma:bound} and the fact that $\viiii{\dmat}_1=O(1)$. Further, we have $A_n=O(1)$ by Assumption \ref{A:GMM2new}-(i), and $A_n^{-1}=O(1)$ and $\widehat{A}_n^{-1}=O_p(1)$ by Assumption \ref{A:GMM2new}-(iv).
Hence we have,
\begin{align}
    \widehat{\theta}_n - \theta_n = \widehat{A}_n^{-1} \widehat{b}_n=O_p\left(\sqrt{\frac{1}{n}}\right).
\end{align}
Rest of the proofs are similar to those in Theorem \ref{Thm:QMLE}.

\subsection{Proof of Theorem \ref{Thm:OCI}}
The proof is identical to that of the Theorem \ref{Thm:NOHARM}.

\section{Proof for Results in Appendix \ref{network_experiments}}\label{proof_network_experiment}
This section provides a proof for Theorem \ref{Theorem:NetworkLinearAdj}. 
\begin{proof}
Assumptions \ref{A:networkprobability} and  \ref{A:dependencegraph} directly imply Assumptions that $\sigma_{\max}\left(\dmat\right)=O(1)$, $\viiii{\dmat}_1=O(1) $, $\sigma_{\max}((\widetilde{\dmat}\otimes \widetilde{\dmat})\circ \s)=o(n)$, and $\sigma_{\max}\left(\dtildep\right)=O(1)$. For linear models, results for the QMLE-GR estimators follow from results for the GR estimators directly with Corollaries \ref{C:Consistency}, \ref{C:FirstOrderExpansion} and \ref{C:ConsistentVariance}. Results for No-harm-GR estimators follow with Assumption \ref{A:NOHARM}. For linear regression adjustments, the optimal GR estimator has a closed form solution 
\begin{equation}
    \beta=(\X'\diag(c'\onesmat')\dmat\diag(c'\onesmat')\X)^{-1}(\X'\diag(c'\onesmat')\dmat\diag(c'\onesmat')y),
\end{equation}
for which we need Assumption $\ref{A:GMM1}$-(ii) to ensure the invertibility of the design matrix. Assumption \ref{A:GMM1}-(i) is not needed because the optimal GR estimator has a closed-form solution. Assumption \ref{A:GMM2new} is satisfied by the linearity of the imputation functions and Assumption \ref{A:BoundedFourthMoments}. \\
We now consider logistic adjustments. For notational simplicity, we redefine $\tilde{x}_i=[1,x_i]$. 
We have the following Taylor expansions:
\begin{align*}
    & g^a(y_{ai},x_i,\beta)=y_{ai}\tilde{x}_i'\beta -\log (1+\exp(\tilde{x}_i'\beta))\\
    &\nabla_{\beta} g^a(y_{ai},x_i,\beta)=y_{ai}\tilde{x}_i - \frac{\tilde{x}_i\exp(\gamma^a+x_i'\beta)}{1+\exp(\gamma^a+x_i'\beta)}\\
    &  \nabla_{\beta\beta} g^a(y_{ai},x_i,\beta)=-\tilde{x}_i\tilde{x}_i' \frac{\exp(\gamma^a+x_i'\beta)}{\left(1+\exp(\gamma^a+x_i'\beta)\right)^2}\\
    & \nabla_{\beta\beta\beta_k} g^a(y_{ai},x_i,\beta)=-\tilde{x}_{ik}\tilde{x}_i\tilde{x}_i'  \frac{\exp(\gamma^a+x_i'\beta)\left(1-\exp(\gamma^a+x_i'\beta)\right)}{\left(1+\exp(\gamma^a+x_i'\beta)\right)^3},
\end{align*}
and \begin{align*}
    & f^a(x_i,\beta) = \frac{\exp(\gamma^a+x_i'\beta)}{1+\exp(\gamma^a+x_i'\beta)}\\
    & \nabla_{\beta}f^a(x_i,\beta)=\tilde{x}_i\frac{\exp(\gamma^a+x_i'\beta)}{\left(1+\exp(\gamma^a+x_i'\beta)\right)^2}\\
    & \nabla_{\beta\beta}f^a(x_i,\beta)=\tilde{x}_i\tilde{x}_i'\frac{\exp(\gamma^a+x_i'\beta)\left(1-\exp(\gamma^a+x_i'\beta)\right)}{\left(1+\exp(\gamma^a+x_i'\beta)\right)^3}\\
    & \nabla_{\beta\beta\beta_k}f^a(x_i,\beta)=\tilde{x}_{ik}\tilde{x}_i\tilde{x}_i' \times \\
   &  \Bigg[\frac{\exp(\gamma^a+x_i'\beta)\left(1-\exp(\gamma^a+x_i'\beta)\right)}{\left(1+\exp(\gamma^a+x_i'\beta)\right)^3}
    -  \frac{\exp^2(\gamma^a+x_i'\beta)}{\left(1+\exp(\gamma^a+x_i'\beta)\right)^3}\\
    - &\frac{3\exp^2(\gamma^a+x_i'\beta)\left(1-\exp(\gamma^a+x_i'\beta\right)}{\left(1+\exp(\gamma^a+x_i'\beta)\right)^4} \Bigg]\\
     = & \tilde{x}_{ik}\tilde{x}_i\tilde{x}_i' \times \left( \frac{\exp\left(\gamma^a+x_i'\beta\right)-4\exp^2\left(\gamma^a+x_i'\beta\right)+\exp^3\left(\gamma^a+x_i'\beta\right) }{\left(1+\exp(\gamma^a+x_i'\beta)\right)^4}\right).
\end{align*}
Notice that besides the moment terms involving the covariates $\tilde{x}_i$ and $y_{ai}$, all other terms are bounded uniformly by a constant.

We now verify Assumptions \ref{A:VanillaConsistency}, \ref{A:Imputation} and \ref{A:GMM2new}. 
\subsection{Logistic Model: QMLE-GR and No-Harm GR}

\begin{itemize}
    \item Assumption \ref{A:VanillaConsistency}-(i),(ii) are assumed.
    \item Assumption \ref{A:VanillaConsistency}-(iii) follows by the form of the $g^a(\cdot)$ functions.
    \item Assumption \ref{A:VanillaConsistency}-(iv) follows by Assumption \ref{A:BoundedFourthMoments}
    upon noticing $
        |\log(1+\exp(x))|\leq |\log (2\exp(|x|))|\leq \log 2 + |x|$.
    \item Assumption \ref{A:VanillaConsistency}-(v) follows by observing:
    \begin{align*}
    & \left|g^a(y_{ai},x_i,\beta_1)-g^a(y_{ai},x_i,\beta_2)\right| = \left|\nabla_{\beta} g^a(y_{ai},x_i,\tilde{\beta})'(\beta_1-\beta_2)\right|\\
    = & \left|\left(y_{ai}\tilde{x}_i - \frac{\tilde{x}_i\exp(\tilde{x}_i'\tilde{\beta})}{1+\exp(x_i'\tilde{\beta})}\right)'(\beta_1-\beta_2)\right|\\
     \leq &  \|y_{ai}\tilde{x}_i\|_2 \times \|\beta_1-\beta_2\|_2 + \|\tilde{x}_i\|_2 \times \|\beta_1-\beta_2\|_2,
\end{align*}
where $\tilde{\beta}$ is a value between $\beta_1$ and $\beta_2$. Thus $D^a(y_{ai},x_i)=\|y_{ai}\tilde{x}_i\|_2+\|\tilde{x}_i\|_2$ and then follows by Assumption \ref{A:BoundedFourthMoments}. 
    \item Assumption \ref{A:VanillaConsistency}-(vi),(viii) follow by Assumption \ref{A:BoundedFourthMoments}  and the fact that all nonlinear terms in the first, second and third derivatives of $g$ are bounded.
    \item Assumption \ref{A:VanillaConsistency}-(viii) follows because $\omega_{ai}=1$ and $\omega_{ai}=\pi_{ai}$.    
    \item Assumption \ref{A:VanillaConsistency}-(ix) is assumed.

    \item Assumption \ref{A:Imputation}-(i) follows because $f$ is a smooth function of $\beta$ for all $x_i$ and $y_{ai}$. 
    \item Assumption \ref{A:Imputation}-(ii) follows by Assumption \ref{A:BoundedFourthMoments}.
    \item Assumption \ref{A:Imputation}-(iii),(iv),(v) follow by Assumption \ref{A:BoundedFourthMoments} and the fact that all nonlinear terms in the first, second derivatives of $f$ are bounded. 
    \item Assumption \ref{A:NOHARM} and Assumption \ref{A:OCI} are assumed.
\end{itemize}
\subsection{Logistic Model: Opt-GR}
\begin{itemize}
    \item Assumption \ref{A:GMM1}: assumed.
    \item Assumption \ref{A:GMM2new}-(i) follows because the conditions on the criterion moment is implied by Assumption \ref{A:BoundedFourthMoments} and the conditions on the derivative moments are implied by Assumption \ref{A:Bounded8thmoment}.
    \item Assumption \ref{A:GMM2new}-(ii) can be checked by inspecting the Taylor expansion. We prove this for first derivatives. Conditions for the function can be checked analogously.  By the mean value theorem
    \begin{align*}
       &\frac{\partial}{\partial \theta_t }f^a(x_i,\theta)- \frac{\partial}{\partial \theta_t}f^a(x_i,\theta_n)= \nabla_{\theta}\frac{\partial}{\partial \theta_t}f^a(x_i,\tilde{\theta}_n)'(\theta-\theta_n)\\
       \leq &  \|\nabla_{\theta}\frac{\partial}{\partial \theta_t}f^a(x_i,\tilde{\theta}_n)\|_2 \times  \|\theta-\theta_n\|_2     
    \end{align*}
    where $\tilde{\theta}_n$ is between $\theta$ and $\theta_n$. We have the inequality:
    \begin{equation*}
        \|\nabla_{\theta}\frac{\partial}{\partial \theta_t}f^a(x_i,\tilde{\theta}_n)\|^2_2\leq  C\sum_{j,k\in[s]}(\tilde{x}_{ik}\tilde{x}_{ij})^2
    \end{equation*}
    where $C$ is a constant independent of $\theta$, $\theta_n$, $x_i$ or $n$, and $\tilde{x}_{ik}$ is the $k$th entry of unit $i$'s covariate vector $\tilde{x}_i$.
Thus, we have:
    \begin{align*}
       &\frac{\partial}{\partial \theta_t}f^a(x_i,\theta)- \frac{\partial}{\partial \theta_t}f^a(x_i,\theta_n)= \nabla_{\theta}\frac{\partial}{\partial \theta_t}f^a(x_i,\tilde{\theta}_n)'(\theta-\theta_n)\\
       \leq & \sqrt{C \sum_{j,k\in[s]}(\tilde{x}_{ik}\tilde{x}_{ij})^2} \times  \|\theta-\theta_n\|_2.  
    \end{align*}
The conditions for the second derivative then follow by noticing $\frac{1}{n}\sum_{i}\sum_{j,k\in [s]}(\tilde{x}_{ik}\tilde{x}_{ij})^2$ are bounded uniformly in $n$ by Assumption \ref{A:Bounded8thmoment} and a H\"{o}lder's inequality.
\item Assumption \ref{A:GMM2new}-(iii) can be proved similarly as Assumption \ref{A:GMM2new}-(ii) with a second order Taylor expansion.
\item Assumption \ref{A:GMM2new}-(iv) is assumed.
\end{itemize}
\end{proof}
\section{Checking regularity conditions for a two-arm completely randomized design}\label{CR}
In this section we check the condition $\sigma_{\max}\left(\dmat\right)=O(1)$, $\sigma_{\max}\left(\dtildep\right)=O(1)$, $\sigma_{\max}((\tilde{\dmat}\otimes \tilde{\dmat})\circ \s)=o(n)$ for a two-arm complete randomization with the treatment probability strictly bounded between 0 and 1.
Table \ref{Table:CRFirstOrder} and Table \ref{Table:CRSecondOrder} provide detail calculations for entries in $\dmat$, and $(\tilde{\dmat}\otimes \tilde{\dmat})\circ \s$.

Recall the setup of a two-arm complete randomization: for a sample of size $n$, $n_t$ units are randomly selected into the treatment group, and the rest of the $n_c=n-n_t$ units are selected into the control group. We assume that there exists positive $c$ and $C$ such that $0<c<\frac{n_t}{n}<C<1$ for all large $n$. Define the rescaled demeaning matrix $\A_n=\frac{n}{n-1}\left(\I_{n}-\frac{1}{n}\ones{n}\ones{n}'\right)\in\mathbb{R}^{n\times n}$. We note that the first-order design matrix for the design is 
\begin{equation}
    \dmat=\begin{bmatrix}
    \frac{n_c}{n_t} \A_n & -\A_n\\
  -\A_n &  \frac{n_t}{n_c} \A_n'
\end{bmatrix}\in\mathbb{R}^{2n\times 2n}.
\end{equation}
We note that the variance bound being used is the Neyman bound. 
The standard Neyman bound (e.g. see \cite{imbens2015causal}) matrix can be written as 
\begin{equation}\label{NeymanBound}
\widetilde{\dmat}^{\tgn}=\dmat+\begin{bmatrix}
\A_n&  \A_n\\
\A_n & \A_n
\end{bmatrix}=\begin{bmatrix}
  \frac{n}{n_t} \A_n&  0\\
0&   \frac{n}{n_c} \A_n
\end{bmatrix}\in\mathbb{R}^{2n\times 2n}.
\end{equation}

As shown in Table \ref{Table:CRFirstOrder}, the absolute row sum of the matrix $\dmat$ is either $\frac{n_t}{n_c}+\frac{n_c}{n_t}+1+1=O(1)$ or $\frac{n_c}{n_t}+\frac{n_t}{n_c}+1+1=O(1)$. Thus, $\sigma_{\max}\left(\dmat\right)=O(1)$. We can use the same argument as to establish that $\sigma_{\max}\left(\dtildep^{\tgn}\right)=O(1)$ and we omit for brevity.

\indent To check $\sigma_{\max}((\tilde{\dmat}^\tgn\otimes \tilde{\dmat}^\tgn)\circ \s)=o(n)$, we use Lemma \ref{LemmA:VanillaRootn}. Table \ref{Table:CRSecondOrder} contains information about the slice indexed by $\R_{1i}$, $i\in[n]$. For information about the slice indexed by $\R_{0i}$, $i\in[n]$, one can simply exchange the roles of $n_t$ and $n_c$ in Table \ref{Table:CRSecondOrder}.We see $\sigma_{\max}((\tilde{\dmat}^\tgn\otimes \tilde{\dmat}^\tgn)\circ \s)=O(1)$ using Lemma \ref{LemmA:VanillaRootn}. \\
\indent This calculation can be easily generalized to check similar assumptions for cluster randomization designs with uniform bounded cluster sizes and stratified randomization designs with a finite number of strata. 
\begin{table*}
\caption{Entries for $\dmat$ of a two-arm completely randomized design}\label{Table:CRFirstOrder}
\begin{tabular}{ccccc}
\hline
     Partition  & Entry Type & Count &  Entry Value & $\frac{1}{n}$Count$\scriptstyle \times$Entry Absolute Value   \\
     \hline $\dmat_{11}$   & Diagonal & $\scriptstyle n$ & $\scriptstyle \frac{n_c}{n_t}$   & $\scriptstyle \frac{n_c}{n_t}=  O(1)$
                 \\        & Off-diagonal & $\scriptstyle n(n-1)$ & $\scriptstyle -\frac{n_c}{n_t(n-1)}$   & $\scriptstyle -\frac{n_c}{n_t}=  O(1)$
    \\ \hline $\dmat_{12}$ or $\dmat_{21}$   & Diagonal & $\scriptstyle  2n$ & $\scriptstyle  -1$   & $\scriptstyle  2=O(1)$
                 \\        & Off-diagonal & $\scriptstyle 2n(n-1)$ & $\scriptstyle \frac{1}{(n-1)}$   & $\scriptstyle  2=O(1)$
     \\ \hline $\dmat_{22}$   & Diagonal & $\scriptstyle n$ & $\scriptstyle \frac{n_t}{n_c}$   & $\scriptstyle \frac{n_t}{n_c}=\scriptstyle O(1)$
                 \\        & Off-diagonal & $\scriptstyle n(n-1)$ & $\scriptstyle-\frac{n_t}{n_c(n-1)}$   & $\scriptstyle -\frac{n_t}{n_c}= O(1)$ 
                 \\ \hline
    \end{tabular}
\legend{$n$ is the number of units. $n_t$ is the number of units in the treatment group, and $n_c$ is the number of units in the control group. The matrix $\dmat$ is divided into four block matrices with $\dmat_{ab}$, $a,b=1,2$. The block matrix $\dmat_{ab}$ consists of positions between rows $(a-1)n+1$ and $an$ and columns $(b-1)n+1$ and $bn$. The Entry Type column indicates whether the entries are diagonal or off-diagonal in the $\dmat_{ab}$ matrix. The Count column counts the number of entries of each type. The Entry Value column records the value of the entry of each type.}
\end{table*}

\begin{landscape}
\begin{table*}
\caption{Entries for $(\tilde{\dmat}^N\otimes \tilde{\dmat}^N)\circ \s$ for a two-arm completely randomized design with a fixed $\R_{1i}$} \label{Table:CRSecondOrder}
\begin{tabular}{ccccc}
\hline 
jkl pattern & \makecell[c]{ $\mathbf{COV}(\R_{1i}\R_{1j},\R_{1k}\R_{1l})\frac{\tilde{\dmat}^N_{ij}\tilde{\dmat}^N_{kl}}{\pi_{ij}^{11}\pi_{lk}^{11}}$\\ $(v_{111})$}  & \makecell[c]{$\mathbf{COV}(\R_{1i}\R_{1j},\R_{0k}\R_{0l})\frac{\tilde{\dmat}^N_{ij}\tilde{\dmat}^N_{kl}}{\pi_{ij}^{11}\pi_{lk}^{00}}$\\ $(v_{100})$}  & Count &  O(1)?\\
\hline 
i=j=k=l &  $\frac{n^2n_c}{n_t^3}$ & $-\frac{n^2}{n_tn_c}$ & 1 & Y\\
j$\not=$i, k=l=i & $-\frac{n^2n_c}{n_t^3(n-1)}$ &$\frac{n^2}{(n-1)n_cn_t}$  &n-1 & Y\\
k$\not=$i, j=l=i & $-\frac{n^2n_c}{n_t^3(n-1)}$&$\frac{n^2}{(n-1)n_tn_c}$&n-1 & Y\\
l$\not=$i, j=k=i & $-\frac{n^2n_c}{n_t^3(n-1)}$ &$\frac{n^2}{(n-1)n_tn_c}$&n-1 & Y\\
j$=$i, k$=$l & $-\frac{n_cn^2}{(n-1)n_t^3}$  & $\frac{n^2}{n_tn_c(n-1)}$ & n-1 & Y \\
j$=$k, i$=$l & $\frac{n_c(n+n_t-1)n^2}{n_t(n_t-1)n_c^2(n-1)^2}$ &  -$\frac{n^2}{n_tn_c(n-1)^2}$&  n-1 & Y\\
j$=$l, i$=$k & $\frac{n_c(n+n_t-1)n^2}{n_t(n_t-1)n_c^2(n-1)^2}$ &  -$\frac{n^2}{n_tn_c(n-1)^2}$&  n-1 & Y \\
j$=$i, k$\not=$l$\not=$i &$\frac{2n_cn^2}{n_t^3(n-1)(n-2)}$& $\frac{2n^2}{n_tn_c(n-1)(n-2)}$   &(n-1)(n-2)& Y\\
k=i, j$\not=$l$\not=$i & $\left[\frac{(n_t-2)n(n-1)}{(n-2)n_t(n_t-1)}-1\right]\frac{n^2}{n_t^2(n-1)^2}$  & $-\frac{n^2}{n_cn_t(n-1)^2}$ & (n-1)(n-2)& Y\\
l=i, j$\not=$k$\not=$i & $\left[\frac{(n_t-2)n(n-1)}{(n-2)n_t(n_t-1)}-1\right]\frac{n^2}{n_t^2(n-1)^2}$ & $-\frac{n^2}{n_cn_t(n-1)^2}$ & (n-1)(n-2) & Y\\
i$\not=$j$\not=$k$\not=$l &$\frac{(4n-6)n_t(n_t-1)-(4n_t-6)n(n-1)}{(n-2)(n-3)n_t(n_t-1)}\frac{n^2}{n_t^2(n-1)^2}$ & $\frac{(4n-6)n^2}{(n-2)(n-3)n_tn_c(n-1)^2}$ & (n-1)(n-2)(n-3)& Y\\
\hline 
\end{tabular}
\legend{This table computes the entries of the tensor $(\tilde{\dmat}^N\otimes \tilde{\dmat}^N)\circ \s$ for completely randomized experiments with the Neyman variance bound. WLOG, the $i$th entry is fixed, and indices $j$, $k$, and $l$ vary. The first column lists all possible $(i,j,k,l)$ patterns. The second column computes the case for $\mathbf{COV}(\R_{1i}\R_{1j},\R_{1k}\R_{1l})\frac{\tilde{\dmat}^N_{ij}\tilde{\dmat}^N_{kl}}{\pi_{ij}\pi_{lk}}$. The third column computes the case for $\mathbf{COV}(\R_{1i}\R_{1j},\R_{0k}\R_{0l})\frac{\tilde{\dmat}^N_{ij}\tilde{\dmat}^N_{kl}}{\pi_{ij}\pi_{lk}}$. The fourth column  counts the number of such patterns with $i$ fixed. The last column confirms if the sum ($(|v_{111}|+|v_{100}|)$)$*$Count is of the order $O(1)$. Note that with the Neyman bound, patterns involving $\mathbf{COV}(\R_{1i}\R_{0j},\R_{1k}\R_{0l})$ and $\mathbf{COV}(\R_{1i}\R_{1j},\R_{1k}\R_{0l})$ (the off-diagonal block) are multiplied with zeros, so we do not need to check them here. $\pi_{ij}^{11}=\E[\R_{1i}\R_{1j}]$ and $\pi_{ij}^{00}=\E[\R_{0i}\R_{0j}]$.  }   
\end{table*}
\end{landscape}

\section{Simulation Details}\label{AppendixSimulationDetails}
The dataset of \cite{cai2015social} is retrieved from the Harvard Dataverse \cite{DVN/CXDJM0_2018}.\footnote{The data used in the simulations below are available at \url{https://www.openicpsr.org/openicpsr/project/113593/version/V1/view}.}  We use the 0422allinfoawnet.dta file to extract social network information and 0422survey.dta file to extract pretreatment covariates information on experimental units. We output the .dta files to .csv files and import them in R for data cleaning. We italicize the variable names in the dataset hereafter. 

Each household is an experimental unit associated with an identifier (\textit{id}). Each household may nominate at most five other households as friends (\textit{network\_id}). The pretreatment covariates we use in the simulation are \textit{male}, \textit{age}, \textit{agpop}, \textit{ricearea2010}, \textit{risk\_average}, \textit{disaster\_prob}, and \textit{literacy}. We choose these pretreatment covariates to mimic the specification of Column 6 in Table 2 of the paper \cite{cai2015social}. In addition, we use \textit{village} and \textit{address} variables. \textit{address} indicates the natural village that a household belongs to. \textit{village} indicates the administrative village that a household belongs to. Administrative villages are larger units and consist of natural villages. The experimental design in the paper is a natural-village stratified design: within each natural village, households are stratified according to household sizes (\textit{agpop}) and rice production areas (\textit{ricearea2010}) and are randomly assigned to different treatment arms.
\subsection{Data Construction }
We take on four tasks: 1) defining the population of interest;  2) create strata for random treatment assignments; 3) impute potential outcomes; 4) impute missing pretreatment covariates data.
\subsubsection{Defining the Population of Interest}
We start with 4902 households in the survey.dta dataset (hereafter \textit{survey} dataset) and 4984 households in the 0422allinfoawnet.dta dataset (hereafter \textit{network} dataset). 
\begin{itemize}
    \item In the \textit{survey} dataset, we drop households with missing household sizes (\textit{agpop}) or rice production areas (\textit{ricearea2010}) information. The two variables are used for stratification.
    \item In the \textit{network} dataset, we remove rows with self-nominations and repeated nominations. 
    \item In the \textit{network} dataset, we remove households (both in \textit{id} and \textit{network\_id} columns) with no match from the \textit{survey} dataset. This is because defining strata requires household sizes (\textit{agpop}) and rice production areas (\textit{ricearea2010}) information, which is in the \textit{survey} dataset. 
    \item Our filtered dataset has 4509 households with friendship network information and they form our population of interest. Together with their friends, there are a total of 4806 units that are randomly assigned to different treatment arms.
\end{itemize}
\subsubsection{Strata for random treatment assignments}\label{A:treatment_assignments}
Our randomization procedure assigns households into four different treatment arms: First Round Simple (FRS), First Round Intensive (FRI), Second Round Simple (FRS) and Second Round Intensive (FRI). We consider three classes of experimental designs: 1) a finely stratified natural-village level randomization; 2) a natural-village level randomization; 3) Bernoulli designs. 

For the finely stratified natural-village level randomization, we partition households in each natural villages into four groups based on their household sizes and rice production areas. For each natural village, we calculate the medians of these two variables and classify households into four groups: LL (below median household sizes \& below median rice production areas), LH (below median household sizes \& above median rice production areas), HL (above median household sizes \& below median rice production areas), HH (above median household sizes \& above median rice production areas). In some villages, there are strata with less than four households. For these strata, we merge each with a stratum of the same type from another village. We choose the other village such that they belong to the same connected component of the friendship network as the natural village of the stratum to be merged. Households in each stratum are then completely randomized to four treatment arms. If the number of households are not a multiple of four, there are at most three remainder households. We assign the first remainder household to Second Round Intensive, the second remainder household (if exists) to Second Round Simple and the third household (if exists) to First Round Intensive.

For the natural-village level randomization, we first merge two villages with less than 10 people each with one other village. With each natural village, we randomly assign households to four arms with proportions $\frac{1}{10}$ in FRS, $\frac{1}{10}$ in FRI, $\frac{2}{5}$ in SRS, and $\frac{2}{5}$ in SRI. Practically, we repeat the vector $(4,3,2,1,4,3,4,3,4,3)$ for each village. When the number of households is not a multiple of 10, we start from leftmost of the vector until all remainder households are exhausted.

For Bernoulli designs, we randomly assign households to four treatment arms with probabilities (1/4,1/4,1/4,1/4) (Design C.1), (1/5,1/5,3/10,3/10) (Design C.2), (1/6,1/6,1/3,1/3) (Design C.3), and (1/9,2/9,1/3,1/3) (Design C.4).  
\subsubsection{Imputing potential outcomes for the scenarios \textit{Sim-Impu}}
We impute the potential outcomes using a logit model $y_{ai}=1\{\Pr\left(\beta_a+x_i'\beta_a\right)>\epsilon_i\},\hspace{2pt} \epsilon_i\sim \text{Uniform}(0,1)$. The pretreatment covariates we use are \textit{male}, \textit{age}, \textit{agpop}, \textit{ricearea2010}, \textit{risk\_average}, \textit{disaster\_prob}, and \textit{literacy}. The coefficients are estimated based on the realized data from \cite{cai2015social}. The potential outcomes are kept fixed across simulations.
\subsubsection{Imputing potential outcomes for the scenarios \textit{Sim-Optimal}}
We impute the potential outcomes such that the Opt-GR estimators will have an efficiency gain over the QMLE-GR estimators. Specifically, for each pair of average potential outcomes, e.g. exposure 3 and exposure 4, we take the sum of the first 5 eigenvectors corresponding to the largest  eigenvalue of the matrix $(\I_{2n}-\X(\X'\X)^{-1}\X')'\sigma_{\max}\left(\Ome^c\right)(\I_{2n}-\X(\X'\X)^{-1}\X')-(\I_{2n}-\X(\X'\dmat\X)^{-1}\X')'\sigma_{\max}\left(\Ome^c\right)(\I_{2n}-\X(\X'\dmat\X)^{-1}\X')$, add 0.5 onto each entry and round them to the nearest integer.\footnote{With $c=(-1,1)$, $\frac{1}{n}y'(\I_{2n}-\X(\X'\X)^{-1}\X')'\Ome^c(\I_{2n}-\X(\X'\X)^{-1}\X')y$ is the asymptotic variance of the WLS estimator for  and $\frac{1}{n}y'(\I_{2n}-\X(\X'\Ome^c\X)^{-1}\X')'\Ome^c(\I_{2n}-\X(\X'\Ome^c\X)^{-1}\X'\Ome^c)y$ is the asymptotic variance of the Opt-GR estimator with a linear model. }
\subsection{Missing Pretreatment Covariates Information}
For the 4509 subjects of interest, the pretreatment covariates have a few missing data points but the missingness patterns are not severe ($\leq1\%$ for all variables).  Following the recommendation of \cite{lin2016standard}, we impute the missing values of each covariate column to the overall mean of the corresponding covariate column.
 \subsection{Implementation Details}
 \subsubsection{Calculating the first-order design matrix $\dmat$}
For our simulations, we are comparing the average potential outcomes of 2 arms with at most 4509 experimental units, so the dimension of $\dmat$ is at most 9018-by-9018. To calculate $\dmat$, we compute the covariance matrix $\V(\ones{kn}\R)\in\mathbb{R}^{kn\times kn}$ and the first-order assignment probabilities $\bpi=\E[\R]$ by simulation with the Welford's online algorithm.\footnote{We choose the online algorithm because it uses less computer memory.  } The $\dmat$ is then calculated using the formula $\dmat=\bpiInv\V(\ones{kn}\R)\bpiInv$. The number of the simulation is $10^8$. This is informed by Remark 4.7.2 in \cite{vershynin2018high} with a relative error $0.01$ in terms of the $\viiii{\cdot}_2$ norm for calculating $\V(\ones{kn}\R)\in\mathbb{R}^{kn\times kn}$. With $10^8$ simulations and the smallest benchmark probability being 0.03\footnote{The benchmark probability is the probability that a unit has no friend assigned to the first round with a Bernoulli design with probabilities $(0.25,0.25,0.25,0.25,0.25)$. The probability is $0.5^5\approx 0.03$.}, results of \cite{fattorini} (P275) suggest the element-wise estimation bias of the first order assignment probability is $\frac{(1-0.03)^{10^8}}{0.03}\approx 0$. 
\subsubsection{Estimators}
We use the command lm() and glm() in \cite{Rstats} to estimate the OLS, WLS and Logit models. IPW, OPT-GR (linear), OPT-I GR (linear) and OPT-I GR (logit) estimators have closed-form expressions. The OPT-GR (logit) estimator is the solution to a minimization problem, which is solved by the optim() command in R. Our default algorithm is the Broyden–Fletcher–Goldfarb–Shanno algorithm (BFGS) algorithm.

\section{Simulation Results}\label{AppendixSimulationResults}
This section includes the complete simulation results in Section \ref{Section:Simulation}.
\begin{center}
\begin{table}[ht]
    \caption{ Simulation Results for Comparing Exposures 3 and 4 in Design C.2 and Scenario \textit{Sim-Impute}, $\sqrt{\sigma_{\max}\left(\Ome\right)/n}=0.156$, $n=4509$ }\label{Table:34DesignC2Impute}
\centering
\begin{tabular}{rrrrrrrr}
  \hline
 & IPW &  WLS & LOGIT & \makecell{OPT \\ {\scriptsize(Linear)}} & \makecell{OPT \\ {\scriptsize(Logit)}} & \makecell{Opt-I \\ {\scriptsize(WLS)}} & \makecell{Opt-I \\ {\scriptsize(Logit)}} \\ 
  \hline
$\textrm{Bias}^2\times n$ & 0.00  & 0.00 & 0.06 & 1.57 & 1.54 & 0.17 & 0.18 \\ 
  $\textrm{Variance}\times n$ & 19.77  & 8.95 & 8.89 & 10.67 & 9.23 & 9.28 & 9.16 \\ 
  $\textrm{MSE}\times n$ & 19.77  & 8.95 & 8.89 & 13.14 & 11.60 & 9.31 & 9.20 \\ 
  $\textrm{Est. Var. Bound}^2\times n$ & 19.82  & 8.66 & 8.78 & 10.62 & 9.12 & 9.20 & 9.02 \\ 
  $\textrm{True Asy. Variance}\times n$ & 19.38  & 8.83 & 8.83 & 8.81 & 8.71 & 8.83 & 8.77 \\ 
  $\textrm{95\% Normal CI Coverage}$ & 0.95 & 0.94 & 0.94 & 0.93 & 0.93 & 0.94 & 0.94 \\ 
   \hline
\end{tabular}
\end{table}
\legend{Table \ref{Table:34DesignC2Impute} reports simulation results for comparing exposures 3 and 4 in Design C.2 in the Sim-Impute scenario. n=4509 is the sample size. The number of simulations is 10000. IPW refers to the inverse-probability weighted estimator, LOGIT refers to the QMLE-GR estimator with a logit model, WLS refers to the inverse-probability weighted least squares estimator, OPT (Linear) refers to the Opt-GR estimator with a linear model, OPT (Logit) refers to the Opt-GR estimator with a logit model, Opt-I (WLS) refers to the Opt-I GR estimator with the WLS model, and Opt-I (Logit) refers to the Opt-I GR estimator with the logit model. \\
The row \textit{$\textrm{Bias}^2\times n$} reports the squared biases of the estimators. The row $\textrm{Variance}\times n$ reports the variances of the estimators. The row $\textrm{MSE}\times n$ reports the MSEs of the estimators. The row $\textrm{Est. Var. Bound}^2\times n$ reports the averaged estimates of the variance bound estimators. The row $\textrm{True Asy. Variance}\times n$ reports the theoretical asymptotic variances. The first five rows are normalized by the sample size. The row  \textit{95\% Normal CI Coverage} reports the empirical coverage rates of nominal 95 percent confidence intervals. }
\end{center}
\newpage

\begin{center}
\begin{table}[ht]
    \caption{ Simulation Results for Comparing Exposures 3 and 6 in Design C.2 and Scenario \textit{Sim-Impute}, $\sqrt{\sigma_{\max}\left(\Ome\right)/n}=0.160$, $n=4303$ }\label{Table:36DesignC2Impute}
\centering
\begin{tabular}{rrrrrrrrr}
  \hline
 & IPW &  WLS & LOGIT & \makecell{OPT \\ {\scriptsize(Linear)}} & \makecell{OPT \\ {\scriptsize(Logit)}} & \makecell{Opt-I \\ {\scriptsize(WLS)}} & \makecell{Opt-I \\ {\scriptsize(Logit)}} \\ 
  \hline
$\textrm{Bias}^2\times n$ & 0.00  & 0.03 & 0.01 & 0.07 & 0.10 & 0.00 & 0.00 \\ 
  $\textrm{Variance}\times n$ & 36.77  & 13.84 & 13.60 & 15.29 & 14.79 & 14.20 & 14.23 \\ 
  $\textrm{MSE}\times n$ & 36.77  & 13.88 & 13.61 & 15.35 & 14.90 & 14.20 & 14.23 \\ 
  $\textrm{Est. Var. Bound}^2\times n$ & 38.23  & 13.34 & 13.55 & 15.64 & 14.87 & 14.14 & 14.12 \\ 
  $\textrm{True Asy. Variance}\times n$ & 36.25 & 13.48 & 13.48 & 13.47 & 13.44 & 13.47 & 13.47 \\ 
  $\textrm{95\% Normal CI Coverage}$ & 0.95 & 0.94 & 0.94 & 0.95 & 0.95 & 0.95 & 0.95 \\ 
   \hline
\end{tabular}
\end{table}
\legend{Table \ref{Table:36DesignC2Impute} reports simulation results for comparing exposures 3 and 6 in Design C.2 in the Sim-Impute scenario. n=4303 is the sample size. The number of simulations is 10000. IPW refers to the inverse-probability weighted estimator, LOGIT refers to the QMLE-GR estimator with a logit model, WLS refers to the inverse-probability weighted least squares estimator, OPT (Linear) refers to the Opt-GR estimator with a linear model, OPT (Logit) refers to the Opt-GR estimator with a logit model, Opt-I (WLS) refers to the Opt-I GR estimator with the WLS model, and Opt-I (Logit) refers to the Opt-I GR estimator with the logit model. \\
The row \textit{$\textrm{Bias}^2\times n$} reports the squared biases of the estimators. The row $\textrm{Variance}\times n$ reports the variances of the estimators. The row $\textrm{MSE}\times n$ reports the MSEs of the estimators. The row $\textrm{Est. Var. Bound}^2\times n$ reports the averaged estimates of the variance bound estimators. The row $\textrm{True Asy. Variance}\times n$ reports the theoretical asymptotic variances. The first five rows are normalized by the sample size. The row  \textit{95\% Normal CI Coverage} reports the empirical coverage rates of nominal 95 percent confidence intervals. }
\end{center}

\newpage

\begin{center}
\begin{table}[ht]
    \caption{ Simulation Results for Comparing Exposures 3 and 4 in Design C.4 and Scenario \textit{Sim-Impute}, $\sqrt{\sigma_{\max}\left(\Ome\right)/n}=0.115$, $n=4509$ }\label{Table:34DesignC4Impute}
\centering
\begin{tabular}{rrrrrrrrr}
  \hline
 & IPW &  WLS & LOGIT & \makecell{OPT \\ {\scriptsize(Linear)}} & \makecell{OPT \\ {\scriptsize(Logit)}} & \makecell{Opt-I \\ {\scriptsize(WLS)}} & \makecell{Opt-I \\ {\scriptsize(Logit)}} \\ 
  \hline
$\textrm{Bias}^2\times n$ & 0.00 & 0.00 & 0.00 & 0.01 & 0.01 & 0.00 & 0.00 \\ 
  $\textrm{Variance}\times n$ & 17.68  & 7.81 & 7.70 & 7.98 & 7.68 & 7.86 & 7.83 \\ 
  $\textrm{MSE}\times n$ & 17.68  & 7.81 & 7.70 & 7.99 & 7.69 & 7.86 & 7.83 \\ 
  $\textrm{Est. Var. Bound}^2\times n$ & 17.74 & 7.66 & 7.71 & 8.14 & 7.85 & 7.88 & 7.79 \\ 
  $\textrm{True Asy. Variance}\times n$ & 17.30  & 7.68 & 7.68 & 7.68 & 7.60 & 7.68  & 7.62 \\ 
  $\textrm{95\% Normal CI Coverage}$ & 0.95 &  0.95 & 0.95 & 0.95 & 0.95 & 0.95 & 0.95 \\ 
   \hline
\end{tabular}
\end{table}
\legend{Table \ref{Table:34DesignC4Impute} reports simulation results for comparing exposures 3 and 4 in Design C.4 in the Sim-Impute scenario. n=4509 is the sample size. The number of simulations is 10000. IPW refers to the inverse-probability weighted estimator, LOGIT refers to the QMLE-GR estimator with a logit model, WLS refers to the inverse-probability weighted least squares estimator, OPT (Linear) refers to the Opt-GR estimator with a linear model, OPT (Logit) refers to the Opt-GR estimator with a logit model, Opt-I (WLS) refers to the Opt-I GR estimator with the WLS model, and Opt-I (Logit) refers to the Opt-I GR estimator with the logit model. \\
The row \textit{$\textrm{Bias}^2\times n$} reports the squared biases of the estimators. The row $\textrm{Variance}\times n$ reports the variances of the estimators. The row $\textrm{MSE}\times n$ reports the MSEs of the estimators. The row $\textrm{Est. Var. Bound}^2\times n$ reports the averaged estimates of the variance bound estimators. The row $\textrm{True Asy. Variance}\times n$ reports the theoretical asymptotic variances. The first five rows are normalized by the sample size. The row  \textit{95\% Normal CI Coverage} reports the empirical coverage rates of nominal 95 percent confidence intervals. }
\end{center}
\newpage

\begin{center}
\begin{table}[ht]
    \caption{ Simulation Results for Comparing Exposures 3 and 6 in Design C.4 and Scenario \textit{Sim-Impute}, $\sqrt{\sigma_{\max}\left(\Ome\right)/n}=0.106$, $n=4303$ }\label{Table:36DesignC4Impute}
\centering
\begin{tabular}{rrrrrrrrr}
  \hline
 & IPW  & WLS & LOGIT & \makecell{OPT \\ {\scriptsize(Linear)}} & \makecell{OPT \\ {\scriptsize(Logit)}} & \makecell{Opt-I \\ {\scriptsize(WLS)}} & \makecell{Opt-I \\ {\scriptsize(Logit)}} \\ 
  \hline
$\textrm{Bias}^2\times n$ & 0.00  & 0.01 & 0.00 & 0.36 & 0.41 & 0.01 & 0.01 \\ 
  $\textrm{Variance}\times n$ & 25.96  & 9.12 & 8.99 & 9.49 & 9.46 & 9.22 & 9.23 \\ 
  $\textrm{MSE}\times n$ & 25.96  & 9.13 & 8.99 & 9.85 & 9.86 & 9.23 & 9.24 \\ 
  $\textrm{Est. Var. Bound}^2\times n$ & 27.80  & 9.27 & 9.37 & 10.07 & 9.89 & 9.66 & 9.65 \\ 
  $\textrm{True Asy. Variance}\times n$ & 25.81  & 8.98 & 8.98 & 8.97 & 8.96 & 8.98 & 8.99 \\ 
  $\textrm{95\% Normal CI Coverage}$ & 0.96  & 0.95 & 0.95 & 0.95 & 0.95 & 0.95 & 0.95 \\ 
   \hline
\end{tabular}
\end{table}
\legend{Table \ref{Table:36DesignC4Impute} reports simulation results for comparing exposures 3 and 6 in Design C.4 in the Sim-Impute scenario. n=4303 is the sample size. The number of simulations is 10000. IPW refers to the inverse-probability weighted estimator, LOGIT refers to the QMLE-GR estimator with a logit model, WLS refers to the inverse-probability weighted least squares estimator, OPT (Linear) refers to the Opt-GR estimator with a linear model, OPT (Logit) refers to the Opt-GR estimator with a logit model, Opt-I (WLS) refers to the Opt-I GR estimator with the WLS model, and Opt-I (Logit) refers to the Opt-I GR estimator with the logit model. \\
The row \textit{$\textrm{Bias}^2\times n$} reports the squared biases of the estimators. The row $\textrm{Variance}\times n$ reports the variances of the estimators. The row $\textrm{MSE}\times n$ reports the MSEs of the estimators. The row $\textrm{Est. Var. Bound}^2\times n$ reports the averaged estimates of the variance bound estimators. The row $\textrm{True Asy. Variance}\times n$ reports the theoretical asymptotic variances. The first five rows are normalized by the sample size. The row  \textit{95\% Normal CI Coverage} reports the empirical coverage rates of nominal 95 percent confidence intervals. }
\end{center}

\newpage
\begin{center}
\begin{table}[ht]
    \caption{ Simulation Results for Comparing Exposures 3 and 4 in Design C.2 and Scenario \textit{Sim-Optimal}, $\sqrt{\sigma_{\max}\left(\Ome\right)/n}=0.156$, $n=4509$ }\label{Table:34DesignC2Optimal}
\centering
\begin{tabular}{rrrrrrrr}
  \hline
 & IPW & WLS & LOGIT & \makecell{OPT \\ {\scriptsize(Linear)}} & \makecell{OPT \\ {\scriptsize(Logit)}} & \makecell{Opt-I \\ {\scriptsize(WLS)}} & \makecell{Opt-I \\ {\scriptsize(Logit)}} \\ 
  \hline
$\textrm{Bias}^2\times n$ & 0.00 & 0.06 & 0.02 & 1.45 & 1.51 & 0.20 & 0.15 \\ 
  $\textrm{Variance}\times n$ & 27.42 & 13.27 & 12.61 & 11.38 & 11.39 & 11.19 & 11.46 \\ 
  $\textrm{MSE}\times n$ & 27.43 & 13.33 & 12.64 & 12.82 & 12.90 & 11.39 & 11.61 \\ 
  $\textrm{Est. Var. Bound}^2\times n$ & 28.82 & 12.86 & 12.70 & 11.63 & 11.53 & 11.60 & 11.71 \\ 
  $\textrm{True Asy. Variance}\times n$ & 27.54 & 12.77 & 12.33 & 10.64 & 10.61 & 10.67 & 10.80 \\ 
  $\textrm{95\% Normal CI Coverage}$ & 0.95 & 0.94 & 0.94 & 0.94 & 0.93 & 0.95 & 0.95 \\ 
   \hline
\end{tabular}
\end{table}
\legend{Table \ref{Table:34DesignC2Optimal} reports simulation results for comparing exposures 3 and 4 in Design C.2 in the Sim-Optimal scenario. n=4509 is the sample size. The number of simulations is 10000. IPW refers to the inverse-probability weighted estimator, LOGIT refers to the QMLE-GR estimator with a logit model, WLS refers to the inverse-probability weighted least squares estimator, OPT (Linear) refers to the Opt-GR estimator with a linear model, OPT (Logit) refers to the Opt-GR estimator with a logit model, Opt-I (WLS) refers to the Opt-I GR estimator with the WLS model, and Opt-I (Logit) refers to the Opt-I GR estimator with the logit model. \\
The row \textit{$\textrm{Bias}^2\times n$} reports the squared biases of the estimators. The row $\textrm{Variance}\times n$ reports the variances of the estimators. The row $\textrm{MSE}\times n$ reports the MSEs of the estimators. The row $\textrm{Est. Var. Bound}^2\times n$ reports the averaged estimates of the variance bound estimators. The row $\textrm{True Asy. Variance}\times n$ reports the theoretical asymptotic variances. The first five rows are normalized by the sample size. The row  \textit{95\% Normal CI Coverage} reports the empirical coverage rates of nominal 95 percent confidence intervals. }
\end{center}

\newpage
\begin{center}
\begin{table}[ht]
    \caption{ Simulation Results for Comparing Exposures 3 and 6 in Design C.2 and Scenario \textit{Sim-Optimal}, $\sqrt{\sigma_{\max}\left(\Ome\right)/n}=0.160$, $n=4303$ }\label{Table:36DesignC2Optimal}
\centering
\begin{tabular}{rrrrrrrr}
  \hline
 & IPW & WLS & LOGIT & \makecell{OPT \\ {\scriptsize(Linear)}} & \makecell{OPT \\ {\scriptsize(Logit)}} & \makecell{Opt-I \\ {\scriptsize(WLS)}} & \makecell{Opt-I \\ {\scriptsize(Logit)}} \\ 
  \hline
$\textrm{Bias}^2\times n$ & 0.00 & 0.01 & 0.03 & 0.01 & 0.05 & 0.01 & 0.01 \\ 
  $\textrm{Variance}\times n$ & 40.40 & 17.66 & 16.86 & 15.65 & 15.71 & 15.47 & 15.56 \\ 
  $\textrm{MSE}\times n$ & 40.40 & 17.67 & 16.89 & 15.66 & 15.76 & 15.48 & 15.57 \\ 
  $\textrm{Est. Var. Bound}^2\times n$ & 43.99 & 17.84 & 17.76 & 16.63 & 16.46 & 16.71 & 16.72 \\ 
  $\textrm{True Asy. Variance}\times n$ & 40.50 & 17.59 & 17.23 & 14.83 & 14.81 & 14.91 & 14.85 \\ 
  $\textrm{95\% Normal CI Coverage}$ & 0.96 & 0.95 & 0.95 & 0.95 & 0.95 & 0.95 & 0.95 \\ 
   \hline
\end{tabular}
\end{table}
\legend{Table \ref{Table:36DesignC2Optimal} reports simulation results for comparing exposures 3 and 6 in Design C.2 in the Sim-Optimal scenario. n=4303 is the sample size. The number of simulations is 10000. IPW refers to the inverse-probability weighted estimator, LOGIT refers to the QMLE-GR estimator with a logit model, WLS refers to the inverse-probability weighted least squares estimator, OPT (Linear) refers to the Opt-GR estimator with a linear model, OPT (Logit) refers to the Opt-GR estimator with a logit model, Opt-I (WLS) refers to the Opt-I GR estimator with the WLS model, and Opt-I (Logit) refers to the Opt-I GR estimator with the logit model. \\
The row \textit{$\textrm{Bias}^2\times n$} reports the squared biases of the estimators. The row $\textrm{Variance}\times n$ reports the variances of the estimators. The row $\textrm{MSE}\times n$ reports the MSEs of the estimators. The row $\textrm{Est. Var. Bound}^2\times n$ reports the averaged estimates of the variance bound estimators. The row $\textrm{True Asy. Variance}\times n$ reports the theoretical asymptotic variances. The first five rows are normalized by the sample size. The row  \textit{95\% Normal CI Coverage} reports the empirical coverage rates of nominal 95 percent confidence intervals. }
\end{center}

\newpage

\begin{center}
\begin{table}[ht]
    \caption{ Simulation Results for Comparing Exposures 3 and 4 in Design C.4 and Scenario \textit{Sim-Optimal}, $\sqrt{\sigma_{\max}\left(\Ome\right)/n}=0.115$, $n=4509$ }\label{Table:34DesignC4Optimal}
\centering
\begin{tabular}{rrrrrrrr}
  \hline
 & IPW & WLS & LOGIT & \makecell{OPT \\ {\scriptsize(Linear)}} & \makecell{OPT \\ {\scriptsize(Logit)}} & \makecell{Opt-I \\ {\scriptsize(WLS)}} & \makecell{Opt-I \\ {\scriptsize(Logit)}} \\ 
  \hline
$\textrm{Bias}^2\times n$ & 0.00 & 0.02 & 0.02 & 0.00 & 0.00 & 0.00 & 0.00 \\ 
  $\textrm{Variance}\times n$ & 28.36 & 10.87 & 10.61 & 9.85 & 9.94 & 10.25 & 10.24 \\ 
  $\textrm{MSE}\times n$ & 28.36 & 10.89 & 10.63 & 9.86 & 9.94 & 10.26 & 10.24 \\ 
  $\textrm{Est. Var. Bound}^2\times n$ & 28.84 & 10.76 & 10.70 & 10.22 & 10.04 & 10.78 & 10.77 \\ 
  $\textrm{True Asy. Variance}\times n$ & 27.84 & 10.48 & 10.36 & 9.50 & 9.33 & 9.80 & 9.54 \\ 
  $\textrm{95\% Normal CI Coverage}$ & 0.95 & 0.95 & 0.95 & 0.95 & 0.95 & 0.95 & 0.95 \\ 
   \hline
\end{tabular}
\end{table}
\legend{Table \ref{Table:34DesignC4Optimal} reports simulation results for comparing exposures 3 and 4 in Design C.4 in the Sim-Optimal scenario. n=4509 is the sample size. The number of simulations is 10000. IPW refers to the inverse-probability weighted estimator, LOGIT refers to the QMLE-GR estimator with a logit model, WLS refers to the inverse-probability weighted least squares estimator, OPT (Linear) refers to the Opt-GR estimator with a linear model, OPT (Logit) refers to the Opt-GR estimator with a logit model, Opt-I (WLS) refers to the Opt-I GR estimator with the WLS model, and Opt-I (Logit) refers to the Opt-I GR estimator with the logit model. \\
The row \textit{$\textrm{Bias}^2\times n$} reports the squared biases of the estimators. The row $\textrm{Variance}\times n$ reports the variances of the estimators. The row $\textrm{MSE}\times n$ reports the MSEs of the estimators. The row $\textrm{Est. Var. Bound}^2\times n$ reports the averaged estimates of the variance bound estimators. The row $\textrm{True Asy. Variance}\times n$ reports the theoretical asymptotic variances. The first five rows are normalized by the sample size. The row  \textit{95\% Normal CI Coverage} reports the empirical coverage rates of nominal 95 percent confidence intervals. }
\end{center}

\newpage
\begin{center}
\begin{table}[ht]
    \caption{ Simulation Results for Comparing Exposures 3 and 6 in Design C.4 and Scenario \textit{Sim-Optimal}, $\sqrt{\sigma_{\max}\left(\Ome\right)/n}=0.106$, $n=4303$ }\label{Table:36DesignC4Optimal}
\centering
\begin{tabular}{rrrrrrrr}
  \hline
 & IPW & WLS & LOGIT & \makecell{OPT \\ {\scriptsize(Linear)}} & \makecell{OPT \\ {\scriptsize(Logit)}} & \makecell{Opt-I \\ {\scriptsize(WLS)}} & \makecell{Opt-I \\ {\scriptsize(Logit)}} \\ 
  \hline
$\textrm{Bias}^2\times n$ & 0.00 & 0.03 & 0.02 & 0.48 & 0.53 & 0.00 & 0.01 \\ 
  $\textrm{Variance}\times n$ & 25.33 & 12.52 & 12.31 & 11.32 & 11.13 & 11.29 & 11.29 \\ 
  $\textrm{MSE}\times n$ & 25.33 & 12.55 & 12.33 & 11.81 & 11.65 & 11.29 & 11.29 \\ 
  $\textrm{Est. Var. Bound}^2\times n$ & 28.95 & 13.56 & 13.63 & 12.45 & 12.21 & 12.76 & 12.76 \\ 
  $\textrm{True Asy. Variance}\times n$ & 25.37 & 12.68 & 12.68 & 10.87 & 10.84 & 11.36 & 11.14 \\ 
  $\textrm{95\% Normal CI Coverage}$ & 0.96 & 0.95 & 0.96 & 0.95 & 0.95 & 0.96 & 0.96 \\ 
   \hline
\end{tabular}
\end{table}
\legend{Table \ref{Table:36DesignC4Optimal} reports simulation results for comparing exposures 3 and 6 in Design C.4 in the Sim-Optimal scenario. n=4303 is the sample size. The number of simulations is 10000. IPW refers to the inverse-probability weighted estimator, LOGIT refers to the QMLE-GR estimator with a logit model, WLS refers to the inverse-probability weighted least squares estimator, OPT (Linear) refers to the Opt-GR estimator with a linear model, OPT (Logit) refers to the Opt-GR estimator with a logit model, Opt-I (WLS) refers to the Opt-I GR estimator with the WLS model, and Opt-I (Logit) refers to the Opt-I GR estimator with the logit model. \\
The row \textit{$\textrm{Bias}^2\times n$} reports the squared biases of the estimators. The row $\textrm{Variance}\times n$ reports the variances of the estimators. The row $\textrm{MSE}\times n$ reports the MSEs of the estimators. The row $\textrm{Est. Var. Bound}^2\times n$ reports the averaged estimates of the variance bound estimators. The row $\textrm{True Asy. Variance}\times n$ reports the theoretical asymptotic variances. The first five rows are normalized by the sample size. The row  \textit{95\% Normal CI Coverage} reports the empirical coverage rates of nominal 95 percent confidence intervals. }
\end{center}

\section{Mathematical Objects, Operations and Quantities}\label{Mathematical Objects, Operators and Quantites}
We define mathematical objects, operations and quantities used in the paper. 
\begin{enumerate}

    \item We define \textit{real tensors}. The definitions are from \cite{qi2017tensor}.  A real tensor $\mathbf{A}= \left [{a_{ \scriptscriptstyle i_{\scriptscriptstyle 1}...i_{\scriptscriptstyle m}}}\right]\in\mathbb{R}^{n_1\times ...\times n_m} $ is a multi-array of entries, where $i_j=1,...,n_j$ for $j=1,...,m$. The positive integer $m$ is called the order of the tensor. When $n=n_1=...=n_m$, $\mathbf{A}$ is called an $m$th order $n$-dimensional tensor. The set of real $m$th order $n$-dimensional tensors is denoted as $\mathbf{T}_{m,n}$.  From this definition, a real matrix is a real tensor with order $m=2$ and $\mathbb{R}^{n\times n}=\mathbf{T}_{2,n}$. In the paper we shall use a 4th order n-dimensional real tensor to describe the fourth moments of the experimental designs. 
\item We define the tensor \textit{Hadamard product}. Let $\mathbf{A}= \left [{a_{ \scriptscriptstyle i_{\scriptscriptstyle 1}...i_{\scriptscriptstyle m}}}\right]\in\mathbb{R}^{n_1\times ...\times n_m} $ and $\mathbf{B}= \left [{b_{ \scriptscriptstyle i_{\scriptscriptstyle 1}...i_{\scriptscriptstyle m}}}\right]\in\mathbb{R}^{n_1\times ...\times n_m}$. The Hadamard product of two tensors is the result of their entrywise multiplications:
\begin{equation*}
        A\circ B = [c_{ \scriptscriptstyle i_{\scriptscriptstyle 1}...i_{\scriptscriptstyle m}}]=[{a_{ \scriptscriptstyle i_{\scriptscriptstyle 1}...i_{\scriptscriptstyle m}}} {b_{ \scriptscriptstyle i_{\scriptscriptstyle 1}...i_{\scriptscriptstyle m}}}]\in\mathbb{R}^{n_1\times ...\times n_m}
\end{equation*}
With a slight abuse of terminology, we denote the tensor \textit{Hadamard division} by $\backslash$:
\begin{equation*}
        A\backslash B = \left[d_{ \scriptscriptstyle i_{\scriptscriptstyle 1}...i_{\scriptscriptstyle m}}\right]=\left[\frac{{a_{ \scriptscriptstyle i_{\scriptscriptstyle 1}...i_{\scriptscriptstyle m}}}}{{b_{ \scriptscriptstyle i_{\scriptscriptstyle 1}...i_{\scriptscriptstyle m}}}} \right]\in\mathbb{R}^{n_1\times ...\times n_m},
\end{equation*}
with the rule $\frac{0}{0}=0$. The tensors used in this paper are designed to avoid the problem of dividing a nonzero number by 0.\\
\item Besides the usual matrix operations, we shall also define the \textit{tensor product} of two matrices. We use $\otimes$ to denote the \textit{tensor product} of two matrices, which results in an order four tensor. For any two matrices $A=[a_{ij}]\in\mathbb{R}^{n_1\times n_2}$ and $B=[b_{ij}]\in\mathbb{R}^{n_3\times n_4}$:
\begin{equation*}
    A\otimes B = [c_{ijkl}]=[a_{ij}b_{kl}]\in \mathbb{R}^{n_1\times n_2\times n_3\times n_4}
\end{equation*}
The tensor product operation can also be defined on higher order tensors but we use it only for matrices in the paper.
\item We define \textit{norms} of matrices. For a real matrix $\mathbf{A}=\left[a_{ij}\right]\in\mathbb{R}^{n_1\times n_2}$. $\|\mathbf{A}\|_2$ denotes the Frobenius norm of the matrix $\mathbf{A}$ where  $\|\mathbf{A}\|_2=\sqrt{\sum_{i=1}^{n_1}\sum_{j=1}^{n_2}a_{ij}^2}$. $\|\mathbf{A}\|_1$ denotes the $l_1$ vector norm of the matrix where  $\|\mathbf{A}\|_1=\sum_{i=1}^{n_1}\sum_{j=1}^{n_2}|a_{ij}|$. We will also use $\|\mathbf{A}\|_4^4=\sum_{i=1}^{n_1}\sum_{j=1}^{n_2}|a_{ij}|^4$. $\sigma_{\max}\left(\mathbf{A}\right)$ denotes the spectral norm where $\sigma_{\max}\left(\mathbf{A}\right)_2=\max\{\sqrt{\lambda}, \lambda \text{ is an eigenvalue of } \mathbf{A}'\mathbf{A}\}$. We will also use $\viiii{\mathbf{A}}_2$ to denote the spectral norm.  $\viiii{\mathbf{A}}_1$ denotes the $l_1$ -induced matrix norm, where $\viiii{\mathbf{A}}_1=\max_{j\in[n]}\{\sum_{i=1}^n |a_{ij}|\} $. For a tensor $\mathbf{A}$ we use $\|\mathbf{A}\|_1$ to denote the sum of the absolute values of the tensor entries. We use the standard notation for vector norms, for example, see Section 5.2 in \cite{horn2012matrix}.
\item For a tensor $\mathbf{A}\in\mathbf{T}^{m,n}$ with an even $m$, we use the symbol $\sigma_{\max}(\mathbf{A})$ to denote the optimal value of the following optimization problem\footnote{This quantity is defined in \cite{lim2005singular}.}:
\begin{align*}
    &\max_{\{v_i\}_{i=1}^m\in\mathbb{R}^n}\mathbf{A}(v_1,...,v_m)\\ &\text{   subject to  }    \sum_{i=1}^n v_i^m=1 \text{ for all $i=1,...,m$}        
\end{align*}
\item For symbols, $\I_{k}$ denotes the identity matrix of dimension $k\times k$, and $\0_{k}$ a zero matrix of dimension $k\times k$, and $\0_{k\times p}$ a zero matrix of dimension $k\times p$. $\ones{k}$ denotes a column k-vector of 1's. $\mathbf{A}^{+}$ denotes the unique Moore-Penrose inverse of the matrix $\mathbf{A}$. $\diag()$ maps a length-n vector to an n-by-n diagonal matrix. We denote the matrix positive-semidefinite partial ordering by $\succeq$: $\mathbf{A}\succeq \mathbf{B}$ if and only if $\mathbf{A}-\mathbf{B}$ is a positive semidefinite matrix. We use $\nabla$ to denote the total differentiation operator. For a function $f:\mathcal{X}\subset \mathbb{R}^k\to\mathbb{R}$, $\nabla_x f$ denotes the gradient function of $f$ (if it exists), $\nabla_{xx'}f$ denotes the Hessian function of $f$ and so on. We use the partial derivative notation $\frac{\partial}{\partial x}f$ to denote the partial derivative $f$ with respect to a particular argument $x$. For a set in $\Theta\subset \mathbb{R}^s$, we use the notation $\operatorname{Bd}(\Theta)$ to denote its boundary with respect to the standard topology of a Euclidean space.
\item For probabilistic convergence, a sequence of random variables $M_n=o_p(1)$ if $\lim_{n\to\infty}\mathbf{P}_n(|M_n|>\epsilon)=0$ for any positive $\epsilon$, and $M_n=O_p(1)$ if for each $\epsilon>0$ there exists a constant $K\geq 0$ and a constant $N\geq 0$ such that $\mathbf{P}_n(|M_n|\geq K)<\epsilon$ for all $n\geq N$. A sequence of random variables $M_n=o_p(a_n)$ if $\frac{M_n}{a_n}=o_p(1)$ and $M_n=O_p(a_n)$ if $\frac{M_n}{a_n}=O_p(1)$. A vector or matrix with fixed dimensions is $o_p(a_n)$ and $O_p(a_n)$ if each entry is $o_p(a_n)$ and $O_p(a_n)$, respectively. For two deterministic sequences $a_n$ and $b_n$, we denote $a_n=\Theta(b_n)$ if there exsits positive $c$ and $C$ and a N such that $cb_n\leq a_n\leq Cb_n$ for $n\geq N$.\\
\end{enumerate}

\end{appendix}

\end{document}